\documentclass[12pt,final,a4paper,twoside]{iopart}
\expandafter\let\csname equation*\endcsname\relax
\expandafter\let\csname endequation*\endcsname\relax
\usepackage{amsmath}
\usepackage[T1]{fontenc}
\usepackage{textcomp}
\usepackage{txfonts}
\usepackage{tensor}
\usepackage{cancel}
\usepackage{verbatim}
\usepackage{mathrsfs}
\usepackage{mathbbol}
\newcommand{\HCd}{\mathcal{H}}
\newcommand{\FCd}{\mathcal{\tilde{F}}}
\newcommand{\LCd}{\mathcal{L}}
\newcommand{\RB}{\mathbb{R}}

\newcommand{\dd}{\mathrm{d}}
\newcommand{\bgamma}{\pmb{\gamma}}
\newcommand{\bsigma}{\pmb{\sigma}}
\newcommand{\psibar}{\bar{\psi}}

\newcommand{\pibar}{\bar{\pi}}

\newcommand{\btau}{\pmb{\tau}}
\newcommand{\bEins}{\pmb{\Eins}}
\newcommand{\onehalf}{{\textstyle\frac{1}{2}}}
\newcommand{\quarter}{{\textstyle\frac{1}{4}}}

\newcommand{\pfrac}[2]{\frac{\partial{#1}}{\partial{#2}}}
\newcommand{\ppfrac}[3]{\frac{\partial^{2}{#1}}{\partial{#2}\partial{#3}}}
\newcommand{\pppfrac}[4]{\frac{\partial^{3}{#1}}{\partial{#2}\partial{#3}\partial{#4}}}
\newcommand{\afffias}{Frankfurt Institute for Advanced Studies (FIAS), Ruth-Moufang-Strasse~1,\\ ~~60438 Frankfurt am Main, Germany}
\newcommand{\affjwg}{Goethe-Universit\"at, Max-von-Laue-Strasse~1, 60438~Frankfurt am Main, Germany}
\newcommand{\affgsi}{GSI Helmholtzzentrum f\"ur Schwerionenforschung GmbH, Planckstrasse~1,\\ ~~64291 Darmstadt, Germany}
\newtheorem{proposition}{Proposition}
\newtheorem{proof}{Proof}
\newtheorem{corollary}{Corollary}
\begin{document}
\title[Generic Theory of Geometrodynamics]{Generic Theory of Geometrodynamics\\
from Noether's theorem for the $\mathrm{Diff}(M)$ symmetry group}
\author{J~Struckmeier$^{1,2,3}$, D~Vasak$^1$, and J~Kirsch$^1$}
\address{$^1$\afffias}\address{$^2$\affjwg}\address{$^3$\affgsi\\[\medskipamount] \today}
\ead{struckmeier@fias.uni-frankfurt.de}

\bigskip\noindent
\begin{indented}
\item[]\emph{Dedicated to the memory of Prof.~Dr.~Walter Greiner, our teacher, mentor, and friend}

\medskip\noindent
In this contribution, we present the canonical transformation formalism in the realm of classical field theory,
where spacetime is treated as a dynamical quantity, and apply it to formulate the gauge theory of gravitation.
In this respect, it generalizes the Extended Hamilton-Lagrange-Jacobi formalism of relativistic point dynamics.
Walter was very much interested in this formalism and therefore added several chapters on the matter to the
second edition of his textbook ``Classical Mechanics''~\cite{greiner03}.
To quote Walter from the Preface to the Second Edition: ``It may come as a surprise that even for the time-honored
subject such as Classical Mechanics in the formulation of Lagrange and Hamilton, new aspects may emerge.''
We are sure, Walter would have loved the following elaboration.
\end{indented}
\begin{abstract}
We work out the most general theory for the interaction of spacetime geometry and matter
fields---commonly referred to as geometrodynamics---for spin-$0$ and spin-$1$ particles.
Actually, we present a Hamilton-Lagrange-Noether formulation of the gauge theory of gravitation.
It is based on the minimum set of postulates to be introduced, namely (i) the action principle and
(ii) the form-invariance of the action under the (local) diffeomorphism group.
The second postulate thus implements the \emph{Principle of General Relativity},
also referred to as the \emph{Principle of General Covariance}.
According to Noether's theorem, this \emph{physical symmetry} gives rise to a conserved Noether current,
from which the complete set of theories compatible with both postulates can be deduced.
This finally results in a new generic Einstein-type equation, which can be interpreted as an
energy-momentum balance equation emerging from the Lagrangian $\LCd_{R}$ for the source-free
dynamics of gravitation and the energy-momentum tensor of the source system $\LCd_{0}$.
Provided that the system has no other symmetries---such as SU$(N)$---the \emph{canonical}
energy-momentum tensor turns out to be the correct source term of gravitation.
For the case of massive spin-$1$ particles, this entails an increased weighting of the kinetic energy
over the mass as the source of gravity, compared to the \emph{metric} energy momentum tensor,
which constitutes the source of gravity in Einstein's General Relativity.
We furthermore confirm that a massive vector field necessarily acts as a source for torsion of spacetime.
From the viewpoint of our generic Einstein-type equation, Einstein's General Relativity constitutes
the particular case for scalar and massless vector particle fields, and the Hilbert Lagrangian
$\LCd_{R,\mathrm{H}}$ as the model for the source-free dynamics of gravitation.
\end{abstract}
\vspace*{-4mm}
\section{Introduction}
The covariant Hamiltonian formalism in the realm of classical field theories was recently
shown to allow generalized canonical transformations which also include arbitrary active
diffeomorphisms within the underlying spacetime manifold $M$~\cite{struckmeier17a}.
With this framework at hand, it is now possible to isolate the complete set of theories
which are based on the action principle and the Principle of
General Relativity, i.e., the condition for a system to be form-invariant under the diffeomorphism group.
This approach naturally leads to a Palatini formulation, where the metric
and the affine connection \emph{a priori} represent independent dynamical quantities.
Systems complying with these principles are thus required to have the $\mathrm{Diff}(M)$ group as an intrinsic symmetry group.
As known from Noether's theorem, each symmetry is associated with a pertaining conserved Noether current $j_{\mathrm{N}}^{\,\mu}$.
The main favor of the Noether approach is that the obtained condition for a conserved Noether current
directly leads to the respective field equations, which describe the \emph{coupling} of metric
and connection to the given source fields of gravitation.

With our actual Noether approach, it is possible to derive the most general form of an Einstein-type equation
for a closed system of scalar and vector fields in dynamical spacetime
satisfying the Principle of General Relativity---including a spacetime with torsion and
without the restriction to a covariantly conserved metric.
We thereby isolate the complete set of possible theories of geometrodynamics for scalar and vector matter
and derive a new form of a generic Einstein-type equation.

Similar to all gauge theories, the Noether approach to geometrodynamics provides the \emph{coupling}
of the fields of the given system to the spacetime geometry, but does \emph{not} fix the Hamiltonian
resp.\ the Lagrangian $\LCd_{R}$ describing the dynamics of the ``free'' (uncoupled)
gravitational field, hence the gravitational field dynamics in classical vacuum.
The Hilbert Lagrangian $\LCd_{R,\mathrm{H}}$---which entails the Einstein tensor of standard General
Relativity---is the simplest example.
Based on analogy with other classical field theories, Einstein himself already proposed
a Lagrangian $\LCd_{R}$ quadratic in the Riemann tensor~\cite{einstein18}, which will be
discussed here as an amendment to the conventional Hilbert Lagrangian $\LCd_{R,\mathrm{H}}$.
Remarkably, the field equation emerging from this Lagrangian is equally satisfied by the
Schwarzschild and even the Kerr metric in the case of classical vacuum~\cite{kehm17}.

The source term of gravity is shown to be given by the \emph{canonical} energy-momentum tensor,
provided that the given system has no additional symmetries---such as a $\mathrm{SU}(N)$
symmetry---besides the $\mathrm{Diff}(M)$ symmetry.
This entails an increased weight of the kinetic energy
over the mass in their roles as sources of gravity.
Also, a massive vector field is shown to \emph{necessarily} induce a torsion of spacetime---which
is in perfect agreement with previous works of Hehl et al.\ (see, for instance, \cite{hehl76}).

We review in section~\ref{sec:cantra} the formalism of canonical transformations in the
covariant Hamiltonian description of classical field theories.
After having formulated Noether's theorem~\cite{noether18} in the realm of covariant Hamiltonian field theory in section~\ref{sec:noether},
we proceed in section~\ref{sec:finite} with the canonical transformation representation of finite $\mathrm{Diff}(M)$ symmetry transformations.

In order to work out the conserved Noether current for the $\mathrm{Diff}(M)$ symmetry trans\-formation,
the finite transformation is reformulated in section~\ref{sec:infini} as the pertaining infinitesimal transformation.
The detailed discussion of the conserved Noether current then follows in section~\ref{sec:discussion}.
It will be shown that the \emph{zero-energy principle}---hence a vanishing energy-momentum tensor
of the total system of source fields and dynamic spacetime---emerges as a direct consequence.

The most general Einstein-type equation of geometrodynamics is presented in section~\ref{sec:lagrange}.
This equation is shown to be equivalent to the ``consistency equation'' of Ref.~\cite{struckmeier17a}
by means of an \emph{identity}, which holds for scalar-valued functions of arbitrary tensors and the metric.
This identity also provides the correlation of the metric (Hilbert) and the canonical energy-momentum tensors of a given Lagrangian system.
We discuss in section~\ref{sec:sample-lag} possible Lagrangians $\LCd_{R}$ for the dynamics of the free gravi\-ta\-tional field
and set up a generalized field equation quadratic and linear in the Riemann-Cartan tensor.
Finally, we discuss the correlation of the spin part of the energy-momentum tensor of the source system
with the torsion of spacetime, described by the then skew-symmetric part of the Ricci tensor.
\section{Canonical transformations under a dynamic spacetime\label{sec:cantra}}
The requirement of form-invariance of the action functional for a real scalar field $\phi$ and a vector
field $a_{\mu}$---in conjunction with their respective canonical conjugate fields $\pi^{\mu}$ and
$p^{\nu\mu}$---in a local inertial frame, hence under a \emph{static} spacetime, is formulated as
\begin{align}
\delta\int_{\Omega_{y}}\left(\pi^{\alpha}\pfrac{\phi}{y^{\alpha}}+
p^{\,\beta\alpha}\pfrac{a_{\beta}}{y^{\alpha}}-\HCd-
\pfrac{\mathcal{F}_{1}^{\alpha}}{y^{\alpha}}\right)\dd^{4}y\stackrel{!}{=}
\delta\int_{\Omega_{y}}\left(\Pi^{\alpha}\pfrac{\Phi}{y^{\alpha}}+
P^{\,\beta\alpha}\pfrac{A_{\beta}}{y^{\alpha}}-\HCd^{\prime}\right)\dd^{4}y.
\label{eq:varprinzip0}
\end{align}
For the transition to non-inertial frames, the volume form $\dd^{4}y$ is expressed
in terms of $\dd^{4}x$ and $\dd^{4}X$ by means of the respective Jacobians
\begin{equation*}
\dd^{4}y=\left|\pfrac{y}{x}\right|\dd^{4}x,\qquad\dd^{4}y=\left|\pfrac{y}{X}\right|\dd^{4}X
\quad\Rightarrow\quad\dd^{4}X=\left|\pfrac{X}{x}\right|\dd^{4}x.
\end{equation*}
We assume the local inertial system $\big(y^{0},\ldots,y^{3}\big)$ to be endowed
by the Minkowski metric $\eta_{jk}=\mathrm{diag}\left(1,-1,-1,-1\right)$.
The metrics $g_{\mu\nu}(x)$ and $G_{\mu\nu}(X)$ in the respective non-inertial frames
$\big(x^{0},\ldots,x^{3}\big)$ and $\big(X^{0},\ldots,X^{3}\big)$ are then
\begin{equation*}
g_{\mu\nu}=\eta_{jk}\pfrac{y^{j}}{x^{\mu}}\pfrac{y^{k}}{x^{\mu}},\qquad
G_{\mu\nu}=\eta_{jk}\pfrac{y^{j}}{X^{\mu}}\pfrac{y^{k}}{X^{\mu}},
\end{equation*}
and hence their determinants, owing to $\det\eta_{jk}=-1$, are
\begin{equation*}
g\equiv\det g_{\mu\nu}=-\left|\pfrac{y}{x}\right|^{2},\qquad
G\equiv\det G_{\mu\nu}=-\left|\pfrac{y}{X}\right|^{2}.
\end{equation*}
The volume form $\dd^{4}y$ thus writes in terms of the determinants of the respective metric
\begin{equation}\label{eq:trans-volumeform}
\dd^{4}y=\sqrt{-g}\,\dd^{4}x,\qquad\dd^{4}y=\sqrt{-G}\,\dd^{4}X.
\end{equation}
The requirement of form-invariance of the action functional for real scalar and vector fields under
transitions $x^{\mu}\mapsto X^{\mu}$ is formulated as
\begin{equation}\label{eq:varprinzip1}
\delta\int_{\Omega_{x}}\left(\tilde{\pi}^{\alpha}\pfrac{\phi}{x^{\alpha}}+
\tilde{p}^{\,\beta\alpha}\pfrac{a_{\beta}}{x^{\alpha}}-\tilde{\HCd}-
\pfrac{\FCd_{1}^{\alpha}}{x^{\alpha}}\right)\dd^{4}x\stackrel{!}{=}
\delta\int_{\Omega_{X}}\left(\tilde{\Pi}^{\alpha}\pfrac{\Phi}{X^{\alpha}}+
\tilde{P}^{\,\beta\alpha}\pfrac{A_{\beta}}{X^{\alpha}}-\tilde{\HCd}^{\prime}\right)\dd^{4}X,
\end{equation}
where the factors $\sqrt{-g}$ and $\sqrt{-G}$ are absorbed into the respective momentum fields,
thereby converting them into relative tensors of weight $w=1$, i.e., into \emph{tensor densities}:
\begin{equation*}
\tilde{\pi}^{\beta}=\pi^{\beta}\sqrt{-g},\qquad\tilde{\Pi}^{\beta}=\Pi^{\beta}\sqrt{-G},
\end{equation*}
and similarly for all other momenta.
The Hamiltonians are thus converted into scalar densities.

As the action integral is to be varied, Eq.~(\ref{eq:varprinzip1})
implies that the \emph{integrands} may differ by the
divergence of a vector density $\FCd_{1}^{\mu}(x)$ whose variation vanishes on the
boundary $\partial\Omega_{x}$ of the integration region $\Omega_{x}$ within spacetime
\begin{equation}\label{eq:surface-term}
\delta\int_{\Omega_{x}}\pfrac{\FCd_{1}^{\alpha}}{x^{\alpha}}\,\dd^{4}x=
\delta\oint_{\partial\Omega_{x}}\FCd_{1}^{\alpha}\,\dd S_{\alpha}\stackrel{!}{=}0.
\end{equation}
The addition of a term $\partial\FCd_{1}^{\alpha}/\partial x^{\alpha}$ to the integrand
which can be converted into a surface integral---commonly referred to briefly as a
\emph{surface term}---thus does not modify the variation of the action integral.
This means that the integrand is only determined up to the divergence
of the functions $\FCd_{1}^{\mu}(\Phi,\phi,A,a,x)$.
With the transformation rule of the volume form from Eq.~(\ref{eq:trans-volumeform}),
and $\FCd_{1}^{\mu}$ to be taken at $x$, the integrand condition obtained from
Eq.~(\ref{eq:varprinzip1}) for an extended canonical transformation thus writes
\begin{align}
&\quad\tilde{\pi}^{\beta}\pfrac{\phi}{x^{\beta}}+\tilde{p}^{\,\alpha\beta}\pfrac{a_{\alpha}}{x^{\beta}}-\tilde{\HCd}-
\left(\tilde{\Pi}^{\beta}\pfrac{\Phi}{X^{\beta}}+\tilde{P}^{\,\alpha\beta}\pfrac{A_{\alpha}}{X^{\beta}}-
\tilde{\HCd}^{\prime}\right)\left|\pfrac{X}{x}\right|\label{eq:integrand-condition}\\
&=\pfrac{\FCd_{1}^{\beta}}{\phi}\pfrac{\phi}{x^{\beta}}+
\pfrac{\FCd_{1}^{\alpha}}{\Phi}\pfrac{X^{\beta}}{x^{\alpha}}\pfrac{\Phi}{X^{\beta}}+
\pfrac{\FCd_{1}^{\beta}}{a_{\alpha}}\pfrac{a_{\alpha}}{x^{\beta}}+
\pfrac{\FCd_{1}^{\xi}}{A_{\alpha}}\pfrac{X^{\beta}}{x^{\xi}}\pfrac{A_{\alpha}}{X^{\beta}}+
\left.\pfrac{\FCd_{1}^{\alpha}}{x^{\alpha}}\right|_{\text{expl}}.\nonumber
\end{align}
Comparing the coefficients yields the transformation rules
\begin{subequations}\label{eq:f1-ext}
\begin{align}
\tilde{\pi}^{\mu}(x)&=\pfrac{\FCd_{1}^{\mu}}{\phi}&
\tilde{\Pi}^{\mu}(X)&=-\pfrac{\FCd_{1}^{\beta}}{\Phi}\pfrac{X^{\mu}}{x^{\beta}}\left|\pfrac{x}{X}\right|\\
\tilde{p}^{\,\nu\mu}(x)&=\pfrac{\FCd_{1}^{\mu}}{a_{\nu}}&
\tilde{P}^{\,\nu\mu}(X)&=-\pfrac{\FCd_{1}^{\beta}}{A_{\nu}}\pfrac{X^{\mu}}{x^{\beta}}\left|\pfrac{x}{X}\right|\\
\left.\tilde{\HCd}^{\prime}\right|_{X}&=\left(\left.\tilde{\HCd}\right|_{x}+
\left.\pfrac{\FCd_{1}^{\alpha}}{x^{\alpha}}\right|_{\text{expl}}\right)\left|\pfrac{x}{X}\right|&
\Rightarrow\left.\tilde{\HCd}^{\prime}\right|_{x}-\left.\tilde{\HCd}\right|_{x}&=\hphantom{-}\!\!
\left.\pfrac{\FCd_{1}^{\alpha}}{x^{\alpha}}\right|_{\text{expl}}.
\end{align}
\end{subequations}
The generating function $\FCd_{1}^{\mu}(\Phi,\phi,A,a,x)$ may be Legendre-transformed into
an equivalent generating function $\FCd_{2}^{\mu}(\tilde{\Pi},\phi,\tilde{P},a,x)$ according to
\begin{equation}\label{eq:extgen-f2}
\FCd_{2}^{\mu}=\FCd_{1}^{\mu}-\Phi\pfrac{\FCd_{1}^{\mu}}{\Phi}-A_{\alpha}\pfrac{\FCd_{1}^{\mu}}{A_{\alpha}}
=\FCd_{1}^{\mu}+\left(\Phi\,\tilde{\Pi}^{\beta}+
A_{\alpha}\,\tilde{P}^{\alpha\beta}\right)\pfrac{x^{\mu}}{X^{\beta}}\left|\pfrac{X}{x}\right|.
\end{equation}
In order to derive the divergence of $\FCd_{2}^{\mu}(x)$, we
make use of the identity~\cite{struckmeier17a} for the right-hand side factor.
Thus
\begin{equation}\label{eq:f2-deri}
\pfrac{\FCd_{1}^{\alpha}}{x^{\alpha}}=\pfrac{\FCd_{2}^{\alpha}}{x^{\alpha}}-
\left|\pfrac{X}{x}\right|\pfrac{}{X^{\beta}}\left(\Phi\,\tilde{\Pi}^{\beta}+
A_{\alpha}\,\tilde{P}^{\alpha\beta}\right).
\end{equation}
Inserting Eq.~(\ref{eq:f2-deri}) into the integrand condition~(\ref{eq:integrand-condition}),
we encounter the modified integrand condition for a generating function of type $\FCd_{2}^{\mu}$,
to be taken at the spacetime event $x$
\begin{align}
&\quad\tilde{\pi}^{\beta}\pfrac{\phi}{x^{\beta}}+\tilde{p}^{\,\alpha\beta}\pfrac{a_{\alpha}}{x^{\beta}}-\tilde{\HCd}+
\left(\Phi\delta_{\xi}^{\beta}\pfrac{\tilde{\Pi}^{\xi}}{X^{\beta}}+
A_{\alpha}\delta_{\xi}^{\beta}\pfrac{\tilde{P}^{\,\alpha\xi}}{X^{\beta}}+
\tilde{\HCd}^{\prime}\right)\left|\pfrac{X}{x}\right|\nonumber\\
&=\pfrac{\FCd_{2}^{\beta}}{\phi}\pfrac{\phi}{x^{\beta}}+\pfrac{\FCd_{2}^{\beta}}{a_{\alpha}}\pfrac{a_{\alpha}}{x^{\beta}}+
\pfrac{\FCd_{2}^{\eta}}{\tilde{\Pi}^{\xi}}\pfrac{X^{\beta}}{x^{\eta}}\pfrac{\tilde{\Pi}^{\xi}}{X^{\beta}}+
\pfrac{\FCd_{2}^{\eta}}{\tilde{P}^{\alpha\xi}}\pfrac{X^{\beta}}{x^{\eta}}\pfrac{\tilde{P}^{\alpha\xi}}{X^{\beta}}+
\left.\pfrac{\FCd_{2}^{\alpha}}{x^{\alpha}}\right|_{\text{expl}},\label{eq:integrand-condition2}
\end{align}
and hence the transformation rules
\begin{subequations}\label{eq:f2-ext}
\begin{align}
\tilde{\pi}^{\mu}(x)&=\pfrac{\FCd_{2}^{\mu}}{\phi}&\delta_{\nu}^{\mu}\Phi(X)&=
\pfrac{\FCd_{2}^{\eta}}{\tilde{\Pi}^{\nu}}\pfrac{X^{\mu}}{x^{\eta}}\left|\pfrac{x}{X}\right|\\
\tilde{p}^{\,\nu\mu}(x)&=\pfrac{\FCd_{2}^{\mu}}{a_{\nu}}&
\delta_{\nu}^{\mu}A_{\alpha}(X)&=\pfrac{\FCd_{2}^{\eta}}{\tilde{P}^{\alpha\nu}}\pfrac{X^{\mu}}{x^{\eta}}\left|\pfrac{x}{X}\right|\\
\left.\tilde{\HCd}^{\prime}\right|_{X}&=\left(\left.\tilde{\HCd}\right|_{x}+
\left.\pfrac{\FCd_{2}^{\alpha}}{x^{\alpha}}\right|_{\text{expl}}\right)\left|\pfrac{x}{X}\right|&
\Rightarrow\left.\tilde{\HCd}^{\prime}\right|_{x}-\left.\tilde{\HCd}\right|_{x}&=
\left.\pfrac{\FCd_{2}^{\alpha}}{x^{\alpha}}\right|_{\text{expl}}.
\end{align}
\end{subequations}
In any case, the integrands of the action integrals~(\ref{eq:varprinzip1}) must be world
scalars in order to keep their form under general spacetime transformations.
This finally ensures that the canonical field equations emerge as \emph{tensor equations}.
Furthermore, the generating function $\FCd_{1}^{\mu}(\Phi,\phi,A,a,x)$ may also be Legendre-transformed into
an equivalent generating function of type $\FCd_{3}^{\mu}(\tilde{\pi},\Phi,\tilde{p},A,x)$ according to
\begin{equation}\label{eq:extgen-f3}
\FCd_{3}^{\mu}=\FCd_{1}^{\mu}-\phi\pfrac{\FCd_{1}^{\mu}}{\phi}-a_{\alpha}\pfrac{\FCd_{1}^{\mu}}{a_{\alpha}}
=\FCd_{1}^{\mu}-\phi\,\tilde{\pi}^{\mu}-a_{\alpha}\,\tilde{p}^{\alpha\mu},
\end{equation}
hence
\begin{equation}\label{eq:f3-deri}
\pfrac{\FCd_{1}^{\alpha}}{x^{\alpha}}=\pfrac{\FCd_{3}^{\alpha}}{x^{\alpha}}+
\pfrac{}{x^{\beta}}\left(\phi\,\tilde{\pi}^{\beta}+a_{\alpha}\,\tilde{p}^{\alpha\beta}\right).
\end{equation}
Inserting Eq.~(\ref{eq:f3-deri}) into the integrand condition~(\ref{eq:integrand-condition}),
we encounter the modified integrand condition for a generating function of type $\FCd_{3}^{\mu}(x)$,
to be taken at the spacetime event $x$
\begin{align}
&\quad-\phi\,\delta_{\xi}^{\beta}\pfrac{\tilde{\pi}^{\xi}}{x^{\beta}}-
a_{\alpha}\delta_{\xi}^{\beta}\pfrac{\tilde{p}^{\,\alpha\xi}}{x^{\beta}}-\tilde{\HCd}-
\left(\tilde{\Pi}^{\beta}\pfrac{\Phi}{X^{\beta}}+\tilde{P}^{\,\alpha\beta}\pfrac{A_{\alpha}}{X^{\beta}}-
\tilde{\HCd}^{\prime}\right)\left|\pfrac{X}{x}\right|\nonumber\\
&=\pfrac{\FCd_{3}^{\eta}}{\Phi}\pfrac{X^{\beta}}{x^{\eta}}\pfrac{\Phi}{X^{\beta}}+
\pfrac{\FCd_{3}^{\eta}}{A_{\alpha}}\pfrac{X^{\beta}}{x^{\eta}}\pfrac{A_{\alpha}}{X^{\beta}}+
\pfrac{\FCd_{3}^{\beta}}{\tilde{\pi}^{\xi}}\pfrac{\tilde{\pi}^{\xi}}{x^{\beta}}+
\pfrac{\FCd_{3}^{\beta}}{\tilde{p}^{\alpha\xi}}\pfrac{\tilde{p}^{\alpha\xi}}{x^{\beta}}+
\left.\pfrac{\FCd_{3}^{\alpha}}{x^{\alpha}}\right|_{\text{expl}},\label{eq:integrand-condition3a}
\end{align}
and hence the transformation rules by comparing the coefficients
\begin{subequations}\label{eq:F3-rules}
\begin{align}
\tilde{\Pi}^{\mu}(X)&=-\pfrac{\FCd_{3}^{\kappa}}{\Phi}\pfrac{X^{\mu}}{x^{\kappa}}\left|\pfrac{x}{X}\right|&
\delta_{\nu}^{\mu}\,\phi(x)&=-\pfrac{\FCd_{3}^{\mu}}{\tilde{\pi}^{\nu}}\\
\tilde{P}^{\,\nu\mu}(X)&=-\pfrac{\FCd_{3}^{\kappa}}{A_{\nu}}\pfrac{X^{\mu}}{x^{\kappa}}\left|\pfrac{x}{X}\right|&
\delta_{\nu}^{\mu}\,a_{\alpha}(x)&=-\pfrac{\FCd_{3}^{\mu}}{\tilde{p}^{\alpha\nu}}\\
\left.\tilde{\HCd}^{\prime}\right|_{X}&=\left(\left.\tilde{\HCd}\right|_{x}+
\left.\pfrac{\FCd_{3}^{\alpha}}{x^{\alpha}}\right|_{\text{expl}}\right)\left|\pfrac{x}{X}\right|&
\Rightarrow\left.\tilde{\HCd}^{\prime}\right|_{x}-\left.\tilde{\HCd}\right|_{x}&=\hphantom{-}\!\!
\left.\pfrac{\FCd_{3}^{\alpha}}{x^{\alpha}}\right|_{\text{expl}}\label{eq:ham-rule-gen}.
\end{align}
In order to allow the description of a \emph{dynamic} spacetime, the Hamiltonians are presumed
to depend in addition on the metric $g_{\alpha\lambda}$ and the---in general---non-symmetric connection
$\gamma\indices{^{\alpha}_{\xi\eta}}$, in conjunction with their respective conjugates,
$\tilde{k}\indices{^{\alpha\lambda\beta}}$ and $\tilde{q}\indices{_{\alpha}^{\xi\eta\beta}}$.
Compared to Eq.~(\ref{eq:integrand-condition3a}), we thus encounter an extended integrand condition
for a generating function of type $\FCd_{3}^{\mu}(x)$,
\begin{align}
&\quad-\phi\,\delta_{\xi}^{\beta}\pfrac{\tilde{\pi}^{\xi}}{x^{\beta}}
-a_{\alpha}\delta_{\xi}^{\beta}\pfrac{\tilde{p}^{\,\alpha\xi}}{x^{\beta}}
-g_{\alpha\lambda}\delta_{\xi}^{\beta}\pfrac{\tilde{k}^{\,\alpha\lambda\xi}}{x^{\beta}}
-\gamma\indices{^{\alpha}_{\lambda\eta}}\delta_{\xi}^{\beta}\pfrac{\tilde{q}\indices{_{\alpha}^{\,\lambda\eta\xi}}}{x^{\beta}}
-\tilde{\HCd}\nonumber\\
&\quad-\left(\tilde{\Pi}^{\beta}\pfrac{\Phi}{X^{\beta}}+\tilde{P}^{\,\alpha\beta}\pfrac{A_{\alpha}}{X^{\beta}}
+\tilde{K}\indices{^{\alpha\lambda\beta}}\pfrac{G\indices{_{\alpha\lambda}}}{X^{\beta}}
+\tilde{Q}\indices{_{\alpha}^{\xi\eta\beta}}\pfrac{\Gamma\indices{^{\alpha}_{\xi\eta}}}{X^{\beta}}
-\tilde{\HCd}^{\prime}\right)\left|\pfrac{X}{x}\right|\nonumber\\
&=\pfrac{\FCd_{3}^{\eta}}{\Phi}\pfrac{X^{\beta}}{x^{\eta}}\pfrac{\Phi}{X^{\beta}}
+\pfrac{\FCd_{3}^{\eta}}{A_{\alpha}}\pfrac{X^{\beta}}{x^{\eta}}\pfrac{A_{\alpha}}{X^{\beta}}
+\pfrac{\FCd_{3}^{\eta}}{G_{\alpha\lambda}}\pfrac{X^{\beta}}{x^{\eta}}\pfrac{G_{\alpha\lambda}}{X^{\beta}}
+\pfrac{\FCd_{3}^{\kappa}}{\Gamma\indices{^{\alpha}_{\xi\eta}}}\pfrac{X^{\beta}}{x^{\kappa}}\pfrac{\Gamma\indices{^{\alpha}_{\xi\eta}}}{X^{\beta}}\nonumber\\
&\quad+\pfrac{\FCd_{3}^{\beta}}{\tilde{\pi}^{\xi}}\pfrac{\tilde{\pi}^{\xi}}{x^{\beta}}
+\pfrac{\FCd_{3}^{\beta}}{\tilde{p}^{\alpha\xi}}\pfrac{\tilde{p}^{\alpha\xi}}{x^{\beta}}
+\pfrac{\FCd_{3}^{\beta}}{\tilde{k}^{\,\alpha\lambda\xi}}\pfrac{\tilde{k}^{\,\alpha\lambda\xi}}{x^{\beta}}
+\pfrac{\FCd_{3}^{\beta}}{\tilde{q}\indices{_{\alpha}^{\,\lambda\eta\xi}}}\pfrac{\tilde{q}\indices{_{\alpha}^{\,\lambda\eta\xi}}}{x^{\beta}}
+\left.\pfrac{\FCd_{3}^{\beta}}{x^{\beta}}\right|_{\text{expl}},\label{eq:integrand-condition3}
\end{align}
and hence the additional transformation rules for the metric tensor $g_{\alpha\beta}$ and for the connection $\gamma\indices{^{\eta}_{\alpha\beta}}$,
together with their respective conjugates, $\tilde{k}^{\alpha\beta\mu}$ and $\tilde{q}\indices{_{\eta}^{\alpha\beta\mu}}$:
\begin{align}
\tilde{K}^{\alpha\beta\mu}(X)&=-\pfrac{\FCd_{3}^{\kappa}}{G_{\alpha\beta}}\pfrac{X^{\mu}}{x^{\kappa}}\!\left|\pfrac{x}{X}\right|&
\delta_{\nu}^{\mu}\,g_{\alpha\beta}(x)&=-\pfrac{\FCd_{3}^{\mu}}{\tilde{k}^{\alpha\beta\nu}}\\
\tilde{Q}\indices{_{\eta}^{\alpha\beta\mu}}(X)&=-\pfrac{\FCd_{3}^{\kappa}}{\Gamma\indices{^{\eta}_{\alpha\beta}}}
\pfrac{X^{\mu}}{x^{\kappa}}\!\left|\pfrac{x}{X}\right|&
\delta_{\nu}^{\mu}\,\gamma\indices{^{\eta}_{\alpha\beta}}(x)&=
-\pfrac{\FCd_{3}^{\mu}}{\tilde{q}\indices{_{\eta}^{\alpha\beta\nu}}}.
\end{align}
\end{subequations}
\section{Generalized Noether Theorem\label{sec:noether}}
In this section, we set up a generating function of type $\FCd_{3}^{\mu}$
to define a general \emph{infinitesimal} canonical transformation.
Specifically, for a sample system of a scalar field $\phi$ and a vector field $a_{\mu}$
in a metric-affine space, the infinitesimal transformation rules are derived from the generating function
\begin{equation}\label{eq:F3-infini}
\left.\FCd_{3}^{\mu}\right|_{x}=\left.\left(-\tilde{\pi}^{\,\mu}\Phi-
\tilde{p}^{\alpha\mu}\,A_{\alpha}-\tilde{k}^{\xi\lambda\mu}\,G_{\xi\lambda}-
\tilde{q}\indices{_{\eta}^{\alpha\beta\mu}}\Gamma\indices{^{\eta}_{\alpha\beta}}-
\epsilon\,\tilde{j}_{\mathrm{N}}^{\,\mu}\right)\right|_{x},
\end{equation}
wherein $\tilde{j}_{\mathrm{N}}^{\,\mu}=\tilde{j}_{\mathrm{N}}^{\,\mu}\left(\tilde{\pi},\phi,\tilde{p},a,\tilde{k},g,\tilde{q},\gamma,x\right)$.
For $\epsilon=0$, Eq.~(\ref{eq:F3-infini}) thus generates the identity transformation for all dynamical quantities the autonomous
on which the Hamiltonian $\tilde{\HCd}=\tilde{\HCd}\left(\tilde{\pi},\phi,\tilde{p},a,\tilde{k},g,\tilde{q},\gamma\right)$ depends on.
All contributions of the general transformation rules~(\ref{eq:F3-rules})
that are associated with a non-identical mapping of fields and spacetime
are encoded in the particular expression for $\tilde{j}_{\mathrm{N}}^{\,\mu}(x)$.
The transformation rules~(\ref{eq:F3-rules}) now read:
\begin{subequations}\label{eq:rules-infini3}
\begin{align}
\delta_{\nu}^{\mu}\delta\phi&=-\epsilon\pfrac{\tilde{j}_{\mathrm{N}}^{\,\mu}}{\tilde{\pi}^{\nu}},&
\delta\tilde{\pi}^{\mu}&=\epsilon\pfrac{\tilde{j}_{\mathrm{N}}^{\,\mu}}{\phi},&
\delta_{\nu}^{\mu}\delta a_{\alpha}&=-\epsilon\pfrac{\tilde{j}_{\mathrm{N}}^{\,\mu}}{\tilde{p}^{\alpha\nu}},&
\delta\tilde{p}^{\alpha\mu}&=\epsilon\pfrac{\tilde{j}_{\mathrm{N}}^{\,\mu}}{a_{\alpha}}\\
\delta_{\nu}^{\mu}\delta g_{\lambda\xi}&=-\epsilon\pfrac{\tilde{j}_{\mathrm{N}}^{\,\mu}}{\tilde{k}^{\lambda\xi\nu}},&
\delta\tilde{k}^{\lambda\xi\mu}&=\epsilon\pfrac{\tilde{j}_{\mathrm{N}}^{\,\mu}}{g_{\lambda\xi}},&
\delta_{\nu}^{\mu}\delta\gamma\indices{^{\eta}_{\alpha\beta}}&=
-\epsilon\pfrac{\tilde{j}_{\mathrm{N}}^{\,\mu}}{\tilde{q}\indices{_{\eta}^{\alpha\beta\nu}}},&
\delta\tilde{q}\indices{_{\eta}^{\alpha\beta\mu}}&=\epsilon\pfrac{\tilde{j}_{\mathrm{N}}^{\,\mu}}{\gamma\indices{^{\eta}_{\alpha\beta}}}
\end{align}
and
\begin{equation}
\left.\delta\tilde{\HCd}\right|_{\mathrm{CT}}\equiv\left.\tilde{\HCd}^{\prime}\right|_{x}-\left.\tilde{\HCd}\right|_{x}
=\left.\pfrac{\FCd_{3}^{\alpha}}{x^{\alpha}}\right|_{\text{expl}}
=-\epsilon\left.\pfrac{\tilde{j}_{\mathrm{N}}^{\,\alpha}}{x^{\alpha}}\right|_{\text{expl}}.
\end{equation}
\end{subequations}
The explicit representation of the infinitesimal transformation rules~(\ref{eq:rules-infini3}) for the
generating function $\FCd_{3}^{\mu}$ of a diffeomorphism will be presented below.

On the other hand, for a closed system, where the Hamiltonian does not explicitly depend on $x$,
the variation of the Hamiltonian emerging from the variation of the fields follows as
\begin{equation*}
\delta\tilde{\HCd}=
\pfrac{\tilde{\HCd}}{\phi}\delta\phi+\pfrac{\tilde{\HCd}}{\tilde{\pi}^{\alpha}}\delta\tilde{\pi}^{\alpha}+
\pfrac{\tilde{\HCd}}{a_{\alpha}}\delta a_{\alpha}+\pfrac{\tilde{\HCd}}{\tilde{p}^{\alpha\beta}}\delta\tilde{p}^{\alpha\beta}+
\pfrac{\tilde{\HCd}}{g_{\lambda\xi}}\delta g_{\lambda\xi}+\pfrac{\tilde{\HCd}}{\tilde{k}^{\lambda\xi\beta}}\delta\tilde{k}^{\lambda\xi\beta}+
\pfrac{\tilde{\HCd}}{\gamma\indices{^{\eta}_{\lambda\xi}}}\delta\gamma\indices{^{\eta}_{\lambda\xi}}+
\pfrac{\tilde{\HCd}}{\tilde{q}\indices{_{\eta}^{\lambda\xi\beta}}}\delta\tilde{q}\indices{_{\eta}^{\lambda\xi\beta}}\!.
\end{equation*}
Inserting the covariant canonical field equations~\cite{struckmeier08},
\begin{subequations}\label{eq:canonical_field_equations}
\begin{align}
-\pfrac{\tilde{\pi}^{\beta}}{x^{\beta}}&=\pfrac{\tilde{\HCd}}{\phi}&
\pfrac{\phi}{x^{\alpha}}&=\pfrac{\tilde{\HCd}}{\tilde{\pi}^{\alpha}}&
-\pfrac{\tilde{p}^{\alpha\beta}}{x^{\beta}}&=\pfrac{\tilde{\HCd}}{a_{\alpha}}&
\pfrac{a_{\alpha}}{x^{\beta}}&=\pfrac{\tilde{\HCd}}{\tilde{p}^{\alpha\beta}}\\
-\pfrac{\tilde{k}^{\lambda\xi\beta}}{x^{\beta}}&=\pfrac{\tilde{\HCd}}{g_{\lambda\xi}}&
\pfrac{g_{\lambda\xi}}{x^{\beta}}&=\pfrac{\tilde{\HCd}}{\tilde{k}^{\lambda\xi\beta}}&
-\pfrac{\tilde{q}\indices{_{\eta}^{\lambda\xi\beta}}}{x^{\beta}}&=\pfrac{\tilde{\HCd}}{\gamma\indices{^{\eta}_{\lambda\xi}}}&
\pfrac{\gamma\indices{^{\eta}_{\lambda\xi}}}{x^{\beta}}&=\pfrac{\tilde{\HCd}}{\tilde{q}\indices{_{\eta}^{\lambda\xi\beta}}},
\end{align}
\end{subequations}
the variation of $\tilde{\HCd}$ is expressed along the system's spacetime evolution
\begin{align}
\delta\tilde{\HCd}&=-\pfrac{\tilde{\pi}^{\beta}}{x^{\alpha}}\delta_{\beta}^{\alpha}\,\delta\phi+
\pfrac{\phi}{x^{\alpha}}\delta\tilde{\pi}^{\alpha}-\pfrac{\tilde{p}^{\alpha\beta}}{x^{\xi}}
\delta_{\beta}^{\xi}\,\delta a_{\alpha}+\pfrac{a_{\alpha}}{x^{\beta}}\delta\tilde{p}^{\alpha\beta}\nonumber\\
&\quad\,-\pfrac{\tilde{k}^{\lambda\xi\beta}}{x^{\alpha}}\delta_{\beta}^{\alpha}\,\delta g_{\lambda\xi}
+\pfrac{g_{\lambda\xi}}{x^{\beta}}\delta\tilde{k}^{\lambda\xi\beta}
-\pfrac{\tilde{q}\indices{_{\eta}^{\lambda\xi\beta}}}{x^{\alpha}}\delta_{\beta}^{\alpha}\,\delta\gamma\indices{^{\eta}_{\lambda\xi}}
+\pfrac{\gamma\indices{^{\eta}_{\lambda\xi}}}{x^{\beta}}\delta\tilde{q}\indices{_{\eta}^{\lambda\xi\beta}}.
\label{eq:delta-H}
\end{align}
With the transformation rules~(\ref{eq:rules-infini3}), this writes in terms of the derivatives of $\tilde{j}_{\mathrm{N}}^{\,\mu}$
\begin{align*}
\delta\tilde{\HCd}&=\epsilon\pfrac{\tilde{\pi}^{\beta}}{x^{\alpha}}\pfrac{\tilde{j}_{\mathrm{N}}^{\,\alpha}}{\tilde{\pi}^{\beta}}+
\epsilon\pfrac{\phi}{x^{\alpha}}\pfrac{\tilde{j}_{\mathrm{N}}^{\,\alpha}}{\phi}+
\epsilon\pfrac{\tilde{p}^{\alpha\beta}}{x^{\xi}}\pfrac{\tilde{j}_{\mathrm{N}}^{\,\xi}}{\tilde{p}^{\alpha\beta}}+
\epsilon\pfrac{a_{\alpha}}{x^{\beta}}\pfrac{\tilde{j}_{\mathrm{N}}^{\,\beta}}{a_{\alpha}}\\
&\quad+\epsilon\pfrac{\tilde{k}^{\lambda\xi\beta}}{x^{\alpha}}\pfrac{\tilde{j}_{\mathrm{N}}^{\,\alpha}}{\tilde{k}^{\lambda\xi\beta}}
+\epsilon\pfrac{g_{\lambda\xi}}{x^{\beta}}\pfrac{\tilde{j}_{\mathrm{N}}^{\,\beta}}{g_{\lambda\xi}}
+\epsilon\pfrac{\tilde{q}\indices{_{\eta}^{\lambda\xi\beta}}}{x^{\alpha}}
\pfrac{\tilde{j}_{\mathrm{N}}^{\,\alpha}}{\tilde{q}\indices{_{\eta}^{\lambda\xi\beta}}}
+\epsilon\pfrac{\gamma\indices{^{\eta}_{\lambda\xi}}}{x^{\beta}}\pfrac{\tilde{j}_{\mathrm{N}}^{\,\beta}}{\gamma\indices{^{\eta}_{\lambda\xi}}}\\
&=\epsilon\pfrac{\tilde{j}_{\mathrm{N}}^{\,\alpha}}{x^{\alpha}}-
\epsilon\left.\pfrac{\tilde{j}_{\mathrm{N}}^{\,\alpha}}{x^{\alpha}}\right|_{\text{expl}}\\
&=\epsilon\pfrac{\tilde{j}_{\mathrm{N}}^{\,\alpha}}{x^{\alpha}}+\left.\delta\tilde{\HCd}\right|_{\mathrm{CT}}.
\end{align*}
The requirement that both variations, $\delta\tilde{\HCd}$ and $\left.\delta\tilde{\HCd}\right|_{\mathrm{CT}}$,
weakly coincide ensures that the canonical transformation defines a symmetry transformation.
Then, the vector $\tilde{j}_{\mathrm{N}}^{\,\mu}$ in the generating function~(\ref{eq:F3-infini})
defines the (weakly) conserved Noether current
\begin{equation}\label{eq:gen-noether}
\delta\tilde{\HCd}-\left.\delta\tilde{\HCd}\right|_{\mathrm{CT}}\stackrel{\not\equiv}{=}0\qquad\Leftrightarrow\qquad
\pfrac{\tilde{j}_{\mathrm{N}}^{\,\alpha}}{x^{\alpha}}\stackrel{\not\equiv}{=}0.
\end{equation}
\section{Finite $\mathrm{Diff}(M)$ transformation\label{sec:finite}}
For a combined closed system of a scalar field $\phi$ and a vector field $a_{\mu}$ in a metric-affine space,
the generating function of type $\FCd_{3}^{\mu}$ for the canonical transformation of the (active)
diffeomorphisms that build the $\mathrm{Diff}(M)$ symmetry group is set up as follows in order
to ensure the proper transformation behavior of the fields:
\begin{equation}
\left.\FCd_{3}^{\mu}\right|_{x}=-\tilde{\pi}^{\mu}\,\Phi
-\tilde{p}^{\alpha\mu}\,A_{\beta}\pfrac{X^{\beta}}{x^{\alpha}}
-\tilde{k}^{\alpha\beta\mu}G_{\xi\lambda}\pfrac{X^{\xi}}{x^{\alpha}}\pfrac{X^{\lambda}}{x^{\beta}}
-\tilde{q}\indices{_{\eta}^{\alpha\beta\mu}}
\left(\Gamma\indices{^{\tau}_{\xi\lambda}}\pfrac{x^{\eta}}{X^{\tau}}\pfrac{X^{\xi}}{x^{\alpha}}
\pfrac{X^{\lambda}}{x^{\beta}}+\pfrac{x^{\eta}}{X^{\tau}}\ppfrac{X^{\tau}}{x^{\alpha}}{x^{\beta}}\right).
\label{eq:gen-conn-coeff}
\end{equation}
We remind that a tilde denotes that the respective quantity represents a
tensor density, i.e., a relative tensor of weight $w=1$.
Notice that the transformed metric tensor $G_{\xi\lambda}(X)$ is \emph{symmetric}, which induces the tensor
$\tilde{k}^{\alpha\beta\mu}$ to be symmetric in its first index pair, $\alpha,\beta$.

From the general rules~(\ref{eq:F3-rules}), the specific generating function~(\ref{eq:gen-conn-coeff})
yields the following particular transformation rules for involved fields and their conjugates:
\begin{subequations}
\begin{align}
\phi(x)&=\Phi(X),&a_{\alpha}(x)&=A_{\beta}(X)\pfrac{X^{\beta}}{x^{\alpha}},\label{eq:rules-1}\\
\tilde{\Pi}^{\mu}\!(X)&=\tilde{\pi}^{\kappa}(x)\pfrac{X^{\mu}}{x^{\kappa}}\!\left|\pfrac{x}{X}\right|,&
\tilde{P}^{\nu\mu}(X)&=\tilde{p}^{\alpha\kappa}(x)\pfrac{X^{\nu}}{x^{\alpha}}\pfrac{X^{\mu}}{x^{\kappa}}\!\left|\pfrac{x}{X}\right|,
\end{align}
and for the ``metric and connection fields'' and their conjugates,
\begin{align}
g_{\alpha\beta}(x)&=G_{\xi\lambda}(X)\pfrac{X^{\xi}}{x^{\alpha}}\pfrac{X^{\lambda}}{x^{\beta}},&
\gamma\indices{^{\eta}_{\alpha\beta}}(x)&=\Gamma\indices{^{\tau}_{\xi\lambda}}(X)\pfrac{x^{\eta}}{X^{\tau}}\pfrac{X^{\xi}}{x^{\alpha}}
\pfrac{X^{\lambda}}{x^{\beta}}+\pfrac{x^{\eta}}{X^{\tau}}\ppfrac{X^{\tau}}{x^{\alpha}}{x^{\beta}},\\
\tilde{K}^{\xi\lambda\mu}(X)&=\tilde{k}^{\alpha\beta\kappa}(x)\pfrac{X^{\xi}}{x^{\alpha}}\pfrac{X^{\lambda}}{x^{\beta}}
\pfrac{X^{\mu}}{x^{\kappa}}\!\left|\pfrac{x}{X}\right|,&
\tilde{Q}\indices{_{\tau}^{\xi\lambda\mu}}(X)&=\tilde{q}\indices{_{\eta}^{\alpha\beta\kappa}}(x)\pfrac{x^{\eta}}{X^{\tau}}
\pfrac{X^{\xi}}{x^{\alpha}}\pfrac{X^{\lambda}}{x^{\beta}}\pfrac{X^{\mu}}{x^{\kappa}}\!\left|\pfrac{x}{X}\right|.\label{eq:rules-2}
\end{align}
Note that the connection field $\gamma\indices{^{\eta}_{\alpha\beta}}(x)$ is not assumed to be symmetric in $\alpha$ and $\beta$.
Finally, the particular transformation rule for the covariant Hamiltonian follows from the general rule~(\ref{eq:ham-rule-gen}) as:
\begin{align}
\left.\delta\tilde{\HCd}\right|_{\mathrm{CT}}&=
\left.\tilde{\HCd}^{\prime}\right|_{x}-\left.\tilde{\HCd}\right|_{x}=
-\tilde{p}^{\alpha\mu}a_{\sigma}\pfrac{x^{\sigma}}{X^{\xi}}\ppfrac{X^{\xi}}{x^{\alpha}}{x^{\mu}}
-\tilde{k}^{\alpha\beta\mu}\left(g_{\sigma\beta}\pfrac{x^{\sigma}}{X^{\xi}}\ppfrac{X^{\xi}}{x^{\alpha}}{x^{\mu}}+
g_{\alpha\sigma}\pfrac{x^{\sigma}}{X^{\xi}}\ppfrac{X^{\xi}}{x^{\beta}}{x^{\mu}}\right)\nonumber\\
&\quad+\tilde{q}\indices{_{\eta}^{\alpha\beta\mu}}\left(
\gamma\indices{^{\sigma}_{\alpha\beta}}\pfrac{x^{\eta}}{X^{\xi}}\ppfrac{X^{\xi}}{x^{\sigma}}{x^{\mu}}
-\gamma\indices{^{\eta}_{\sigma\beta}}\pfrac{x^{\sigma}}{X^{\xi}}\ppfrac{X^{\xi}}{x^{\alpha}}{x^{\mu}}
-\gamma\indices{^{\eta}_{\alpha\sigma}}\pfrac{x^{\sigma}}{X^{\xi}}\ppfrac{X^{\xi}}{x^{\beta}}{x^{\mu}}\right.\nonumber\\
&\left.\qquad\qquad\;+\,\pfrac{x^{\eta}}{X^{\tau}}\ppfrac{X^{\tau}}{x^{\sigma}}{x^{\beta}}
\pfrac{x^{\sigma}}{X^{\xi}}\ppfrac{X^{\xi}}{x^{\alpha}}{x^{\mu}}+\pfrac{x^{\eta}}{X^{\tau}}\ppfrac{X^{\tau}}{x^{\sigma}}{x^{\alpha}}
\pfrac{x^{\sigma}}{X^{\xi}}\ppfrac{X^{\xi}}{x^{\beta}}{x^{\mu}}-\pfrac{x^{\eta}}{X^{\xi}}\pppfrac{X^{\xi}}{x^{\alpha}}{x^{\beta}}{x^{\mu}}\right).
\label{eq:ham-rule}
\end{align}
\end{subequations}
We will show that the transformation following from the generating function~(\ref{eq:gen-conn-coeff})
defines a symmetry transformation, in the sense that any action integral in one point of the
manifold $M$ is mapped to an action integral \emph{of the same form} at another point of $M$~\cite{Gaul:1999ys}.
In other words, the \emph{Principle of General Relativity} is implemented here via the generating function $\FCd_{3}^{\mu}$.
\section{Infinitesimal $\mathrm{Diff}(M)$ transformation\label{sec:infini}}
In order to work out the particular form of the Noether current
$\tilde{j}_{\mathrm{N}}^{\,\mu}$ which is associated with an invariance of a given field theory under \emph{active} diffeomorphisms,
the generating function for the corresponding \emph{infinitesimal} canonical transformation needs to be set up first.
To this end, a local parameter vector $h^{\mu}(x)$ is introduced, which defines the (active)
\emph{infinitesimal local diffeomorphism} $x^{\mu}\mapsto X^{\mu}$ on a manifold $M$:
\begin{equation}\label{eq:x-shift}
X^{\mu}=x^{\mu}+\epsilon\,h^{\mu}(x),\qquad L_{h}=h^{\alpha}(x)\,\pfrac{}{x^{\alpha}},
\end{equation}
with $\epsilon\ll1$ and $L_{h}$ the Lie algebra generators of the diffeomorphism group of $M$.
To first order in $\epsilon$, the spacetime dependent coefficients of~(\ref{eq:gen-conn-coeff}) are then expressed as
\begin{equation}\label{eq:infini-spt}
\pfrac{X^{\mu}}{x^{\nu}}=\delta_{\nu}^{\mu}+\epsilon\pfrac{h^{\mu}}{x^{\nu}},\qquad
\pfrac{x^{\mu}}{X^{\nu}}=\delta_{\nu}^{\mu}-\epsilon\pfrac{h^{\mu}}{x^{\nu}},\qquad
\ppfrac{X^{\tau}}{x^{\alpha}}{x^{\beta}}=\epsilon\ppfrac{h^{\tau}}{x^{\alpha}}{x^{\beta}},\qquad
\left|\pfrac{x}{X}\right|=1-\epsilon\pfrac{h^{\tau}}{x^{\tau}}.
\end{equation}
The finite transformation rules~(\ref{eq:rules-1}) to~(\ref{eq:rules-2}) of the fields now follow up to first order in $\epsilon$ as
\begin{subequations}\label{eq:rules-infini}
\begin{alignat}{2}
\delta\phi(x)&=\Phi(x)-\phi(x)&&=-\epsilon h^{\beta}\pfrac{\phi}{x^{\beta}}=-\epsilon\mathscr{L}_{h}\phi(x)\\
\delta a_{\alpha}(x)&=A_{\alpha}(x)-a_{\alpha}(x)&&=-\epsilon\left(
h^{\beta}\pfrac{a_{\alpha}}{x^{\beta}}+a_{\beta}\pfrac{h^{\beta}}{x^{\alpha}}\right)=-\epsilon\mathscr{L}_{h}a_{\alpha}(x)\\
\delta g_{\xi\lambda}(x)&=G_{\xi\lambda}(x)-g_{\xi\lambda}(x)&&=-\epsilon\left(h^{\beta}\pfrac{g_{\xi\lambda}}{x^{\beta}}
+g_{\beta\lambda}\pfrac{h^{\beta}}{x^{\xi}}+g_{\beta\xi}\pfrac{h^{\beta}}{x^{\lambda}}\right)=-\epsilon\mathscr{L}_{h}g_{\xi\lambda}(x)\\
\delta\gamma\indices{^{\eta}_{\lambda\tau}}(x)&=\Gamma\indices{^{\eta}_{\lambda\tau}}(x)-\gamma\indices{^{\eta}_{\lambda\tau}}(x)
&&=-\epsilon\left(h^{\beta}\pfrac{\gamma\indices{^{\eta}_{\lambda\tau}}}{x^{\beta}}-
\gamma\indices{^{\,\beta}_{\lambda\tau}}\pfrac{h^{\eta}}{x^{\beta}}+
\gamma\indices{^{\eta}_{\beta\tau}}\pfrac{h^{\beta}}{x^{\lambda}}+
\gamma\indices{^{\eta}_{\lambda\beta}}\pfrac{h^{\beta}}{x^{\tau}}
+\ppfrac{h^{\eta}}{x^{\lambda}}{x^{\tau}}\right)\nonumber\\
&&&=-\epsilon\mathscr{L}_{h}\gamma\indices{^{\eta}_{\lambda\tau}}(x).
\end{alignat}
\end{subequations}
The differences of the scalar, vector, and tensor fields are actually proportional to conventional (non-covariant)
\emph{Lie derivatives} along the vector field $h^{\beta}(x)$, denoted by $\mathscr{L}_{h}$~\cite{matteucci03,godina05}.
This reflects the general fact that a finite-dimensional continuous group of transformations which lie
infinitesimally close to the identity define the Lie algebra of said group.
The difference of the non-tensorial connection field $\gamma\indices{^{\eta}_{\lambda\tau}}(x)$
is equally proportional to its Lie derivative by virtue of its trans\-formation property~(\ref{eq:rules-2})
(see Schouten~\cite{schouten54}, Eq.~(II~10.34)).

The corresponding differences of the conjugate momentum tensors are proportional
to conventional Lie derivatives of the respective tensor \emph{densities}
\begin{subequations}
\begin{alignat}{2}
\delta\tilde{\pi}^{\mu}(x)&=\tilde{\Pi}^{\mu}(x)-\tilde{\pi}^{\mu}(x)&&=
\epsilon\left[-h^{\beta}\pfrac{\tilde{\pi}^{\mu}}{x^{\beta}}+\tilde{\pi}^{\beta}\left(
\pfrac{h^{\mu}}{x^{\beta}}-\delta_{\beta}^{\mu}\pfrac{h^{\tau}}{x^{\tau}}\right)\right]
=-\epsilon\mathscr{L}_{h}\tilde{\pi}^{\mu}(x)\label{eq:rules-infini-mom-1}\\
\delta\tilde{p}^{\alpha\mu}(x)&=\tilde{P}^{\alpha\mu}(x)-\tilde{p}^{\alpha\mu}(x)&&=
\epsilon\left[-h^{\beta}\pfrac{\tilde{p}^{\alpha\mu}}{x^{\beta}}+
\tilde{p}^{\beta\mu}\pfrac{h^{\alpha}}{x^{\beta}}+\tilde{p}^{\alpha\beta}\left(
\pfrac{h^{\mu}}{x^{\beta}}-\delta_{\beta}^{\mu}\pfrac{h^{\tau}}{x^{\tau}}\right)\right]
=-\epsilon\mathscr{L}_{h}\tilde{p}^{\alpha\mu}(x)\label{eq:rules-infini-mom-2}\\
\delta\tilde{k}^{\lambda\tau\mu}(x)&=\tilde{K}^{\lambda\tau\mu}(x)-\tilde{k}^{\lambda\tau\mu}(x)&&=
\epsilon\left[-h^{\beta}\pfrac{\tilde{k}^{\lambda\tau\mu}}{x^{\beta}}+
\tilde{k}^{\beta\tau\mu}\pfrac{h^{\lambda}}{x^{\beta}}+\tilde{k}^{\lambda\beta\mu}\pfrac{h^{\tau}}{x^{\beta}}+
\tilde{k}^{\lambda\tau\beta}\left(\pfrac{h^{\mu}}{x^{\beta}}-\delta_{\beta}^{\mu}\pfrac{h^{\tau}}{x^{\tau}}\right)\right]\nonumber\\
&&&=-\epsilon\mathscr{L}_{h}\tilde{k}^{\lambda\tau\mu}(x)
\end{alignat}
and
\begin{alignat}{2}
\delta\tilde{q}\indices{_{\eta}^{\lambda\tau\mu}}(x)&=\tilde{Q}\indices{_{\eta}^{\lambda\tau\mu}}(x)
-\tilde{q}\indices{_{\eta}^{\lambda\tau\mu}}(x)&&=
\epsilon\left[-h^{\beta}\pfrac{\tilde{q}\indices{_{\eta}^{\lambda\tau\mu}}}{x^{\beta}}-
\tilde{q}\indices{_{\beta}^{\lambda\tau\mu}}\pfrac{h^{\beta}}{x^{\eta}}+
\tilde{q}\indices{_{\eta}^{\beta\tau\mu}}\pfrac{h^{\lambda}}{x^{\beta}}+
\tilde{q}\indices{_{\eta}^{\lambda\beta\mu}}\pfrac{h^{\tau}}{x^{\beta}}\right.\nonumber\\
&&&\qquad\left.\mbox{}+\tilde{q}\indices{_{\eta}^{\lambda\tau\beta}}\left(
\pfrac{h^{\mu}}{x^{\beta}}-\delta_{\beta}^{\mu}\pfrac{h^{\tau}}{x^{\tau}}\right)\right]
=-\epsilon\mathscr{L}_{h}\tilde{q}\indices{_{\eta}^{\lambda\tau\mu}}(x).
\end{alignat}
\end{subequations}
As a particular feature of the covariant Hamiltonian formalism of field theories,
merely the \emph{divergences} of momentum fields are determined by the system's Hamiltonian
and not the individual components of the respective canonical momentum tensor.
The canonical momentum tensors are thus determined only up to additional divergence-free tensors.
The transformation rule for the \emph{divergence} of the momenta $\tilde{\pi}^{\mu}$ is set up on the basis of Eq.~(\ref{eq:rules-infini-mom-1})
\begin{align*}
\pfrac{\left(\delta\tilde{\pi}^{\mu}(x)\right)}{x^{\mu}}=
\pfrac{\tilde{\Pi}^{\mu}}{x^{\mu}}-\pfrac{\tilde{\pi}^{\mu}}{x^{\mu}}
&=\epsilon\left[\cancel{-\pfrac{h^{\beta}}{x^{\mu}}\pfrac{\tilde{\pi}^{\mu}}{x^{\beta}}}
-h^{\beta}\ppfrac{\tilde{\pi}^{\mu}}{x^{\beta}}{x^{\mu}}
+\pfrac{\tilde{\pi}^{\beta}}{x^{\mu}}\left(\cancel{\pfrac{h^{\mu}}{x^{\beta}}}-
\delta_{\beta}^{\mu}\pfrac{h^{\tau}}{x^{\tau}}\right)\right]\\
&=-\epsilon\left(h^{\mu}\ppfrac{\tilde{\pi}^{\beta}}{x^{\mu}}{x^{\beta}}
+\pfrac{h^{\mu}}{x^{\mu}}\pfrac{\tilde{\pi}^{\beta}}{x^{\beta}}\right)\\
&=-\epsilon\pfrac{}{x^{\mu}}\left(h^{\mu}\pfrac{\tilde{\pi}^{\beta}}{x^{\beta}}\right).
\end{align*}
We are thus allowed to replace the transformation rule~(\ref{eq:rules-infini-mom-1}) by
\begin{subequations}\label{eq:rules-infini-mom}
\begin{equation}
\delta\tilde{\pi}^{\mu}=-\epsilon h^{\mu}\pfrac{\tilde{\pi}^{\beta}}{x^{\beta}}=\epsilon h^{\mu}\pfrac{\tilde{\HCd}}{\phi},
\end{equation}
which amounts to replacing the momentum tensor components $\tilde{\pi}^{\mu}$
by modified components of an equivalent momentum tensor with the same divergence.
Similarly, the divergences of the momenta $\tilde{p}^{\alpha\mu}$ transform as
\begin{align*}
\pfrac{\tilde{P}^{\alpha\mu}}{x^{\mu}}-\pfrac{\tilde{p}^{\alpha\mu}}{x^{\mu}}
&=\epsilon\left[\cancel{-\pfrac{h^{\beta}}{x^{\mu}}\pfrac{\tilde{p}^{\alpha\mu}}{x^{\beta}}}
-h^{\beta}\ppfrac{\tilde{p}^{\alpha\mu}}{x^{\beta}}{x^{\mu}}
+\pfrac{\tilde{p}^{\beta\mu}}{x^{\mu}}\pfrac{h^{\alpha}}{x^{\beta}}
+\tilde{p}^{\beta\mu}\ppfrac{h^{\alpha}}{x^{\beta}}{x^{\mu}}
+\pfrac{\tilde{p}^{\alpha\beta}}{x^{\mu}}\left(
\cancel{\pfrac{h^{\mu}}{x^{\beta}}}-\delta_{\beta}^{\mu}\pfrac{h^{\tau}}{x^{\tau}}\right)\right]\\
&=-\epsilon\left(h^{\mu}\ppfrac{\tilde{p}^{\alpha\beta}}{x^{\mu}}{x^{\beta}}
-\pfrac{\tilde{p}^{\beta\mu}}{x^{\mu}}\pfrac{h^{\alpha}}{x^{\beta}}
-\tilde{p}^{\beta\mu}\ppfrac{h^{\alpha}}{x^{\beta}}{x^{\mu}}
+\pfrac{\tilde{p}^{\alpha\beta}}{x^{\beta}}\pfrac{h^{\mu}}{x^{\mu}}\right)\\
&=-\epsilon\pfrac{}{x^{\mu}}\left(h^{\mu}\pfrac{\tilde{p}^{\alpha\beta}}{x^{\beta}}
-\tilde{p}^{\beta\mu}\pfrac{h^{\alpha}}{x^{\beta}}\right).
\end{align*}
The transformation rule for the momenta $\tilde{p}^{\alpha\mu}$ from Eq.~(\ref{eq:rules-infini-mom-2})
can thus equivalently be expressed as
\begin{equation}
\delta\tilde{p}^{\alpha\mu}=
-\epsilon\left(h^{\mu}\pfrac{\tilde{p}^{\alpha\beta}}{x^{\beta}}
-\tilde{p}^{\beta\mu}\pfrac{h^{\alpha}}{x^{\beta}}\right)
=\epsilon\left(h^{\mu}\pfrac{\tilde{\HCd}}{a_{\alpha}}
+\tilde{p}^{\beta\mu}\pfrac{h^{\alpha}}{x^{\beta}}\right).
\end{equation}
The corresponding modified transformation rules apply for the momenta $\tilde{k}^{\lambda\xi\mu}$
and $\tilde{q}\indices{_{\eta}^{\lambda\tau\mu}}$:
\begin{align}
\delta\tilde{k}^{\lambda\xi\mu}&=\epsilon\left(
h^{\mu}\pfrac{\tilde{\HCd}}{g_{\lambda\xi}}+\tilde{k}^{\beta\xi\mu}\pfrac{h^{\lambda}}{x^{\beta}}+
\tilde{k}^{\lambda\beta\mu}\pfrac{h^{\xi}}{x^{\beta}}\right)\\
\delta\tilde{q}\indices{_{\eta}^{\lambda\tau\mu}}&=\epsilon\left(
h^{\mu}\pfrac{\tilde{\HCd}}{\gamma\indices{^{\eta}_{\lambda\tau}}}-
\tilde{q}\indices{_{\beta}^{\lambda\tau\mu}}\pfrac{h^{\beta}}{x^{\eta}}+
\tilde{q}\indices{_{\eta}^{\beta\tau\mu}}\pfrac{h^{\lambda}}{x^{\beta}}+
\tilde{q}\indices{_{\eta}^{\lambda\beta\mu}}\pfrac{h^{\tau}}{x^{\beta}}\right).
\end{align}
\end{subequations}
\section{The Noether theorem in the form of Eq.~(\ref{eq:gen-noether})}
Finally, the finite canonical transformation rule for the Hamiltonian density
from Eq.~(\ref{eq:ham-rule}) has the infinitesimal representation according to Eqs.~(\ref{eq:infini-spt}):
\begin{align*}
\left.\delta\tilde{\HCd}\right|_{\mathrm{CT}}=-\epsilon&\left[\ppfrac{h^{\beta}}{x^{\alpha}}{x^{\mu}}
\left(\tilde{p}^{\alpha\mu}\,a_{\beta}
+\tilde{k}^{\lambda\alpha\mu}g_{\beta\lambda}+\tilde{k}^{\alpha\lambda\mu}g_{\lambda\beta}
+\tilde{q}\indices{_{\eta}^{\lambda\alpha\mu}}\gamma\indices{^{\eta}_{\lambda\beta}}+
\tilde{q}\indices{_{\eta}^{\alpha\lambda\mu}}\gamma\indices{^{\eta}_{\beta\lambda}}-
\tilde{q}\indices{_{\beta}^{\eta\lambda\mu}}\gamma\indices{^{\alpha}_{\eta\lambda}}\right)
\vphantom{\pppfrac{h^{\beta}}{x^{\alpha}}{x^{\eta}}{x^{\mu}}}\right.\\
&\left.\mbox{}+\pppfrac{h^{\beta}}{x^{\alpha}}{x^{\eta}}{x^{\mu}}\tilde{q}\indices{_{\beta}^{\alpha\eta\mu}}\right].
\end{align*}
In contrast, the variation of $\delta\tilde{\HCd}$ emerging from the transformations of the fields, as stated by Eq.~(\ref{eq:delta-H}),
has the following infinitesimal representation according to Eqs.~(\ref{eq:rules-infini}) and Eqs.~(\ref{eq:rules-infini-mom}):
\begin{align*}
\delta\tilde{\HCd}=\epsilon&\left[\pfrac{h^{\beta}}{x^{\alpha}}\pfrac{}{x^{\mu}}\left(\tilde{p}^{\alpha\mu}\,a_{\beta}
+\tilde{k}^{\lambda\alpha\mu}g_{\beta\lambda}+\tilde{k}^{\alpha\lambda\mu}g_{\lambda\beta}
+\tilde{q}\indices{_{\eta}^{\lambda\alpha\mu}}\gamma\indices{^{\eta}_{\lambda\beta}}+
\tilde{q}\indices{_{\eta}^{\alpha\lambda\mu}}\gamma\indices{^{\eta}_{\beta\lambda}}-
\tilde{q}\indices{_{\beta}^{\eta\lambda\mu}}\gamma\indices{^{\alpha}_{\eta\lambda}}\right)
\vphantom{\pfrac{\tilde{q}\indices{_{\beta}^{\alpha\eta\mu}}}{x^{\mu}}}\right.\\
&\left.\mbox{}+\ppfrac{h^{\beta}}{x^{\alpha}}{x^{\eta}}\pfrac{\tilde{q}\indices{_{\beta}^{\alpha\eta\mu}}}{x^{\mu}}\right].
\end{align*}
The difference of the transformations of the Hamiltonian yields the divergence:
\begin{align}
\delta\tilde{\HCd}-\left.\delta\tilde{\HCd}\right|_{\mathrm{CT}}=
\epsilon\pfrac{}{x^{\mu}}&\left[\pfrac{h^{\beta}}{x^{\alpha}}\left(\tilde{p}^{\alpha\mu}\,a_{\beta}
+\tilde{k}^{\lambda\alpha\mu}g_{\beta\lambda}+\tilde{k}^{\alpha\lambda\mu}g_{\lambda\beta}
+\tilde{q}\indices{_{\eta}^{\lambda\alpha\mu}}\gamma\indices{^{\eta}_{\lambda\beta}}+
\tilde{q}\indices{_{\eta}^{\alpha\lambda\mu}}\gamma\indices{^{\eta}_{\beta\lambda}}-
\tilde{q}\indices{_{\beta}^{\eta\lambda\mu}}\gamma\indices{^{\alpha}_{\eta\lambda}}\right)
\vphantom{\ppfrac{h^{\beta}}{x^{\alpha}}{x^{\eta}}}\right.\nonumber\\
&\left.\mbox{}+\ppfrac{h^{\beta}}{x^{\alpha}}{x^{\eta}}\tilde{q}\indices{_{\beta}^{\alpha\eta\mu}}\right].
\label{eq:noether-current0}
\end{align}
According to Noether's theorem from Eq.~(\ref{eq:gen-noether}), the divergence of the Noether current
$\tilde{j}_{\mathrm{N}}^{\,\mu}$ vanishes exactly if $\delta\tilde{\HCd}-\left.\delta\tilde{\HCd}\right|_{\mathrm{CT}}=0$.
Then, the finite canonical transformation defined by the generating function~(\ref{eq:gen-conn-coeff})
establishes a \emph{symmetry transformation} which leaves the form of the given system invariant under
the Diff$(M)$ symmetry group---and thereby establishes the \emph{General Principle of Relativity}.
As we will see in the next section, requiring a conserved Noether current, namely a (weakly) vanishing
of the right-hand side of Eq.~(\ref{eq:noether-current0}), provides us with the generic theory of geometrodynamics.
\section{The Noether current $\tilde{j}_{\mathrm{N}}^{\,\mu}$ associated with diffeomorphism\label{sec:noether-current}}
The total set of canonical transformation rules~(\ref{eq:rules-infini3}) for the infinitesimal local translation~(\ref{eq:x-shift}) is:
\begin{subequations}\label{eq:rules-infini2}
\begin{alignat}{2}
-\frac{1}{\epsilon}\delta_{\nu}^{\mu}\,\delta\phi&=\pfrac{\tilde{j}_{\mathrm{N}}^{\,\mu}}{\tilde{\pi}^{\nu}}&&
\stackrel{!}{=}\,\delta_{\nu}^{\mu} h^{\beta}\pfrac{\phi}{x^{\beta}}\\
\frac{1}{\epsilon}\delta\tilde{\pi}^{\mu}&=\pfrac{\tilde{j}_{\mathrm{N}}^{\,\mu}}{\phi}&&
\stackrel{!}{=}\,h^{\mu}\pfrac{\tilde{\HCd}}{\phi}\\
-\frac{1}{\epsilon}\delta_{\nu}^{\mu}\,\delta a_{\alpha}&=\pfrac{\tilde{j}_{\mathrm{N}}^{\,\mu}}{\tilde{p}^{\alpha\nu}}&&
\stackrel{!}{=}\,\delta_{\nu}^{\mu}\left(
h^{\beta}\pfrac{a_{\alpha}}{x^{\beta}}+a_{\beta}\pfrac{h^{\beta}}{x^{\alpha}}\right)\\
\frac{1}{\epsilon}\delta\tilde{p}^{\alpha\mu}&=\pfrac{\tilde{j}_{\mathrm{N}}^{\,\mu}}{a_{\alpha}}&&\stackrel{!}{=}
h^{\mu}\pfrac{\tilde{\HCd}}{a_{\alpha}}+\tilde{p}^{\beta\mu}\pfrac{h^{\alpha}}{x^{\beta}}\\
-\frac{1}{\epsilon}\delta_{\nu}^{\mu}\,\delta g_{\xi\lambda}&=\pfrac{\tilde{j}_{\mathrm{N}}^{\,\mu}}{\tilde{k}^{\lambda\xi\nu}}
&&\stackrel{!}{=}\,\delta_{\nu}^{\mu}\left(h^{\beta}\pfrac{g_{\xi\lambda}}{x^{\beta}}
+g_{\beta\lambda}\pfrac{h^{\beta}}{x^{\xi}}+g_{\beta\xi}\pfrac{h^{\beta}}{x^{\lambda}}\right)\\
\frac{1}{\epsilon}\delta\tilde{k}^{\lambda\xi\mu}&=\pfrac{\tilde{j}_{\mathrm{N}}^{\,\mu}}{g_{\lambda\xi}}&&\stackrel{!}{=}
h^{\mu}\pfrac{\tilde{\HCd}}{g_{\lambda\xi}}+\tilde{k}^{\beta\xi\mu}\pfrac{h^{\lambda}}{x^{\beta}}+
\tilde{k}^{\lambda\beta\mu}\pfrac{h^{\xi}}{x^{\beta}}
\end{alignat}
and
\begin{alignat}{2}
-\frac{1}{\epsilon}\delta_{\nu}^{\mu}\,\delta\gamma\indices{^{\eta}_{\lambda\tau}}&=
\pfrac{\tilde{j}_{\mathrm{N}}^{\,\mu}}{\tilde{q}\indices{_{\eta}^{\lambda\tau\nu}}}
&&\stackrel{!}{=}\,\delta_{\nu}^{\mu}\left(h^{\beta}\pfrac{\gamma\indices{^{\eta}_{\lambda\tau}}}{x^{\beta}}-
\gamma\indices{^{\beta}_{\lambda\tau}}\pfrac{h^{\eta}}{x^{\beta}}+
\gamma\indices{^{\eta}_{\beta\tau}}\pfrac{h^{\beta}}{x^{\lambda}}+
\gamma\indices{^{\eta}_{\lambda\beta}}\pfrac{h^{\beta}}{x^{\tau}}
+\ppfrac{h^{\eta}}{x^{\lambda}}{x^{\tau}}\right)\\
\frac{1}{\epsilon}\delta\tilde{q}\indices{_{\eta}^{\lambda\tau\mu}}&=
\pfrac{\tilde{j}_{\mathrm{N}}^{\,\mu}}{\gamma\indices{^{\eta}_{\lambda\tau}}}
&&\stackrel{!}{=}h^{\mu}\pfrac{\tilde{\HCd}}{\gamma\indices{^{\eta}_{\lambda\tau}}}-
\tilde{q}\indices{_{\beta}^{\lambda\tau\mu}}\pfrac{h^{\beta}}{x^{\eta}}+
\tilde{q}\indices{_{\eta}^{\beta\tau\mu}}\pfrac{h^{\lambda}}{x^{\beta}}+
\tilde{q}\indices{_{\eta}^{\lambda\beta\mu}}\pfrac{h^{\tau}}{x^{\beta}}.
\end{alignat}
\end{subequations}
The Noether current $\tilde{j}_{\mathrm{N}}^{\,\mu}$, which defines the particular
infinitesimal transformation rules~(\ref{eq:rules-infini2}), follows as
\begin{equation}
\boxed{
\tilde{j}_{\mathrm{N}}^{\,\mu}=h^{\,\beta}\tilde{B}\indices{_{\beta}^{\mu}}+
\pfrac{h^{\,\beta}}{x^{\alpha}}\tilde{C}\indices{_{\beta}^{\alpha\mu}}+
\ppfrac{h^{\,\beta}}{x^{\alpha}}{x^{\eta}}\,\tilde{q}\indices{_{\beta}^{\alpha\eta\mu}},}
\label{eq:noether-current}
\end{equation}
wherein $\tilde{B}\indices{_{\beta}^{\mu}}$ abbreviates the sum of all terms
emerging from Eqs.~(\ref{eq:rules-infini2}) proportional to $h^{\beta}$
\begin{subequations}
\begin{align}
\tilde{B}\indices{_{\beta}^{\mu}}&=\tilde{\pi}^{\mu}\pfrac{\phi}{x^{\beta}}+
\tilde{p}^{\,\alpha\mu}\pfrac{a_{\alpha}}{x^{\beta}}+\tilde{k}^{\alpha\lambda\mu}\pfrac{g_{\alpha\lambda}}{x^{\beta}}+
\tilde{q}\indices{_{\eta}^{\alpha\lambda\mu}}\pfrac{\gamma\indices{^{\eta}_{\alpha\lambda}}}{x^{\beta}}\nonumber\\
&\quad\mbox{}-\delta_{\beta}^{\mu}\left(\tilde{\pi}^{\tau}\pfrac{\phi}{x^{\tau}}+
\tilde{p}^{\,\alpha\tau}\pfrac{a_{\alpha}}{x^{\tau}}+\tilde{k}^{\alpha\lambda\tau}\pfrac{g_{\alpha\lambda}}{x^{\tau}}+
\tilde{q}\indices{_{\eta}^{\alpha\lambda\tau}}\pfrac{\gamma\indices{^{\eta}_{\alpha\lambda}}}{x^{\tau}}-\tilde{\HCd}\right),
\label{eq:B-def}
\end{align}
whereas $\tilde{C}\indices{_{\beta}^{\alpha\mu}}$ stands for the collection of the terms proportional to $\partial h^{\beta}/\partial x^{\alpha}$:
\begin{align}
\tilde{C}\indices{_{\beta}^{\alpha\mu}}&=\tilde{p}^{\alpha\mu}\,a_{\beta}
+\tilde{k}^{\alpha\lambda\mu}g_{\beta\lambda}+\tilde{k}^{\lambda\alpha\mu}g_{\lambda\beta}
+\tilde{q}\indices{_{\eta}^{\alpha\lambda\mu}}\gamma\indices{^{\eta}_{\beta\lambda}}
+\tilde{q}\indices{_{\eta}^{\lambda\alpha\mu}}\gamma\indices{^{\eta}_{\lambda\beta}}
-\tilde{q}\indices{_{\beta}^{\eta\lambda\mu}}\gamma\indices{^{\alpha}_{\eta\lambda}}.\label{eq:C-def}
\end{align}
\end{subequations}
According to the Hamiltonian form of Noether's theorem, the canonical
trans\-formation rules emerging from the generating function $\FCd_{3}^{\mu}$
(Eq.~(\ref{eq:F3-infini})) represent a symmetry transformation of the given system exactly if
the divergence of the function $\tilde{j}_{\mathrm{N}}^{\,\mu}(x)$ contained therein vanishes:
\begin{equation*}
\pfrac{\tilde{j}_{\mathrm{N}}^{\,\mu}}{x^{\mu}}=0\qquad\Leftrightarrow\qquad
\tilde{j}_{\mathrm{N}}^{\,\mu}\text{ in }\FCd_{3}^{\mu}\text{ defines an (infinitesimal) symmetry transformation}.
\end{equation*}
Then, $\tilde{j}_{\mathrm{N}}^{\,\mu}$ represents the conserved Noether current.
Here, $\tilde{B}\indices{_{\beta}^{\mu}}$ is actually the \emph{local} Hamiltonian representation
of the canonical energy-momentum tensor density of the total dynamical system consisting of scalar
and vector fields in conjunction with a dynamic metric and connection.
As will be derived in the following section, the terms emerging from $\tilde{C}\indices{_{\beta}^{\alpha\mu}}$
in the condition for a \emph{conserved} Noether current convert the \emph{partial} derivatives in
$\tilde{B}\indices{_{\beta}^{\mu}}$ into \emph{covariant} derivatives---and hence into the
\emph{global} Hamiltonian representation of the canonical energy-momentum tensor of the total system.
\section{Discussion of the conserved Noether current for a Poincar\'e symmetry transformation}
For the particular case of $h^{\beta}(x)=h\indices{_0^{\beta}}+h\indices{_1^{\beta}_{\alpha}}\,x^{\alpha}$ a \emph{linear} function of $x$,
Eq.~(\ref{eq:x-shift}) defines the infinitesimal Poincar\'e transformation if $h_{1,\beta\alpha}$ is skew-symmetric:
\begin{equation}\label{eq:linear-h}
X_{\beta}=x_{\beta}+\epsilon\left(h_{0,\beta}+h_{1,\beta\alpha}\,x^{\alpha}\right),
\qquad h_{0,\beta},h_{1,\beta\alpha}=\mathrm{const.},\qquad h_{1,(\beta\alpha)}=0.
\end{equation}
The Noether current~(\ref{eq:noether-current}) then reduces to
\begin{equation*}
\tilde{j}_{\mathrm{N}}^{\,\mu}=h_{0,\beta}\tilde{B}\indices{^{\beta}^{\mu}}+
h\indices{_1_{,[\beta}_{\alpha]}}\left(\tilde{B}\indices{^{\beta}^{\mu}}x^\alpha+\tilde{C}\indices{^{\beta\alpha\mu}}\right),
\end{equation*}
hence the pertaining condition for a conserved Noether current follows as:
\begin{equation*}
\pfrac{\tilde{j}_{\mathrm{N}}^{\,\mu}}{x^{\mu}}=h_{0,\beta}\pfrac{\tilde{B}^{\beta\mu}}{x^{\mu}}
+h_{1,[\beta\alpha]}\left(\tilde{B}^{\beta\alpha}+\pfrac{\tilde{B}^{\beta\mu}}{x^\mu}x^\alpha
+\pfrac{\tilde{C}^{\beta\alpha\mu}}{x^{\mu}}\right)\stackrel{!}{=}0,
\end{equation*}
which splits for a constant spacetime geometry into the two conservation equations~\cite{hehl76b,Hehl:2014eja}:
\begin{equation*}
\pfrac{B\indices{^{\beta}^{\mu}}}{x^{\mu}}=0,\qquad
B\indices{^{[\beta\alpha]}}+\pfrac{C\indices{^{[\beta\alpha]\mu}}}{x^{\mu}}=0
\end{equation*}
for both the energy-momentum and the angular momentum, which hold exactly if the given system is invariant under Poincar\'e transformations.
With
\begin{equation*}
C\indices{^{[\beta\alpha]\mu}}=\onehalf\left(p^{\alpha\mu}a^{\beta}-p^{\beta\mu}a^{\alpha}\right),
\end{equation*}
the skew-symmetric part of $C\indices{^{\beta\alpha\mu}}$
then defines the canonical spin tensor of the source vector field.

The Noether current~(\ref{eq:noether-current}) thus generalizes this case to a dynamic spacetime geometry if
$h^{\,\beta}(x)$ stands for a differentiable vector function of spacetime with non-vanishing second and third derivatives.
This will be discussed in the following section.
\section{Discussion of the general condition for a conserved Noether current\label{sec:discussion}}
For the general case, the divergence of the Noether current~(\ref{eq:noether-current}) is obtained as
\begin{align}
\pfrac{\tilde{j}_{\mathrm{N}}^{\,\mu}}{x^{\mu}}&=h^{\beta}\pfrac{\tilde{B}\indices{_{\beta}^{\mu}}}{x^{\mu}}
+\pfrac{h^{\beta}}{x^{\alpha}}\left(\tilde{B}\indices{_{\beta}^{\alpha}}+\pfrac{\tilde{C}\indices{_{\beta}^{\alpha\mu}}}{x^{\mu}}\right)
+\ppfrac{h^{\beta}}{x^{\alpha}}{x^{\eta}}\left(\tilde{C}\indices{_{\beta}^{\alpha\eta}}+
\pfrac{\tilde{q}\indices{_{\beta}^{\alpha\eta\mu}}}{x^{\mu}}\right)
+\pppfrac{h^{\beta}}{x^{\alpha}}{x^{\eta}}{x^{\mu}}\,\tilde{q}\indices{_{\beta}^{\alpha\eta\mu}}\stackrel{!}{=}0.
\label{eq:gen-conn-infini1}
\end{align}
With this equation involving a vanishing \emph{partial} derivative of the
Noether current $\tilde{j}_{\mathrm{N}}^{\,\mu}$, it establishes a proper (local) conservation law.
Yet, the field equations emerging from Eq.~(\ref{eq:gen-conn-infini1})
will turn out to be tensor equations and thus hold invariantly in any reference frame.
As $h^{\beta}(x)$ is supposed to be an arbitrary function of $x$, Eq.~(\ref{eq:gen-conn-infini1}) has
$4+16+40+80=140$ independent coefficients and thus defines a finite-dimensional subgroup of the infinite-dimensional diffeomorphism group.
We note that in this description the torsion degrees of freedom do not emerge from separate dynamical quantities
but are implicitly contained in the additional freedom of the Riemann-Cartan tensor, which is constructed
on the basis of the $64$ non-symmetric connection coefficients $\gamma\indices{^{\eta}_{\alpha\beta}}$.
In contrast, the Riemann tensor of Einstein's general relativity emerges from the $40$ connection
coefficients---referred to as Christoffel symbols---that are symmetric in their lower index pair.

The  four separate conditions for each order of derivatives of $h^{\beta}(x)$
will be worked out in the following sections.
\subsection{Condition~1: term proportional to the third partial derivatives of $h^{\beta}$}
The only term proportional to the third derivative of $h^{\beta}$ is the canonical momentum
$\tilde{q}\indices{_{\beta}^{\alpha\eta\mu}}$, hence the dual of the partial $x^{\mu}$-derivative
of the connection $\gamma\indices{^{\,\beta}_{\alpha\eta}}$.
A necessary and sufficient condition for this term to vanish is that the
(generally non-zero) momentum $\tilde{q}\indices{_{\beta}^{\alpha\eta\mu}}$
is skew-symmetric in one of the index pairs formed out of $\alpha$, $\eta$, and $\mu$.
We choose here the last index pair, namely $\eta$ and $\mu$, and define
\begin{equation}\label{eq:noe-cond3}
\boxed{\tilde{q}\indices{_{\beta}^{\alpha\eta\mu}}=-\tilde{q}\indices{_{\beta}^{\alpha\mu\eta}}}
\end{equation}
which implies that $\tilde{q}\indices{_{\beta}^{\alpha\eta\mu}}$ need not
in addition be skew-symmetric in $\alpha$ and $\eta$.
Equation~(\ref{eq:gen-conn-infini1}) then simplifies to
\begin{equation}\label{eq:gen-conn-infini2}
\pfrac{\tilde{j}_{\mathrm{N}}^{\,\mu}}{x^{\mu}}=h^{\,\beta}\pfrac{\tilde{B}\indices{_{\beta}^{\mu}}}{x^{\mu}}
+\pfrac{h^{\,\beta}}{x^{\alpha}}\left(\tilde{B}\indices{_{\beta}^{\alpha}}+
\pfrac{\tilde{C}\indices{_{\beta}^{\alpha\mu}}}{x^{\mu}}\right)
+\ppfrac{h^{\beta}}{x^{\alpha}}{x^{\eta}}\left(\tilde{C}\indices{_{\beta}^{\alpha\eta}}+
\pfrac{\tilde{q}\indices{_{\beta}^{\alpha\eta\mu}}}{x^{\mu}}\right)\stackrel{!}{=}0.
\end{equation}
\subsection{Condition~2: terms proportional to the second partial derivatives of $h^{\beta}$}
A zero divergence of the Noether current for any symmetry transformation~(\ref{eq:x-shift})---hence
for arbitrary functions $h^{\beta}(x)$---requires in particular that the sum of terms related to the
\emph{second} derivatives of $h^{\beta}(x)$ in Eq.~(\ref{eq:gen-conn-infini2}) vanishes.
This means with $\tilde{C}\indices{_{\beta}^{\alpha\eta}}$ from Eq.~(\ref{eq:C-def})
inserted into the last term of Eq.~(\ref{eq:gen-conn-infini2}):
\begin{equation*}
\ppfrac{h^{\,\beta}}{x^{\alpha}}{x^{\eta}}\left(
\pfrac{\tilde{q}\indices{_{\beta}^{\alpha\eta\mu}}}{x^{\mu}}+\tilde{p}^{\alpha\eta}\,a_{\beta}
+\tilde{k}^{\alpha\lambda\eta}g_{\beta\lambda}+\tilde{k}^{\lambda\alpha\eta}g_{\lambda\beta}
+\tilde{q}\indices{_{\xi}^{\alpha\lambda\eta}}\gamma\indices{^{\xi}_{\beta\lambda}}
+\cancel{\tilde{q}\indices{_{\xi}^{\lambda\alpha\eta}}\gamma\indices{^{\xi}_{\lambda\beta}}}
-\tilde{q}\indices{_{\beta}^{\xi\lambda\eta}}\gamma\indices{^{\alpha}_{\xi\lambda}}\right)\stackrel{!}{=}0.
\end{equation*}
Due to the symmetry of the second partial derivatives of $h^{\,\beta}$ in $\alpha$ and $\eta$,
the term $\tilde{q}\indices{_{\xi}^{\lambda\alpha\eta}}\gamma\indices{^{\xi}_{\lambda\beta}}$
drops out by virtue of the skew-symmetry condition~(\ref{eq:noe-cond3}).
As no symmetries in $\alpha$ and $\eta$ are implied in the remaining terms, one encounters the \emph{sufficient condition}:
\begin{equation}\label{eq:noe-cond1}
\boxed{\pfrac{\tilde{q}\indices{_{\beta}^{\alpha\eta\mu}}}{x^{\mu}}
+\tilde{p}^{\alpha\eta}\,a_{\beta}+\tilde{k}^{\alpha\lambda\eta}g_{\beta\lambda}+\tilde{k}^{\lambda\alpha\eta}g_{\lambda\beta}
+\tilde{q}\indices{_{\xi}^{\alpha\lambda\eta}}\gamma\indices{^{\xi}_{\beta\lambda}}
-\tilde{q}\indices{_{\beta}^{\xi\lambda\eta}}\gamma\indices{^{\alpha}_{\xi\lambda}}=0.}
\end{equation}
We can express Eq.~(\ref{eq:noe-cond1}) equivalently as the tensor equation
\begin{equation}\label{eq:noe-cond1a}
\tilde{q}\indices{_{\beta}^{\alpha\eta\mu}_{;\mu}}
+\tilde{p}^{\alpha\eta}\,a_{\beta}+\tilde{k}^{\alpha\lambda\eta}g_{\beta\lambda}+\tilde{k}^{\lambda\alpha\eta}g_{\lambda\beta}
-\tilde{q}\indices{_{\beta}^{\alpha\xi\mu}}s\indices{^{\eta}_{\xi\mu}}
-2\tilde{q}\indices{_{\beta}^{\alpha\eta\mu}}s\indices{^{\xi}_{\mu\xi}}=0,
\end{equation}
with $s\indices{^{\eta}_{\xi\mu}}\equiv\gamma\indices{^{\eta}_{[\xi\mu]}}$ the Cartan torsion tensor.
It agrees with the corresponding field equation~(56) of Ref.~\cite{struckmeier17a}.
Its implications will be discussed in Sects.~\ref{sec:dis1} and~\ref{sec:dis0}.

With Eq.~(\ref{eq:noe-cond1}), the condition~(\ref{eq:gen-conn-infini2})
for the divergence of the Noether current now further simplifies to
\begin{equation}
\pfrac{\tilde{j}_{\mathrm{N}}^{\,\mu}}{x^{\mu}}=h^{\,\beta}\pfrac{\tilde{B}\indices{_{\beta}^{\mu}}}{x^{\mu}}
+\pfrac{h^{\,\beta}}{x^{\alpha}}\left(\tilde{B}\indices{_{\beta}^{\alpha}}+
\pfrac{\tilde{C}\indices{_{\beta}^{\alpha\eta}}}{x^{\eta}}\right)\stackrel{!}{=}0.
\label{eq:gen-conn-infini4}
\end{equation}
\subsection{Condition~3: terms proportional to the first partial derivatives of $h^{\,\beta}$}
For a generally conserved Noether current, the coefficient proportional to the first
derivative of $h^{\,\beta}$ in Eq.~(\ref{eq:gen-conn-infini4}) must vanish as well, hence
\begin{equation}\label{eq:noe-cond2}
\tilde{B}\indices{_{\beta}^{\alpha}}+\pfrac{\tilde{C}\indices{_{\beta}^{\alpha\eta}}}{x^{\eta}}\stackrel{!}{=}0.
\end{equation}
Equation~(\ref{eq:noe-cond2}) writes in expanded form with $\tilde{C}\indices{_{\beta}^{\alpha\mu}}$ from Eq.~(\ref{eq:C-def})
\begin{equation*}
\tilde{B}\indices{_{\beta}^{\alpha}}+\pfrac{}{x^{\eta}}\left(
\tilde{p}^{\alpha\eta}\,a_{\beta}+\tilde{k}^{\alpha\lambda\eta}g_{\beta\lambda}+\tilde{k}^{\lambda\alpha\eta}g_{\lambda\beta}
+\tilde{q}\indices{_{\xi}^{\alpha\lambda\eta}}\gamma\indices{^{\xi}_{\beta\lambda}}
+\tilde{q}\indices{_{\xi}^{\lambda\alpha\eta}}\gamma\indices{^{\xi}_{\lambda\beta}}
-\tilde{q}\indices{_{\beta}^{\xi\lambda\eta}}\gamma\indices{^{\alpha}_{\xi\lambda}}\right)=0,
\end{equation*}
which is expressed equivalently inserting Eq.~(\ref{eq:noe-cond1})
\begin{equation*}
\tilde{B}\indices{_{\beta}^{\alpha}}+\pfrac{}{x^{\eta}}\left(
-\pfrac{\tilde{q}\indices{_{\beta}^{\alpha\eta\mu}}}{x^{\mu}}+
\tilde{q}\indices{_{\xi}^{\lambda\alpha\eta}}\gamma\indices{^{\xi}_{\lambda\beta}}\right)=0.
\end{equation*}
This equation reduces due to the skew-symmetry of $\tilde{q}\indices{_{\beta}^{\alpha\mu\eta}}$ in its last index pair to
\begin{equation}
\boxed{\tilde{B}\indices{_{\beta}^{\alpha}}+\pfrac{}{x^{\eta}}\left(
\tilde{q}\indices{_{\xi}^{\lambda\alpha\eta}}\gamma\indices{^{\xi}_{\lambda\beta}}\right)=0.}
\label{eq:noe-cond2a}
\end{equation}
As $\tilde{B}\indices{_{\beta}^{\alpha}}$---defined by Eq.~(\ref{eq:B-def})---is the local representation
of the canonical energy-momentum tensor of the \emph{total} system of source fields and dynamic spacetime,
Eq.~(\ref{eq:noe-cond2a}) establishes a correlation of this (pseudo-)tensor with the dynamic spacetime.
The explicit form of this equation will be discussed in Sect.~\ref{sec:dis2}.
\subsection{Condition~4: term proportional to $h^{\beta}$}
Finally, the term proportional to $h^{\,\beta}$ in Eq.~(\ref{eq:gen-conn-infini4}) must separately vanish
\begin{equation}\label{eq:noe-cond4}
\boxed{\pfrac{\tilde{B}\indices{_{\beta}^{\mu}}}{x^{\mu}}=0.}
\end{equation}
Equation~(\ref{eq:noe-cond4}) thus establishes a \emph{local} energy and momentum conservation law
of the \emph{total} system of scalar and vector source fields fields on the one hand, and the dynamic
spacetime, described by the metric and the connection on the other hand.
It turns out to coincide with the divergence of Eq.~(\ref{eq:noe-cond2a})
by virtue of the skew-symmetry of $\tilde{q}\indices{_{\eta}^{\lambda\alpha\mu}}$ in its last index pair
\begin{align*}
0&=\pfrac{}{x^{\mu}}\left(\tilde{B}\indices{_{\beta}^{\mu}}+
\pfrac{\tilde{q}\indices{_{\eta}^{\lambda\mu\alpha}}}{x^{\alpha}}\gamma\indices{^{\eta}_{\lambda\beta}}+
\tilde{q}\indices{_{\eta}^{\lambda\mu\alpha}}\pfrac{\gamma\indices{^{\eta}_{\lambda\beta}}}{x^{\alpha}}\right)\\
&=\pfrac{\tilde{B}\indices{_{\beta}^{\mu}}}{x^{\mu}}+\ppfrac{\tilde{q}\indices{_{\eta}^{\lambda\mu\alpha}}}{x^{\mu}}{x^{\alpha}}
\gamma\indices{^{\eta}_{\lambda\beta}}+\pfrac{\tilde{q}\indices{_{\eta}^{\lambda\mu\alpha}}}{x^{\alpha}}
\pfrac{\gamma\indices{^{\eta}_{\lambda\beta}}}{x^{\mu}}+
\pfrac{\tilde{q}\indices{_{\eta}^{\lambda\mu\alpha}}}{x^{\mu}}\pfrac{\gamma\indices{^{\eta}_{\lambda\beta}}}{x^{\alpha}}+
\tilde{q}\indices{_{\eta}^{\lambda\mu\alpha}}\ppfrac{\gamma\indices{^{\eta}_{\lambda\beta}}}{x^{\mu}}{x^{\alpha}}\\
&=\pfrac{\tilde{B}\indices{_{\beta}^{\mu}}}{x^{\mu}}.
\end{align*}
Equation~(\ref{eq:noe-cond4}) is thus equivalent to a vanishing divergence of
Eq.~(\ref{eq:noe-cond2a}) as the divergence of its last term vanishes identically.
This demonstrates the consistency of the set of equations~(\ref{eq:noe-cond3}), (\ref{eq:noe-cond1}), (\ref{eq:noe-cond2a}),
and~(\ref{eq:noe-cond4}), which were obtained from the Noether condition~(\ref{eq:gen-conn-infini1}).
\subsection{Amended canonical energy-momentum tensor $\tilde{\Theta}\indices{_{\nu}^{\,\mu}}$\label{sec:dis2}}
We now express Eq.~(\ref{eq:noe-cond2a}) in expanded form
by inserting the local representation of the canonical energy-momentum tensor
$\tilde{B}\indices{_{\beta}^{\alpha}}$ from Eq.~(\ref{eq:B-def}) and the field equation for the partial
divergence of $\tilde{q}\indices{_{\eta}^{\lambda\alpha\mu}}$ from Eq.~(\ref{eq:noe-cond1}):
\begin{align}
\tilde{\Theta}\indices{_{\beta}^{\,\alpha}}\equiv\,&\tilde{B}\indices{_{\beta}^{\alpha}}+\pfrac{}{x^{\mu}}\left(
\tilde{q}\indices{_{\eta}^{\lambda\alpha\mu}}\gamma\indices{^{\eta}_{\lambda\beta}}\right)\nonumber\\
=\,&\tilde{\pi}^{\alpha}\pfrac{\phi}{x^{\beta}}
+\tilde{p}^{\,\xi\alpha}\pfrac{a_{\xi}}{x^{\beta}}-\tilde{p}^{\xi\alpha}a_{\eta}\gamma\indices{^{\eta}_{\xi\beta}}
+\tilde{k}^{\xi\lambda\alpha}\pfrac{g_{\xi\lambda}}{x^{\beta}}
-\tilde{k}^{\xi\lambda\alpha}g_{\eta\lambda}\gamma\indices{^{\eta}_{\xi\beta}}
-\tilde{k}^{\lambda\xi\alpha}g_{\lambda\eta}\gamma\indices{^{\eta}_{\xi\beta}}\nonumber\\
&+\tilde{q}\indices{_{\eta}^{\xi\lambda\alpha}}\left(\pfrac{\gamma\indices{^{\eta}_{\xi\lambda}}}{x^{\beta}}
-\pfrac{\gamma\indices{^{\eta}_{\xi\beta}}}{x^{\lambda}}
+\gamma\indices{^{\eta}_{\tau\beta}}\gamma\indices{^{\tau}_{\xi\lambda}}
-\gamma\indices{^{\eta}_{\tau\lambda}}\gamma\indices{^{\tau}_{\xi\beta}}\right)\nonumber\\
&-\delta_{\beta}^{\alpha}\left(\tilde{\pi}^{\tau}\pfrac{\phi}{x^{\tau}}+
\tilde{p}^{\,\xi\tau}\pfrac{a_{\xi}}{x^{\tau}}+\tilde{k}^{\xi\lambda\tau}\pfrac{g_{\xi\lambda}}{x^{\tau}}+
\tilde{q}\indices{_{\eta}^{\xi\lambda\tau}}\pfrac{\gamma\indices{^{\eta}_{\xi\lambda}}}{x^{\tau}}-\tilde{\HCd}\right)=0.
\label{eq:noe-cond6a}
\end{align}
The term proportional to $\tilde{q}\indices{_{\eta}^{\xi\lambda\alpha}}$ is exactly the Riemann tensor,
defined in the convention of Misner et al.~\cite{misner73} by
\begin{equation}\label{eq:riemann-tensor}
R\indices{^{\eta}_{\xi\beta\lambda}}=\pfrac{\gamma\indices{^{\eta}_{\xi\lambda}}}{x^{\beta}}-
\pfrac{\gamma\indices{^{\eta}_{\xi\beta}}}{x^{\lambda}}+
\gamma\indices{^{\eta}_{\tau\beta}}\gamma\indices{^{\tau}_{\xi\lambda}}-
\gamma\indices{^{\eta}_{\tau\lambda}}\gamma\indices{^{\tau}_{\xi\beta}}.
\end{equation}
We remark that the tensor~(\ref{eq:riemann-tensor}) is actually the \emph{Riemann-Cartan tensor},
as it is defined here from a non-symmetric connection, $\gamma\indices{^{\eta}_{[\xi\lambda]}}\not\equiv0$.
Moreover, the torsion---hence the addressed skew-symmetric part of the connection---does \emph{not}
emerge as a separate dynamic quantitiy in our description as all terms containing the connection
are absorbed into the covariant derivatives and into the Riemann-Cartan tensor.

Equation~(\ref{eq:noe-cond6a}) now writes equivalently after merging the partial derivatives
with the $\gamma$-dependent terms into covariant derivatives:
\begin{align}
\tilde{\Theta}\indices{_{\beta}^{\,\alpha}}=\,&\tilde{\pi}^{\alpha}\pfrac{\phi}{x^{\beta}}
+\tilde{p}^{\,\xi\alpha}\,a_{\xi;\beta}+\tilde{k}^{\xi\lambda\alpha}\,g_{\xi\lambda;\beta}
-\tilde{q}\indices{_{\eta}^{\xi\lambda\alpha}}\,R\indices{^{\eta}_{\xi\lambda\beta}}\nonumber\\
&-\delta_{\beta}^{\alpha}\left(\tilde{\pi}^{\tau}\pfrac{\phi}{x^{\tau}}+
\tilde{p}^{\,\xi\tau}\pfrac{a_{\xi}}{x^{\tau}}+\tilde{k}^{\xi\lambda\tau}\pfrac{g_{\xi\lambda}}{x^{\tau}}+
\tilde{q}\indices{_{\eta}^{\xi\lambda\tau}}\pfrac{\gamma\indices{^{\eta}_{\xi\lambda}}}{x^{\tau}}-\tilde{\HCd}\right)=0.
\label{eq:noe-cond6}
\end{align}
The remaining partial derivatives can similarly be rewritten as tensors
if we subtract the corresponding ``gauge Hamiltonian'' terms from the total Hamiltonian
$\tilde{\HCd}(\tilde{\pi},\phi,\tilde{p},a,\tilde{k},g,\tilde{q},\gamma)$:
\begin{equation}\label{eq:split-ham}
\tilde{\HCd}=\tilde{\HCd}_{0}+\tilde{\HCd}_{\mathrm{G}}+\tilde{\HCd}_{\mathrm{Gr}},
\end{equation}
with the ``gauge Hamiltonian'' $\tilde{\HCd}_{\mathrm{G}}(\tilde{p},a,\tilde{k},g,\tilde{q},\gamma)$ given by
\begin{equation}
\tilde{\HCd}_{\mathrm{G}}=\tilde{p}^{\,\xi\tau}a_{\eta}\,\gamma\indices{^{\eta}_{\xi\tau}}
+\tilde{k}^{\lambda\xi\tau}\,g_{\eta\xi}\,\gamma\indices{^{\eta}_{\lambda\tau}}
+\tilde{k}^{\xi\lambda\tau}\,g_{\xi\eta}\,\gamma\indices{^{\eta}_{\lambda\tau}}
+\tilde{q}\indices{_{\eta}^{\xi\lambda\tau}}\gamma\indices{^{\eta}_{\alpha\lambda}}\gamma\indices{^{\alpha}_{\xi\tau}},
\label{eq:gauge-ham4}
\end{equation}
which agrees with the gauge Hamiltonian derived in Ref.~\cite{struckmeier17a}.
The partial derivatives of the fields in Eq.~(\ref{eq:noe-cond6}) are thus converted into covariant derivatives,
whereas the partial derivative of the connection reemerges as one-half the Riemann-Cartan tensor:
\begin{align}
\tilde{\Theta}\indices{_{\nu}^{\,\mu}}&=\tilde{\pi}^{\,\mu}\pfrac{\phi}{x^{\nu}}+
\tilde{p}^{\,\alpha\mu}a_{\alpha;\nu}+\tilde{k}^{\alpha\beta\mu}g_{\alpha\beta;\nu}-
\tilde{q}\indices{_{\eta}^{\alpha\beta\mu}}R\indices{^{\eta}_{\alpha\beta\nu}}\nonumber\\
&\quad-\delta_{\nu}^{\mu}\left(\tilde{\pi}^{\tau}\pfrac{\phi}{x^{\tau}}+
\tilde{p}^{\,\alpha\tau}a_{\alpha;\tau}+\tilde{k}^{\alpha\beta\tau}g_{\alpha\beta;\tau}-
\onehalf\tilde{q}\indices{_{\eta}^{\alpha\beta\tau}}R\indices{^{\eta}_{\alpha\beta\tau}}-\tilde{\HCd}_{0}-\tilde{\HCd}_{\mathrm{Gr}}\right)=0.
\label{eq:noe-cond6b}
\end{align}
$\tilde{\Theta}\indices{_{\nu}^{\,\mu}}$ obviously represents the canonical
energy-momentum tensor of the closed total system of dynamical fields and spacetime.
The non-existing covariant $x^{\tau}$-derivative of the connection $\gamma\indices{^{\eta}_{\alpha\beta}}$
happens to be replaced by the tensor $-\onehalf R\indices{^{\eta}_{\alpha\beta\tau}}$.
The Hamiltonian $\tilde{\HCd}_{0}\big(\phi,\tilde{\pi},a,\tilde{p},g\big)$ describes the dynamics of the given
system of scalar and vector fields, whereas $\tilde{\HCd}_{\mathrm{Gr}}\big(g,\tilde{k},\tilde{q}\big)$
is supposed to describe the dynamics of the free (uncoupled) gravitational field.
Remarkably, the value of the total energy-momentum tensor $\tilde{\Theta}\indices{_{\nu}^{\,\mu}}$ is zero, which confirms
the conjecture of a \emph{zero-energy universe}~\cite{lorentz1916,levi-civita1917,jordan39,sciama53,feynman62,hawking03}.
Dividing by $\sqrt{-g}$, it can be split into the energy-momentum tensors of the source fields $\theta\indices{_{\nu}^{\,\mu}}$
and of the gravitational field $\vartheta\indices{_{\nu}^{\,\mu}}$:
\begin{align*}
\theta\indices{_\nu^\mu}&=\pi^\mu\pfrac{\phi}{x^\nu}+p^{\alpha\mu}a_{\alpha;\nu}
-\delta_\nu^\mu\left(\pi^\tau\pfrac{\phi}{x^\tau}+p^{\alpha\tau}a_{\alpha;\tau}-\HCd_0\right)\\
\vartheta\indices{_\nu^\mu}&=k^{\alpha\beta\mu}g_{\alpha\beta;\nu}-
q\indices{_\eta^{\alpha\beta\mu}}R\indices{^\eta_{\alpha\beta\nu}}-\delta_\nu^\mu
\left(k^{\alpha\beta\tau}g_{\alpha\beta;\tau}-\onehalf q\indices{_\eta^{\alpha\beta\tau}}R\indices{^\eta_{\alpha\beta\tau}}-\HCd_{\mathrm{Gr}}\right),
\end{align*}
which then succinctly writes as the energy-momentum tensor balance equation
\begin{equation*}
\theta\indices{_\nu^\mu}+\vartheta\indices{_\nu^\mu}=0.
\end{equation*}
\section{Discussion of the field equations in the Lagrangian description\label{sec:lagrange}}
\subsection{Lagrangian representation of the Noether condition~(\ref{eq:noe-cond6b})}
The sum in parentheses of Eq.~(\ref{eq:noe-cond6}) represents the Lagrangian $\tilde{\LCd}$ of the total dynamical system
established by the scalar field, the vector field, the metric, and the connection.
This Lagrangian must be a world scalar density in order for Eq.~(\ref{eq:noe-cond6})
to be a tensor equation, hence to be form-invariant under the $\mathrm{Diff}(M)$ symmetry group.
The Lagrangian $\tilde{\LCd}$ is thus equivalently expressed as
\begin{equation}\label{eq:L-total}
\tilde{\LCd}=\tilde{\pi}^{\tau}\pfrac{\phi}{x^{\tau}}+
\tilde{p}^{\,\alpha\tau}a_{\alpha;\tau}+\tilde{k}^{\alpha\beta\tau}g_{\alpha\beta;\tau}-
\onehalf\tilde{q}\indices{_{\eta}^{\alpha\beta\tau}}R\indices{^{\eta}_{\alpha\beta\tau}}-\tilde{\HCd}_{0}.
\end{equation}
The field equations~(\ref{eq:noe-cond1a}) and~(\ref{eq:noe-cond6b}) will be rewritten in the following on the basis of this Lagrangian.
The canonical momenta are obtained from the Lagrangian~(\ref{eq:L-total}) as
\begin{subequations}\label{eq:trans-lag}
\begin{align}
\tilde{\pi}^{\,\mu}&=\pfrac{\tilde{\LCd}}{\left(\pfrac{\phi}{x^{\mu}}\right)},&
\tilde{p}^{\,\alpha\mu}&=\pfrac{\tilde{\LCd}}{a_{\alpha;\mu}}\\
\tilde{k}^{\alpha\beta\mu}&=\pfrac{\tilde{\LCd}}{g_{\alpha\beta;\mu}},&
-\onehalf\tilde{q}\indices{_{\eta}^{\alpha\beta\mu}}&=\pfrac{\tilde{\LCd}}{R\indices{^{\eta}_{\alpha\beta\mu}}}.
\end{align}
\end{subequations}
The Noether condition~(\ref{eq:noe-cond6b})---hence the assertion of a vanishing canonical energy-momentum
tensor $\tilde{\Theta}\indices{_{\nu}^{\,\mu}}=0$ of the total system of dynamical fields
and spacetime described by $\tilde{\LCd}$ ---is now encountered in the equivalent form:
\begin{equation}
\tilde{\Theta}\indices{_{\nu}^{\,\mu}}=
\pfrac{\tilde{\LCd}}{\left(\pfrac{\phi}{x^{\mu}}\right)}\pfrac{\phi}{x^{\nu}}
+\pfrac{\tilde{\LCd}}{a_{\alpha;\mu}}a_{\alpha;\nu}
+\pfrac{\tilde{\LCd}}{g_{\alpha\beta;\mu}}g_{\alpha\beta;\nu}
+2\pfrac{\tilde{\LCd}}{R\indices{^{\eta}_{\alpha\beta\mu}}}R\indices{^{\eta}_{\alpha\beta\nu}}
-\delta_{\nu}^{\mu}\,\tilde{\LCd}=0.
\label{eq:noe-cond6c}
\end{equation}
We may now split the Lagrangian $\tilde{\LCd}$ of the total system into
a Lagrangian $\tilde{\LCd}_{0}$ for the dynamics of the base fields $\phi$ and $a_{\mu}$,
and a Lagrangian $\tilde{\LCd}_{R}$ for the dynamics of the free gravitational field $R\indices{^{\eta}_{\alpha\beta\nu}}$
according to
\begin{equation*}
\tilde{\LCd}=\tilde{\LCd}_{0}+\tilde{\LCd}_{R},\qquad
\tilde{\LCd}_{0}=\tilde{\LCd}_{0}\big(\phi,\partial\phi,a,\partial a,g,\gamma\big),\qquad
\tilde{\LCd}_{R}=\tilde{\LCd}_{R}\big(\gamma,\partial\gamma,g,\partial g\big),
\end{equation*}
where each Lagrangian represents separately a world scalar density.
In the most general case, $\tilde{\LCd}_{R}$ also describes the dynamics of the covariant derivative
of the metric, $g_{\mu\nu;\xi}\neq0$, which is zero for the case of metric compatibility.
As no derivative with respect to the metric appears in Eq.~(\ref{eq:noe-cond6c}),
we are allowed to divide all terms by $\sqrt{-g}$, whereby the field equation acquires the form
of the \emph{generic Einstein-type equation}:
\begin{equation}\label{eq:noe-cond10}
2\pfrac{\LCd_{R}}{R\indices{^{\eta}_{\alpha\beta\mu}}}R\indices{^{\eta}_{\alpha\beta\nu}}+
\pfrac{\LCd_{R}}{g\indices{_{\alpha\beta;\mu}}}g\indices{_{\alpha\beta;\nu}}-\delta_{\nu}^{\mu}\,\LCd_{R}
=-\pfrac{\LCd_{0}}{\left(\pfrac{\phi}{x^{\mu}}\right)}\pfrac{\phi}{x^{\nu}}-
\pfrac{\LCd_{0}}{a_{\alpha;\mu}}a_{\alpha;\nu}+\delta_{\nu}^{\mu}\LCd_{0}.
\end{equation}
The left-hand side of Eq.~(\ref{eq:noe-cond10}), pertaining to $\LCd_{R}$,
can be regarded as the \emph{canonical} energy-momentum tensor of dynamical spacetime.
With the right-hand side the negative canonical energy-momentum tensor of the system $\LCd_{0}$,
\begin{equation}\label{eq:emt-lag}
\theta\indices{_{\nu}^{\,\mu}}\equiv
\pfrac{\LCd_{0}}{\left(\pfrac{\phi}{x^{\mu}}\right)}\pfrac{\phi}{x^{\nu}}+
\pfrac{\LCd_{0}}{a_{\alpha;\mu}}a_{\alpha;\nu}-\delta_{\nu}^{\mu}\LCd_{0},
\end{equation}
Eq.~(\ref{eq:noe-cond10}) thus establishes a \emph{generally covariant energy-momentum balance relation},
with the coupling of spacetime and source fields induced by both, the metric $g_{\alpha\beta}$
and the connection $\gamma\indices{^{\alpha}_{\xi\tau}}$.
Hence, the energy-momentum tensor of the closed total system $\LCd=\LCd_{0}+\LCd_{R}$
is equal to zero---a result also found by Jordan~\cite{jordan39} and,
independently in each case, by Sciama~\cite{sciama53}, Feynman~\cite[p.~10]{feynman62},
and Hawking~\cite{hawking03}, based on different physical reasoning.
A vanishing total energy-momentum tensor of the universe---taken, by definition, as a closed system---is commonly referred to as the \emph{zero-energy principle}.
As we see now, the zero-energy principle follows directly from the requirement that the action
integral~(\ref{eq:varprinzip1}) be diffeomorphism invariant, which has the physical content that
the laws of physics should take the same form in all reference
frames---which is exactly the gist of the \emph{Principle of General Relativity}.

For the particular case of a covariantly conserved metric,
$g_{\alpha\beta;\nu}\equiv0$, the respective terms in Eq.~(\ref{eq:noe-cond10}) drop out.
This yields the generic Einstein-type equation for the case of \emph{metric compatibility}:
\begin{equation}\label{eq:noe-cond10a}\boxed{
2\pfrac{\LCd_{R}}{R\indices{^{\eta}_{\alpha\beta\mu}}}R\indices{^{\eta}_{\alpha\beta\nu}}-
\delta_{\nu}^{\mu}\,\LCd_{R}=-\theta\indices{_{\nu}^{\mu}}.}
\end{equation}
It applies for all Lagrangians $\LCd_R$ which (i) describe the observed dynamics of the
``free'' (uncoupled) gravitational field and (ii) entail a consistent field equation
with regard to its trace, its symmetries, and its covariant derivatives.
Obviously, the expression
\begin{equation*}
\vartheta\indices{_{\nu}^{\mu}}\equiv2\pfrac{\LCd_{R}}{R\indices{^{\eta}_{\alpha\beta\mu}}}R\indices{^{\eta}_{\alpha\beta\nu}}-
\delta_{\nu}^{\mu}\,\LCd_{R}
\end{equation*}
can be interpreted as the \emph{energy-momentum tensor of spacetime}.
The zero-energy principle from Eq.~(\ref{eq:noe-cond10}) then reduces to $\vartheta\indices{_{\nu}^{\mu}}+\theta\indices{_{\nu}^{\mu}}=0$.

In general, the covariant and contravariant representations of Eq.~(\ref{eq:noe-cond10a})
are not necessarily symmetric in $\nu$ and $\mu$ and thus include a possible spin
of the source field and the then emerging torsion of spacetime.
We thus obtain the following relation of spin and torsion:
\begin{equation}\label{eq:spin-tosion}\boxed{
\pfrac{\LCd_{R}}{R\indices{^{\eta}_{\alpha\beta\nu}}}R\indices{^{\eta}_{\alpha\beta}^{\mu}}-
\pfrac{\LCd_{R}}{R\indices{^{\eta}_{\alpha\beta\mu}}}R\indices{^{\eta}_{\alpha\beta}^{\nu}}=\theta^{[\nu\mu]}.}
\end{equation}
For the Hilbert Lagrangian---and even for Lagrangians with additional quadratic terms in $R\indices{^{\eta}_{\alpha\beta}^{\mu}}$---
the left-hand side of this equation simplifies to the skew-symmetric part of the Ricci tensor
\begin{equation}\label{eq:spin-tosion2}
R^{[\nu\mu]}=8\pi G\,\theta^{[\nu\mu]}.
\end{equation}
For the particular case of the Hilbert Lagrangian $\LCd_{R}=\LCd_{R,H}=-R/16\pi G$, and $\theta^{\nu\mu}$
being the canonical energy-momentum tensor of the Klein-Gordon system---which is symmetric and
coincides with the metric one---Eq.~(\ref{eq:noe-cond10a}) yields the proper Einstein equation,
in conjunction with a vanishing source term for the torsion of spacetime.
We will discuss these issues in detail in Sect.~\ref{sec:sample-lag}.
\subsection{Consistency relation\label{sec:dis1}}
As the second partial derivative of $\tilde{q}\indices{_{\beta}^{\alpha\eta\mu}}$ vanishes identically
by virtue of Eq.~(\ref{eq:noe-cond3}), we conclude from the field equation Eq.~(\ref{eq:noe-cond1}):
\begin{equation*}
\ppfrac{\tilde{q}\indices{_{\beta}^{\alpha\eta\mu}}}{x^{\eta}}{x^{\mu}}\equiv0\qquad\Rightarrow\qquad
\pfrac{}{x^{\eta}}\left(\tilde{p}^{\alpha\eta}\,a_{\beta}+\tilde{k}^{\alpha\lambda\eta}g_{\beta\lambda}
+\tilde{k}^{\lambda\alpha\eta}g_{\lambda\beta}+\tilde{q}\indices{_{\xi}^{\alpha\lambda\eta}}\gamma\indices{^{\xi}_{\beta\lambda}}
-\tilde{q}\indices{_{\beta}^{\xi\lambda\eta}}\gamma\indices{^{\alpha}_{\xi\lambda}}\right)=0.
\end{equation*}
Two terms cancel after inserting the divergence of $\tilde{q}\indices{_{\beta}^{\alpha\eta\mu}}$ from Eq.~(\ref{eq:noe-cond1})
\begin{align}
&\quad\pfrac{\tilde{p}^{\alpha\eta}}{x^{\eta}}a_{\beta}+\tilde{p}^{\alpha\eta}\pfrac{a_{\beta}}{x^{\eta}}+
2\pfrac{\tilde{k}^{\lambda\alpha\eta}}{x^{\eta}}g_{\lambda\beta}+2\tilde{k}^{\lambda\alpha\eta}\pfrac{g_{\lambda\beta}}{x^{\eta}}\nonumber\\
&-\left(\tilde{p}^{\alpha\lambda}\,a_{\xi}+2\tilde{k}^{\tau\alpha\lambda}g_{\tau\xi}-
\tilde{q}\indices{_{\eta}^{\alpha\lambda\tau}}\gamma\indices{^{\eta}_{\xi\tau}}-
\cancel{\tilde{q}\indices{_{\xi}^{\tau\eta\lambda}}\gamma\indices{^{\alpha}_{\tau\eta}}}\,
\right)\gamma\indices{^{\xi}_{\beta\lambda}}+
\tilde{q}\indices{_{\xi}^{\alpha\lambda\eta}}\pfrac{\gamma\indices{^{\xi}_{\beta\lambda}}}{x^{\eta}}\nonumber\\
&+\left(\tilde{p}^{\xi\lambda}\,a_{\beta}+2\tilde{k}^{\tau\xi\lambda}g_{\tau\beta}-
\cancel{\tilde{q}\indices{_{\eta}^{\xi\lambda\tau}}\gamma\indices{^{\eta}_{\beta\tau}}}-
\tilde{q}\indices{_{\beta}^{\eta\tau\lambda}}\gamma\indices{^{\xi}_{\eta\tau}}
\right)\gamma\indices{^{\alpha}_{\xi\lambda}}-
\tilde{q}\indices{_{\beta}^{\xi\lambda\eta}}\pfrac{\gamma\indices{^{\alpha}_{\xi\lambda}}}{x^{\eta}}=0.
\label{eq:noe-cond2b}
\end{align}
The terms in Eq.~(\ref{eq:noe-cond2b}) can now be rearranged such that
the partial derivatives are converted into covariant derivatives:
\begin{align}
&\quad\,\,\:\left(\pfrac{\tilde{p}^{\alpha\eta}}{x^{\eta}}+\tilde{p}^{\xi\eta}\gamma\indices{^{\alpha}_{\xi\eta}}\right)a_{\beta}
+\tilde{p}^{\alpha\eta}\left(\pfrac{a_{\beta}}{x^{\eta}}-a_{\xi}\gamma\indices{^{\xi}_{\beta\eta}}\right)\nonumber\\
&+2\left(\pfrac{\tilde{k}^{\lambda\alpha\eta}}{x^{\eta}}+\tilde{k}^{\xi\alpha\eta}\gamma\indices{^{\lambda}_{\xi\eta}}
+\tilde{k}^{\lambda\xi\eta}\gamma\indices{^{\alpha}_{\xi\eta}}\right)g_{\lambda\beta}
+2\tilde{k}^{\lambda\alpha\eta}\left(\pfrac{g_{\lambda\beta}}{x^{\eta}}-g_{\xi\beta}\gamma\indices{^{\xi}_{\lambda\eta}}
-g_{\lambda\xi}\gamma\indices{^{\xi}_{\beta\eta}}\right)\nonumber\\
&+\tilde{q}\indices{_{\xi}^{\alpha\lambda\eta}}\left(\pfrac{\gamma\indices{^{\xi}_{\beta\lambda}}}{x^{\eta}}
+\gamma\indices{^{\xi}_{\tau\eta}}\gamma\indices{^{\tau}_{\beta\lambda}}\right)
-\tilde{q}\indices{_{\beta}^{\xi\lambda\eta}}\left(\pfrac{\gamma\indices{^{\alpha}_{\xi\lambda}}}{x^{\eta}}
+\gamma\indices{^{\alpha}_{\tau\eta}}\gamma\indices{^{\tau}_{\xi\lambda}}
\vphantom{\pfrac{\gamma\indices{^{\xi}_{\beta\lambda}}}{x^{\eta}}}\right)=0.
\label{eq:noe-cond5}
\end{align}
Due to the skew-symmetry of $\tilde{q}\indices{_{\xi}^{\alpha\lambda\eta}}$ in its last index pair,
the last two terms are one-half the Riemann curvature tensor~(\ref{eq:riemann-tensor}), respectively.
The sums proportional to $a_{\beta}$ and $g_{\lambda\beta}$ in Eq.~(\ref{eq:noe-cond5}) can be expressed
as canonical field equations~\cite{struckmeier17a} of the Hamiltonian $\tilde{\HCd}_{0}=\tilde{\HCd}-\tilde{\HCd}_{\mathrm{G}}$:
\begin{alignat*}{2}
\pfrac{\tilde{p}^{\alpha\eta}}{x^{\eta}}+\tilde{p}^{\xi\eta}\gamma\indices{^{\alpha}_{\xi\eta}}
&=\tilde{p}\indices{^{\alpha\eta}_{;\eta}}-2\tilde{p}^{\alpha\xi}s\indices{^{\eta}_{\xi\eta}}
&&=-\pfrac{\tilde{\HCd}_{0}}{a_{\alpha}}\,=\pfrac{\tilde{\LCd}}{a_{\alpha}}\\
\pfrac{\tilde{k}^{\lambda\alpha\eta}}{x^{\eta}}+\tilde{k}^{\xi\alpha\eta}\gamma\indices{^{\lambda}_{\xi\eta}}
+\tilde{k}^{\lambda\xi\eta}\gamma\indices{^{\alpha}_{\xi\eta}}
&=\tilde{k}\indices{^{\lambda\alpha\eta}_{;\eta}}-2\tilde{k}^{\lambda\alpha\xi}s\indices{^{\eta}_{\xi\eta}}
&&=-\pfrac{\tilde{\HCd}_{0}}{g_{\lambda\alpha}}=\pfrac{\tilde{\LCd}}{g_{\lambda\alpha}}.
\end{alignat*}
With Eqs.~(\ref{eq:trans-lag}), the consistency relation~(\ref{eq:noe-cond5}) follows as
\begin{equation}
\pfrac{\tilde{\LCd}}{a_{\alpha}}a_{\beta}+\pfrac{\tilde{\LCd}}{a_{\alpha;\eta}}a_{\beta;\eta}
+2\pfrac{\tilde{\LCd}}{g_{\lambda\alpha}}g_{\lambda\beta}
+2\pfrac{\tilde{\LCd}}{g_{\lambda\alpha;\eta}}g_{\lambda\beta;\eta}
+\pfrac{\tilde{\LCd}}{R\indices{^{\xi}_{\alpha\lambda\eta}}}R\indices{^{\xi}_{\beta\lambda\eta}}
-\pfrac{\tilde{\LCd}}{R\indices{^{\,\beta}_{\xi\lambda\eta}}}R\indices{^{\alpha}_{\xi\lambda\eta}}=0,
\label{eq:consistency-lag}
\end{equation}
which is exactly the Lagrangian representation of the consistency relation of Ref.~\cite{struckmeier17a}.

For metric compatibility ($g_{\lambda\beta;\eta}\equiv0$), Eq.~(\ref{eq:consistency-lag}) can be split into
two groups, namely the terms depending on the Lagrangian $\LCd_{R}$ for the ``free'' gravitational field
on the left-hand side and the Lagrangian $\LCd_{0}$ of the scalar and vector source fields on the right-hand side.
This yields, after dividing by $\sqrt{-g}$:
\begin{align}
2\pfrac{\LCd_{R}}{g^{\nu\beta}}g^{\mu\beta}
-\pfrac{\LCd_{R}}{R\indices{^{\tau}_{\mu\beta\lambda}}}R\indices{^{\tau}_{\nu\beta\lambda}}
+\pfrac{\LCd_{R}}{R\indices{^{\nu}_{\tau\beta\lambda}}}R\indices{^{\mu}_{\tau\beta\lambda}}-\delta_{\nu}^{\mu}\,\LCd_{R}
=-2\pfrac{\LCd_{0}}{g^{\nu\beta}}g^{\mu\beta}+\pfrac{\LCd_{0}}{a_{\mu;\beta}}a_{\nu;\beta}
+\pfrac{\LCd_{0}}{a_{\mu}}a_{\nu}+\delta_{\nu}^{\mu}\LCd_{0},
\label{eq:consistency-lag2}
\end{align}
where the derivatives of the Lagrangian densities $\tilde{\LCd}=\LCd\sqrt{-g}$ with respect to
the metric are replaced by the corresponding derivatives of the Lagrangians $\LCd$ as
\begin{equation}\label{eq:Lag-Lagden}
\frac{2}{\sqrt{-g}}\pfrac{\tilde{\LCd}}{g^{\nu\beta}}g^{\mu\beta}=2\pfrac{\LCd}{g^{\nu\beta}}g^{\mu\beta}-\delta_{\nu}^{\mu}\LCd,\qquad
\frac{2}{\sqrt{-g}}\pfrac{\tilde{\LCd}}{g_{\mu\beta}}g_{\nu\beta}=2\pfrac{\LCd}{g_{\mu\beta}}g_{\nu\beta}+\delta_{\nu}^{\mu}\LCd.
\end{equation}
\subsection{Equivalence of the Noether condition~(\ref{eq:noe-cond6c}) and the consistency relation~(\ref{eq:consistency-lag})\label{sec:dis3}}
On first sight, the consistency relation~(\ref{eq:consistency-lag}) appears to
emerge in addition to the Noether condition~(\ref{eq:noe-cond6c}).
Yet, as it turns out, both conditions are \emph{equivalent}.
In order to prove this, we sum up both equations:
\begin{align}
&\quad\pfrac{\tilde{\LCd}}{\left(\pfrac{\phi}{x^{\alpha}}\right)}\pfrac{\phi}{x^{\beta}}
+\pfrac{\tilde{\LCd}}{a_{\alpha;\eta}}a_{\beta;\eta}
+\pfrac{\tilde{\LCd}}{a_{\eta;\alpha}}a_{\eta;\beta}+\pfrac{\tilde{\LCd}}{a_{\alpha}}a_{\beta}
+\pfrac{\tilde{\LCd}}{g_{\lambda\eta;\alpha}}g_{\lambda\eta;\beta}
+2\pfrac{\tilde{\LCd}}{g_{\lambda\alpha;\eta}}g_{\lambda\beta;\eta}
+2\pfrac{\tilde{\LCd}}{g_{\lambda\alpha}}g_{\lambda\beta}\nonumber\\
&+2\pfrac{\tilde{\LCd}}{R\indices{^{\xi}_{\lambda\eta\alpha}}}R\indices{^{\xi}_{\lambda\eta\beta}}
+\pfrac{\tilde{\LCd}}{R\indices{^{\xi}_{\alpha\lambda\eta}}}R\indices{^{\xi}_{\beta\lambda\eta}}
-\pfrac{\tilde{\LCd}}{R\indices{^{\,\beta}_{\xi\lambda\eta}}}R\indices{^{\alpha}_{\xi\lambda\eta}}
\equiv\delta_{\beta}^{\alpha}\,\tilde{\LCd}.
\label{eq:consistency-lag12}
\end{align}
According to Eq.~(\ref{eq:assertion3}) of Corollary~\ref{sec:Lagr-case},
the resulting equation~(\ref{eq:consistency-lag12}) represents an \emph{identity}---and
thus does not constitute a dynamical law, i.e., an equation of motion.
Hence, Eq.~(\ref{eq:consistency-lag}) holds if Eq.~(\ref{eq:noe-cond6c}) is satisfied and vice versa.
\subsection{Correlation of spin and torsion\label{sec:dis0}}
The field equation~(\ref{eq:noe-cond1a}) has the Lagrangian representation
\begin{equation}
\left(\pfrac{\tilde{\LCd}}{R\indices{^{\eta}_{\alpha\beta\mu}}}\right)_{;\mu}=\frac{1}{2}
\pfrac{\tilde{\LCd}}{a_{\alpha;\beta}}a_{\eta}+\pfrac{\tilde{\LCd}}{g_{\tau\alpha;\beta}}g_{\tau\eta}
+\pfrac{\tilde{\LCd}}{R\indices{^{\eta}_{\alpha\tau\mu}}}\,s\indices{^{\,\beta}_{\tau\mu}}
+2\pfrac{\tilde{\LCd}}{R\indices{^{\eta}_{\alpha\beta\mu}}}\,s\indices{^{\tau}_{\mu\tau}}.
\end{equation}
Splitting again the total Lagrangian $\tilde{\LCd}$ into the sum of the Lagrangian
$\tilde{\LCd}_{R}$ for the free gravitational field, the Lagrangian $\tilde{\LCd}_{g}$ for
nonmetricity, and the Lagrangian $\tilde{\LCd}_{0}$ describing the dynamics of the scalar and vector field,
this yields
\begin{equation}
\left(\pfrac{\tilde{\LCd}_{R}}{R\indices{^{\eta}_{\alpha\beta\mu}}}\right)_{;\mu}=\frac{1}{2}
\pfrac{\tilde{\LCd}_{0}}{a_{\alpha;\beta}}a_{\eta}+\pfrac{\tilde{\LCd}_{g}}{g_{\tau\alpha;\beta}}g_{\tau\eta}
+\pfrac{\tilde{\LCd}_{R}}{R\indices{^{\eta}_{\alpha\tau\mu}}}\,s\indices{^{\,\beta}_{\tau\mu}}
+2\pfrac{\tilde{\LCd}_{R}}{R\indices{^{\eta}_{\alpha\beta\mu}}}\,s\indices{^{\tau}_{\mu\tau}}.
\end{equation}
For the usual case of a covariantly conserved metric, hence for metric compatibility, we can divide by $\sqrt{-g}$
\begin{equation}\label{eq:div_LR}
\left(\pfrac{\LCd_{R}}{R\indices{_{\eta\alpha\beta\mu}}}\right)_{;\mu}
-\pfrac{\LCd_{R}}{R\indices{_{\eta\alpha\tau\mu}}}\,s\indices{^{\,\beta}_{\tau\mu}}
-2\pfrac{\LCd_{R}}{R\indices{_{\eta\alpha\beta\mu}}}\,s\indices{^{\tau}_{\mu\tau}}
=\pfrac{\tilde{\LCd}_{g}}{g_{\eta\alpha;\beta}}+\frac{1}{2}\pfrac{\LCd_{0}}{a_{\alpha;\beta}}a^{\eta}.
\end{equation}
Equation~(\ref{eq:div_LR}) can now be split into a symmetric part in $\eta$ and $\alpha$
\begin{equation}\label{eq:div_LR-symm}
\pfrac{\tilde{\LCd}_{g}}{g_{\eta\alpha;\beta}}+\frac{1}{4}\left(\pfrac{\LCd_{0}}{a_{\alpha;\beta}}a^{\eta}+
\pfrac{\LCd_{0}}{a_{\eta;\beta}}a^{\alpha}\right)=0,
\end{equation}
and a skew-symmetric part in $\eta$ and $\alpha$
\begin{equation}\label{eq:div_LR-skew}\boxed{
\left(\pfrac{\LCd_{R}}{R\indices{_{\eta\alpha\beta\mu}}}\right)_{;\mu}
-\pfrac{\LCd_{R}}{R\indices{_{\eta\alpha\tau\mu}}}\,s\indices{^{\,\beta}_{\tau\mu}}
-2\pfrac{\LCd_{R}}{R\indices{_{\eta\alpha\beta\mu}}}\,s\indices{^{\tau}_{\mu\tau}}
=\frac{1}{4}\left(\pfrac{\LCd_{0}}{a_{\alpha;\beta}}a^{\eta}-\pfrac{\LCd_{0}}{a_{\eta;\beta}}a^{\alpha}\right).}
\end{equation}
We will show in the next section that the leftmost term,
hence the divergence associated with $\LCd_{R}$, vanishes for the Hilbert Lagrangian $\LCd_{R,H}$.
For that case, Eq.~(\ref{eq:div_LR-skew}) yields an algebraic equation for the torsion emerging from the vector field $a_{\mu}$.
All other choices of $\LCd_{R}$, yield a non-algebraic equation for the correlation of spin and torsion.
So, the question whether spacetime torsion either propagates or is merely tied to spinning matter depends
on the model for the dynamics of the gravitational field in classical vacuum.
In any case, for a non-vanishing right-hand side of Eq.~(\ref{eq:div_LR-skew}), the vector field necessarily
acts as a source of torsion of spacetime.
\section{Sample Lagrangians\label{sec:sample-lag}}
\subsection{Lagrangian $\LCd_{R}$ of the ``free'' gravitational field}
With Eq.~(\ref{eq:noe-cond10}) and the equivalent equation~(\ref{eq:consistency-lag}), we have derived
energy-momentum balance equations for the \emph{interaction} of given source fields---whose free dynamics
is described by $\LCd_{0}$---with the gravitational field, whose free dynamics is described by $\LCd_{R}$.
The Lagrangians $\LCd_{0}$ and $\LCd_{R}$ for the dynamics in the absence of any coupling must be
set up on the basis of physical reasoning.
From analogies to other field theories, we chose a Lagrangian $\LCd_{R}(R,g)$ as a sum of a constant
and of linear and quadratic terms in the Riemann curvature tensor~\cite{struckmeier17a}
\begin{equation}\label{eq:quad-lag}
\LCd_{R}=\frac{1}{4}R\indices{^{\eta}_{\alpha\beta\tau}}\left[g_{1}R\indices{^{\alpha}_{\eta\xi\lambda}}g^{\beta\xi}g^{\tau\lambda}+
\frac{1}{8\pi G}\left(\delta_{\eta}^{\tau}\,g^{\alpha\beta}-\delta_{\eta}^{\beta}\,g^{\alpha\tau}\right)\right]+\frac{\Lambda}{8\pi G}.
\end{equation}
The terms in parentheses can be regarded as the Riemann tensor of the \emph{maximally symmetric} $4$-dimensional manifold
\begin{equation*}
\hat{R}\indices{^{\alpha}_{\eta}^{\beta\tau}}=\frac{1}{8\pi G}
\left(\delta_{\eta}^{\tau}\,g^{\alpha\beta}-\delta_{\eta}^{\beta}\,g^{\alpha\tau}\right),
\end{equation*}
which can be interpreted as the ``ground state'' of spacetime~\cite{carroll13}.
The coupling constant $g_{1}$ is dimensionless, whereas $G$ has the natural dimension of Length$^{2}$.
The Lagrangian~(\ref{eq:quad-lag}) is thus the sum of the Hilbert Lagrangian $\LCd_{R,\mathrm{H}}$
\begin{equation}\label{eq:hilbert-lag}
\LCd_{R,\mathrm{H}}=-\frac{1}{16\pi G}\left(R-2\Lambda\right)
\end{equation}
plus a Lagrangian quadratic in the Riemann tensor.
The latter was proposed earlier by A.~Einstein in a personal letter to H.~Weyl~\cite{einstein18}.

To set up the energy-momentum balance equation~(\ref{eq:noe-cond10a}) for the Lagrangian $\LCd_{R}$, we first calculate
\begin{equation}\label{eq:quad-lag-deri}
\pfrac{\LCd_{R}}{R\indices{^{\eta}_{\alpha\beta\mu}}}=\frac{g_{1}}{2}R\indices{^{\alpha}_{\eta\xi\lambda}}g^{\beta\xi}g^{\mu\lambda}+
\frac{1}{32\pi G}\left(\delta_{\eta}^{\mu}\,g^{\alpha\beta}-\delta_{\eta}^{\beta}\,g^{\alpha\mu}\right)=
\frac{g_{1}}{2}R\indices{^{\alpha}_{\eta}^{\beta\mu}}+\frac{1}{4}\hat{R}\indices{^{\alpha}_{\eta}^{\beta\mu}},
\end{equation}
and hence
\begin{align*}
2\pfrac{\LCd_{R}}{R\indices{^{\eta}_{\alpha\beta\mu}}}R\indices{^{\eta}_{\alpha\beta\nu}}&=
g_{1}R\indices{^{\alpha}_{\eta\xi\lambda}}g^{\beta\xi}g^{\mu\lambda}R\indices{^{\eta}_{\alpha\beta\nu}}+
\frac{1}{16\pi G}\left(\delta_{\eta}^{\mu}\,g^{\alpha\beta}-\delta_{\eta}^{\beta}\,g^{\alpha\mu}\right)R\indices{^{\eta}_{\alpha\beta\nu}}\\
&=-g_{1}R\indices{_{\eta}^{\alpha\beta\mu}}\,R\indices{^{\eta}_{\alpha\beta\nu}}-\frac{1}{8\pi G}R\indices{^{\mu}_{\nu}}.
\end{align*}
The energy-momentum balance equation~(\ref{eq:noe-cond10a}) now follows as
\begin{equation*}
g_{1}R\indices{_{\eta}^{\alpha\beta\mu}}R\indices{^{\eta}_{\alpha\beta\nu}}
+\frac{1}{8\pi G}R\indices{^{\mu}_{\nu}}+\delta_{\nu}^{\mu}\,\LCd_{R}=\theta\indices{_{\nu}^{\mu}},
\end{equation*}
which finally yields the Einstein-type equation after inserting the Lagrangian $\LCd_{R}$ from Eq.~(\ref{eq:quad-lag}):
\begin{equation}
g_{1}\left(R\indices{_{\eta}^{\alpha\beta\mu}}R\indices{^{\eta}_{\alpha\beta\nu}}-
\quarter\delta_{\nu}^{\mu}R\indices{_{\eta}^{\alpha\beta\tau}}R\indices{^{\eta}_{\alpha\beta\tau}}\right)+
\frac{1}{8\pi G}\left(R\indices{^{\mu}_{\nu}}-\onehalf\delta_{\nu}^{\mu} R+\Lambda\delta_{\nu}^{\mu}\right)=\theta\indices{_{\nu}^{\mu}}.
\label{eq:noe-cond8b}
\end{equation}
Note that here the contravariant representations of the energy-momentum tensor $\theta^{\nu\mu}$
as well as that of the Ricci tensor $R^{\mu\nu}$ are not necessarily symmetric.
Thus, Eq.~(\ref{eq:noe-cond8b}) can be split into a symmetric and a skew-symmetric equation in $\mu$ and $\nu$
(see also Kibble~\cite{kibble61} and Sciama~\cite{sciama62}):
\begin{subequations}\label{eq:noe-cond8c}
\begin{align}
g_{1}\left(R\indices{_{\eta}^{\alpha\beta\mu}}R\indices{^{\eta}_{\alpha\beta}^{\nu}}-
\quarter g^{\mu\nu}R\indices{_{\eta}^{\alpha\beta\tau}}R\indices{^{\eta}_{\alpha\beta\tau}}\right)+
\frac{1}{8\pi G}\left(R^{(\mu\nu)}-\onehalf g^{\mu\nu} R+\Lambda g^{\mu\nu}\right)&=\theta^{(\mu\nu)}\label{eq:symm-ricci}\\
\frac{1}{8\pi G}R^{[\mu\nu]}&=\theta^{[\nu\mu]}\label{eq:skew-ricci}.
\end{align}
\end{subequations}
Setting $g_{1}=0$ in Eq.~(\ref{eq:symm-ricci}), hence neglecting the term proportional to  quadratic in the Riemann tensor,
reduces Eq.~(\ref{eq:noe-cond8b}) to the form of the Einstein equation proper.
Equation~(\ref{eq:skew-ricci}) is the representation of the general Eq.~(\ref{eq:spin-tosion})
for the particular Lagrangian $\LCd_{R}$ from Eq.~(\ref{eq:quad-lag}).

The field equation~(\ref{eq:div_LR-skew}), which describes the coupling of spin and torsion,
emerges for the Lagrangian~(\ref{eq:quad-lag}) as the differential equation:
\begin{equation}\label{eq:div_LR_gen}
\frac{g_{1}}{2}R\indices{^{\alpha\eta\beta\mu}_{;\mu}}-
\left(\onehalf g_{1}R\indices{^{\alpha\eta\tau\mu}}+\quarter\hat{R}\indices{^{\alpha\eta\tau\mu}}\right)s\indices{^{\,\beta}_{\tau\mu}}+
\left(g_{1}R\indices{^{\alpha\eta\beta\mu}}+\onehalf\hat{R}\indices{^{\alpha\eta\beta\mu}}\right)s\indices{^{\tau}_{\tau\mu}}
=\quarter\left(\pfrac{\LCd_{0}}{a_{\alpha;\beta}}a^{\eta}-\pfrac{\LCd_{0}}{a_{\eta;\beta}}a^{\alpha}\right).
\end{equation}
For the Hilbert Lagrangian~(\ref{eq:hilbert-lag}), this reduces to the algebraic equation:
\begin{equation}\label{eq:div_LRH}
2\hat{R}\indices{^{\alpha\eta\beta\mu}}s\indices{^{\tau}_{\tau\mu}}
-\hat{R}\indices{^{\alpha\eta\tau\mu}}s\indices{^{\,\beta}_{\tau\mu}}
=\pfrac{\LCd_{0}}{a_{\alpha;\beta}}a^{\eta}-\pfrac{\LCd_{0}}{a_{\eta;\beta}}a^{\alpha}.
\end{equation}
Thus, a massive spin-$1$ particle field $a_{\mu}$ always acts as a source of torsion of spacetime.
The right-hand side will be specified for the Proca system in Sect.~\ref{sec:proca}.
Obviously, the spin-$0$ particle field $\phi$, i.e., the Klein-Gordon system, to be discussed in the following section,
does not act as a source of torsion as Eq.~(\ref{eq:div_LRH}) is identically satisfied for $s\indices{^{\,\beta}_{\tau\mu}}\equiv0$.
\subsection{Klein-Gordon Lagrangian $\LCd_{0}$}
The Klein-Gordon Lagrangian density $\tilde{\LCd}_{0}\big(\phi,\partial_{\nu}\phi,g^{\mu\nu}\big)$
for a system of a real scalar field $\phi$ in a dynamic spacetime is given by
\begin{equation}\label{eq:kg-ham0}
\tilde{\LCd}_{0}=\onehalf\left(\pfrac{\phi}{x^{\alpha}}\pfrac{\phi}{x^{\beta}}\,g^{\alpha\beta}-m^{2}\,\phi^{2}\right)\sqrt{-g}.
\end{equation}
For this Lagrangian, the identity~(\ref{eq:assertion3}) takes on the particular form for a symmetric metric $g^{\alpha\beta}$:
\begin{equation*}
2\pfrac{\tilde{\LCd}_{0}}{g^{\nu\lambda}}g^{\mu\lambda}-\pfrac{\tilde{\LCd}_{0}}{\left(\pfrac{\phi}{x^{\mu}}\right)}\pfrac{\phi}{x^{\nu}}
\equiv-\delta_{\nu}^{\mu}\tilde{\LCd}_{0},
\end{equation*}
hence, dividing by $\sqrt{-g}$:
\begin{equation}\label{eq:KG-identity}
\frac{2}{\sqrt{-g}}\pfrac{\tilde{\LCd}_{0}}{g^{\nu\lambda}}g^{\mu\lambda}\equiv
\pfrac{\LCd_{0}}{\left(\pfrac{\phi}{x^{\mu}}\right)}\pfrac{\phi}{x^{\nu}}-\delta_{\nu}^{\mu}\LCd_{0},
\end{equation}
As $\tilde{\LCd}_{0}$ does not depend on the partial derivative of the metric,
the left-hand side of Eq.~(\ref{eq:KG-identity}) defines the mixed tensor representation of the \emph{metric}
energy-momentum tensor $T\indices{_{\nu}^{\mu}}$ of the Klein-Gordon system~(\ref{eq:kg-ham0}),
whereas the right-hand side represents its \emph{canonical} energy-momentum tensor $\theta\indices{_{\nu}^{\mu}}$:
\begin{equation}\label{eq:homogeneity-KG}
T\indices{_{\nu}^{\mu}}=\frac{2}{\sqrt{-g}}\pfrac{\tilde{\LCd}_{0}}{g^{\nu\lambda}}g^{\mu\lambda},\qquad
\theta\indices{_{\nu}^{\mu}}=\pfrac{\LCd_{0}}{\left(\pfrac{\phi}{x^{\mu}}\right)}\pfrac{\phi}{x^{\nu}}-\delta_{\nu}^{\mu}\LCd_{0}.
\end{equation}
Both tensors thus coincide for this system and have the explicit symmetric contravariant form:
\begin{equation}\label{eq:emt-KG}
T^{\nu\mu}=\pfrac{\phi}{x^{\alpha}}\pfrac{\phi}{x^{\beta}}g^{\alpha\nu}g^{\beta\mu}-\onehalf g^{\nu\mu}\left(
\pfrac{\phi}{x^{\alpha}}\pfrac{\phi}{x^{\beta}}\,g^{\alpha\beta}-m^{2}\,\phi^{2}\right)=\theta^{\nu\mu}=\theta^{(\nu\mu)}.
\end{equation}
The Euler-Lagrange equation for the Lagrangian~(\ref{eq:kg-ham0}) follows for a covariantly conserved metric as
\begin{equation*}
g^{\alpha\beta}\left(\ppfrac{\phi}{x^{\alpha}}{x^{\beta}}-
\gamma\indices{^{\xi}_{\alpha\beta}}\pfrac{\phi}{x^{\xi}}-
2s\indices{^{\xi}_{\alpha\xi}}\pfrac{\phi}{x^{\beta}}\right)+m^{2}\phi=0.
\end{equation*}
The second derivative of $\phi$ as well as the term proportional to $\gamma\indices{^{\xi}_{\alpha\beta}}$ are no tensors.
Yet their sum is just the covariant $x^{\beta}$-derivative of the covector $\partial\phi/\partial x^{\alpha}$,
\begin{equation}\label{eq:eqmo-KG2}
g^{\alpha\beta}\left[{\left(\pfrac{\phi}{x^{\alpha}}\right)}_{;\beta}-2s\indices{^{\xi}_{\alpha\xi}}\pfrac{\phi}{x^{\beta}}\right]+m^{2}\phi=0
\end{equation}
and thus holds as a tensor equation in any reference system.
The term related to the torsion vector $s\indices{^{\xi}_{\alpha\xi}}$ states that the dynamics of the
scalar field $\phi$ also couples to the torsion of spacetime, but does not act as a source of torsion as
the canonical energy-momentum tensor~(\ref{eq:emt-KG}) is symmetric.

For non-zero torsion neither the covariant divergence of the Einstein tensor $G^{\mu\nu}$
nor the covariant divergence of the energy-momentum tensor vanishes.
One thus encounters from the conventional Einstein equation
\begin{equation}
G\indices{^{\alpha}_{\mu;\alpha}}\equiv\left(R\indices{^{\alpha}_{\mu}}-
\onehalf\delta^{\alpha}_{\mu}R\right)_{;\alpha}=8\pi G\,T\indices{^{\alpha}_{\mu;\alpha}},
\end{equation}
the following explicit form for a covariantly conserved metric
\begin{equation}\label{eq:KG-torsion}
\onehalf R\indices{^{\beta\xi}_{\mu\tau}}s\indices{^{\tau}_{\beta\xi}}+
R\indices{^{\tau\beta}_{\tau\xi}}s\indices{^{\xi}_{\beta\mu}}=8\pi G\,
\pfrac{\phi}{x^{\alpha}}g^{\alpha\beta}\left(\pfrac{\phi}{x^{\mu}}s\indices{^{\xi}_{\beta\xi}}+
\pfrac{\phi}{x^{\xi}}s\indices{^{\xi}_{\beta\mu}}\right).
\end{equation}
Equation~(\ref{eq:KG-torsion}) is obviously satisfied for an identically vanishing torsion tensor
$s\indices{^{\xi}_{\beta\mu}}\equiv0$ and thus shows that the spacetime dynamics of the Klein-Gordon
system~(\ref{eq:kg-ham0}) is compatible with an identically vanishing torsion of spacetime.
\subsection{Proca Lagrangian\label{sec:proca}}
The Proca Lagrangian density $\tilde{\LCd}_{0}\big(a_{\mu},\partial_{\nu}a_{\mu},g^{\mu\nu},\gamma\indices{^{\xi}_{\mu\nu}}\big)$ writes
\begin{equation}\label{eq:Proca-Lag}
\tilde{\LCd}_{0}=\left(-\quarter f_{\alpha\beta}\,f_{\xi\eta}\,g^{\alpha\xi}\,g^{\beta\eta}+
\onehalf m^{2}a_{\alpha}\,a_{\xi}\,g^{\alpha\xi}\right)\sqrt{-g},\qquad
f_{\alpha\beta}=a_{\beta;\alpha}-a_{\alpha;\beta}=-f_{\beta\alpha},
\end{equation}
with $f_{\alpha\beta}$ denoting the skew-symmetric field tensor.
With the particular Lagrangian~(\ref{eq:Proca-Lag}), the identity~(\ref{eq:assertion3}) takes on the form:
\begin{equation}\label{eq:proca-identity}
-2\pfrac{\tilde{\LCd}_{0}}{g^{\nu\lambda}}g^{\mu\lambda}+\pfrac{\tilde{\LCd}_{0}}{a_{\mu;\lambda}}a_{\nu;\lambda}+
\pfrac{\tilde{\LCd}_{0}}{a_{\lambda;\mu}}a_{\lambda;\nu}+\pfrac{\tilde{\LCd}_{0}}{a_{\mu}}a_{\nu}\equiv\delta_{\nu}^{\mu}\tilde{\LCd}_{0}.
\end{equation}
Dividing Eq.~(\ref{eq:proca-identity}) by $\sqrt{-g}$, its leftmost term represents the \emph{metric} energy-momentum tensor
$T\indices{_{\nu}^{\mu}}$ from Eq.~(\ref{eq:homogeneity-KG}), as $\tilde{\LCd}_{0}$ does not depend on the derivative of the metric.
By virtue of the identity~(\ref{eq:proca-identity}), $T\indices{_{\nu}^{\mu}}$
can equivalently be obtained from the derivatives with respect to the fields:
\begin{align*}
T\indices{_{\nu}^{\mu}}&=
\pfrac{\LCd_{0}}{a_{\mu;\lambda}}a_{\nu;\lambda}+
\pfrac{\LCd_{0}}{a_{\lambda;\mu}}a_{\lambda;\nu}+
\pfrac{\LCd_{0}}{a_{\mu}}a_{\nu}-\delta_{\nu}^{\mu}\LCd_{0}\\
&=\theta\indices{_{\nu}^{\mu}}+\pfrac{\LCd_{0}}{a_{\mu;\lambda}}a_{\nu;\lambda}+\pfrac{\LCd_{0}}{a_{\mu}}a_{\nu},
\end{align*}
wherein $\theta\indices{_{\nu}^{\mu}}$ denotes the \emph{canonical} energy-momentum tensor
\begin{equation}\label{eq:can-em-tensor}
\theta\indices{_{\nu}^{\mu}}=\pfrac{\LCd_{0}}{a_{\lambda;\mu}}a_{\lambda;\nu}-\delta_{\nu}^{\mu}\LCd_{0},
\end{equation}
and thus
\begin{equation}\label{eq:can-em-tensor-contr}
\theta^{\nu\mu}=T^{\nu\mu}-\pfrac{\LCd_{0}}{a_{\mu;\lambda}}a\indices{^{\nu}_{;\lambda}}-\pfrac{\LCd_{0}}{a_{\mu}}a^{\nu}.
\end{equation}
With
\begin{equation*}
\pfrac{\LCd_{0}}{a_{\mu;\nu}}=f^{\mu\nu},
\end{equation*}
the energy-momentum tensors follow as:
\begin{alignat}{2}
T\indices{_{\nu}^{\mu}}&=f\indices{^{\lambda\mu}}\left(a_{\lambda;\nu}-a_{\nu;\lambda}\right)+
m^{2}\,a_{\nu}\,a^{\mu}&&+\quarter \delta_{\nu}^{\mu}\left(
f_{\alpha\beta}\,f^{\alpha\beta}-2m^{2}a_{\alpha}\,a^{\alpha}\right)\\
\theta\indices{_{\nu}^{\mu}}&=f\indices{^{\lambda\mu}}\,\,a_{\lambda;\nu}&&+\quarter \delta_{\nu}^{\mu}\left(
f_{\alpha\beta}\,f^{\alpha\beta}-2m^{2}a_{\alpha}\,a^{\alpha}\right).\label{eq:can-emt-proca}
\end{alignat}
Our conclusion is that for a Proca system the asymmetric \emph{canonical} energy-momentum
tensor $\theta^{\,\nu\mu}$ represents the correct source term for gravitation.
Remarkably, the tensor $\theta^{\,\nu\mu}$ thus entails an increased weighting of the kinetic energy
over the mass as compared to the \emph{metric} energy momentum tensor $T^{\nu\mu}$ in their roles as the source of gravity.
This holds independently of the particular model for the ``free'' (uncoupled)
gravitational field, whose dynamics is encoded in the Lagrangian $\LCd_{R}$ of Eq.~(\ref{eq:noe-cond10a}).

From the Euler-Lagrange field equation for the vector field $a_{\mu}$,
\begin{equation*}
f\indices{^{\mu\alpha}_{;\alpha}}-2f^{\mu\beta}s\indices{^{\alpha}_{\beta\alpha}}-m^{2}a^{\mu}=0,
\end{equation*}
the covariant divergence of the canonical energy-momentum tensor is obtained as
\begin{equation*}
\theta\indices{_{\nu}^{\mu}_{;\mu}}=f^{\alpha\beta}\left[-a_{\xi}R\indices{^{\xi}_{\alpha\beta\nu}}+
2\left(a_{\alpha;\xi}s\indices{^{\xi}_{\beta\nu}}+a_{\alpha;\nu}s\indices{^{\xi}_{\beta\xi}}\right)\right].
\end{equation*}
This tensor does generally not vanish in a curved spacetime, even if we neglect torsion.
Hence, from
\begin{equation}
G\indices{^{\alpha}_{\nu;\alpha}}\equiv\left(R\indices{^{\alpha}_{\nu}}-
\onehalf\delta^{\alpha}_{\nu}R\right)_{;\alpha}=8\pi G\,\theta\indices{_{\nu}^{\alpha}_{;\alpha}},
\end{equation}
which has the following explicit form for the Proca system,
\begin{equation}\label{eq:Proca-torsion}
\onehalf R\indices{^{\beta\xi}_{\nu\tau}}s\indices{^{\tau}_{\beta\xi}}+
R\indices{^{\tau\beta}_{\tau\xi}}s\indices{^{\xi}_{\beta\nu}}=8\pi G\,f^{\alpha\beta}\left[-a_{\xi}R\indices{^{\xi}_{\alpha\beta\nu}}+
2\left(a_{\alpha;\xi}s\indices{^{\xi}_{\beta\nu}}+a_{\alpha;\nu}s\indices{^{\xi}_{\beta\xi}}\right)\right],
\end{equation}
we conclude that the Einstein equation is consistent only if torsion is included.
In contrast to the corresponding equation for the Klein-Gordon system from Eq.~(\ref{eq:KG-torsion}),
this equation has no solution for a vanishing torsion.

The skew-symmetric part of the canonical energy-momentum tensor~(\ref{eq:can-emt-proca}) of the Proca system follows as
\begin{equation}
\theta_{[\nu\mu]}=\onehalf\left(f_{\nu\beta}\,a\indices{^{\beta}_{;\mu}}-f_{\mu\beta}\,a\indices{^{\beta}_{;\nu}}\right)
=\onehalf\left(a_{\mu;\beta}\,a\indices{^{\beta}_{;\nu}}-a_{\nu;\beta}\,a\indices{^{\beta}_{;\mu}}\right),
\end{equation}
which yields, according to Eq.~(\ref{eq:skew-ricci}), the skew-symmetric part of the Ricci tensor
from the contraction of the generalized first Bianchi identity~\cite{plebanski06}:
\begin{equation}\label{eq:Proca-torsion2}
R_{[\nu\mu]}\equiv s\indices{^{\alpha}_{\alpha\nu;\mu}}-s\indices{^{\alpha}_{\alpha\mu;\nu}}+
s\indices{^{\alpha}_{\nu\mu;\alpha}}-2s\indices{^{\alpha}_{\beta\alpha}}s\indices{^{\beta}_{\nu\mu}}
=4\pi G\left(a_{\mu;\beta}\,a\indices{^{\beta}_{;\nu}}-a_{\nu;\beta}\,a\indices{^{\beta}_{;\mu}}\right).
\end{equation}
Again, this equation satisfied only with a non-vanishing torsion.
\section{Summary and Conclusions\label{sec:conclusions}}
The minimum set of postulates for a theory of spacetime geometry interacting with matter is that
(i) the theory should be derived from an action principle and that (ii) the \emph{Principle of
General Relativity} should hold, i.e., the theory should be diffeomorphism-invariant.
An appropriate basis for the formulation of such a theory is given by Noether's
theorem---which directly follows from the action principle:
it provides for any symmetry of the given action a pertaining conserved Noether current.
Noether's theorem is most efficiently formulated in the Hamiltonian formalism by means
of the generating function of the respective infinitesimal symmetry transformation as this
function directly yields the weakly conserved Noether current.
From the latter, one can then set up the most general field equation of geometrodynamics
for systems with $\mathrm{Diff}(M)$ symmetry.
The covariant field-theoretical version of the canonical transformation formalism is applied
to work out the particular form of the action integral that is maintained under this symmetry.

The general recipe to set up this equation is as follows:
\begin{enumerate}
\item Establish the covariant representation of the canonical energy-momentum tensor for
the given matter Lagrangian $\LCd_{0}$---which must be a \emph{world scalar}.
This means for a system of a scalar field $\phi$ and a massive vector field $a_{\alpha}$
\begin{equation*}
\theta\indices{_{\nu}^{\,\mu}}=
\pfrac{\LCd_{0}}{\left(\pfrac{\phi}{x^{\mu}}\right)}\pfrac{\phi}{x^{\nu}}+
\pfrac{\LCd_{0}}{a_{\alpha;\mu}}a_{\alpha;\nu}-\delta_{\nu}^{\mu}\LCd_{0}.
\end{equation*}
The covariant form of the canonical energy-momentum tensor thus contains \emph{direct}
coupling terms of the vector field $a_{\alpha}$ and the connection $\gamma\indices{^{\alpha}_{\mu\nu}}$,
and thereby also causes a coupling to a torsion of spacetime.
For vector fields which represent the classical limit of massive spin
particles, the correct source of gravitation is constituted by the canonical energy-momentum
tensor and \emph{not} by the conventionally used metric (Hilbert) energy-momentum tensor---in agreement with Hehl~\cite{hehl76b}.
The source term changes for systems with additional symmetries---such as a
system with additional $\mathrm{U}(1)$ symmetry, in which case the metric
(Hilbert) energy-momentum tensor turns out to be the appropriate source term~\cite{struckmeier17b}.
\item With the source term and the \emph{postulated}
Lagrangians for both, the dynamics of the ``free'' (uncoupled) connection and metric, $\LCd_{R}$,
the new and most general equation of geometrodynamics for scalar and vector source fields is given by:
\begin{equation*}
2\pfrac{\LCd_{R}}{R\indices{^{\eta}_{\alpha\beta\mu}}}R\indices{^{\eta}_{\alpha\beta\nu}}+
\pfrac{\LCd_{R}}{g_{\alpha\beta;\mu}}g_{\alpha\beta;\nu}-
\delta_{\nu}^{\mu}\,\LCd_{R}=-\theta\indices{_{\nu}^{\mu}}.
\end{equation*}
On first sight, $\LCd_{R}$ may be any world scalar formed out of the
Riemann tensor $R\indices{^{\mu}_{\tau\alpha\beta}}$ and the metric $g^{\mu\nu}$ and its covariant derivative.
Yet, as the Noether current is merely \emph{weakly} conserved, the choice of $\LCd_{R}$ is actually restricted
by the requirement that the subsequent field equation is consistent with regard to its trace and its covariant divergence.
\item In the particular case of metric compatibility, hence for a covariantly conserved metric,
the connection and the metric are correlated according to
\begin{equation*}
g_{\alpha\beta;\nu}\equiv\pfrac{g_{\alpha\beta}}{x^{\nu}}-g_{\tau\beta}\gamma\indices{^{\tau}_{\alpha\nu}}-
g_{\alpha\tau}\gamma\indices{^{\tau}_{\beta\nu}}=0,
\end{equation*}
and the correlation of the Riemann tensor $R\indices{^{\eta}_{\alpha\beta\mu}}$
to the source simplifies to the following form of a generic Einstein-type equation:
\begin{equation*}
\vartheta\indices{_{\nu}^{\mu}}\equiv 2\pfrac{\LCd_{R}}{R\indices{^{\eta}_{\alpha\beta\mu}}}R\indices{^{\eta}_{\alpha\beta\nu}}-
\delta_{\nu}^{\mu}\,\LCd_{R}=-\theta\indices{_{\nu}^{\mu}}.
\end{equation*}
The left-hand side associated with the Lagrangian $\LCd_{R}$ describing the dynamics of the gravitational
field in classical vacuum, can be interpreted as the covariant canonical energy-momentum tensor of spacetime,
which balances the canonical energy-momentum tensor of matter on the right-hand side.
$\vartheta\indices{_{\nu}^{\mu}}+\theta\indices{_{\nu}^{\mu}}=0$ represents the zero energy principle.
\end{enumerate}
The simplest case for zero torsion is given by the Hilbert Lagrangian~(\ref{eq:hilbert-lag}),
which directly yields the Einstein tensor on the left-hand side.
This requires the energy-momentum tensor to be symmetric as well as its covariant
divergence  to be zero in order for the resulting field equation to be consistent.

Summarizing, our generic theory of geometrodynamics generalizes Einstein's General Relativity as follows:
\begin{enumerate}
\item The description of the dynamics of the ``free'' gravitational fields is not restricted to the Hilbert Lagrangian.
As was shown by Hayashi and Shirafuji~\cite{hayashi80}, the most general Lagrangian
compatible with also parity invariance can be at most quadratic in the Riemann-Cartan curvature tensor.

In the case of a quadratic and linear dependence of $\LCd_{R}(R,g)$ on the Riemann tensor, the field equation
\begin{equation*}
g_{1}\left(R\indices{_{\eta}^{\alpha\beta\mu}}R\indices{^{\eta}_{\alpha\beta\nu}}-
\quarter\delta_{\nu}^{\mu}R\indices{_{\eta}^{\alpha\beta\tau}}R\indices{^{\eta}_{\alpha\beta\tau}}\right)+
\frac{1}{8\pi G}\left(R\indices{^{\mu}_{\nu}}-\onehalf\delta_{\nu}^{\mu} R+
\Lambda\delta_{\nu}^{\mu}\right)=\theta\indices{_{\nu}^{\mu}},
\end{equation*}
is encountered~\cite{struckmeier17a}.
It equally complies with the Principle of General Relativity.
The additional term proportional to the \emph{dimensionless} coupling constant $g_{1}$
is equally satisfied by the Schwarzschild and the Kerr metric~\cite{kehm17} in the absence of torsion.
Yet, it entails a different description of the dynamics of spacetime in the case of a
non-vanishing source term $\theta\indices{_{\nu}^{\mu}}$ as compared to the solution
based on only the Einstein tensor---which follows setting $g_{1}=0$.
\item The generalized theory is not restricted to a covariantly conserved metric, hence to \emph{metric compatibility}.
\item The spacetime is not assumed to be generally torsion-free.
Based on the Riemann-Cartan curvature tensor, the generalized theory allows for sources of gravity which generate and couple to a torsion of spacetime.
This applies in particular to those vector fields, which represent the classical limit of massive spin-$1$ particles.
For this case, the \emph{canonical} energy-momentum tensor is the appropriate source term.
Its skew-symmetric part is then related to the skew-symmetric part of the then non-symmetric Ricci tensor according to
\begin{equation*}
R^{[\nu\mu]}=8\pi G\,\theta^{[\nu\mu]}.
\end{equation*}
This equation states that a skew-symmetric part of the canonical energy-momentum
tensor is necessarily associated with a non-vanishing torsion of spacetime.
The corresponding additional degrees of freedom are encoded in a Riemann tensor that is \emph{not symmetric}
under exchange of its first and second index pair, which gives rise to a non-symmetric Ricci tensor.

Moreover, for the case of the Hilbert Lagrangian, hence for $g_{1}=0$, one encounters the following field
equation from the covariant derivatives of the Einstein- and the energy-momentum tensors of the Proca system:
\begin{equation*}
\onehalf R\indices{^{\beta\xi}_{\nu\tau}}s\indices{^{\tau}_{\beta\xi}}+
R\indices{^{\tau\beta}_{\tau\xi}}s\indices{^{\xi}_{\beta\nu}}=8\pi G\,f^{\alpha\beta}\left[-a_{\xi}R\indices{^{\xi}_{\alpha\beta\nu}}+
2\left(a_{\alpha;\xi}s\indices{^{\xi}_{\beta\nu}}+a_{\alpha;\nu}s\indices{^{\xi}_{\beta\xi}}\right)\right],
\end{equation*}
which necessarily gives rise to a non-vanishing torsion.
\end{enumerate}
\ack
First of all, we want to remember our revered academic teacher Walter Greiner, whose charisma and
passion for physics inspired us to stay engaged in physics for all of our lives.

The authors thank Patrick Liebrich and Julia Lienert (Goethe University Frankfurt am Main and FIAS),
and Horst Stoecker (FIAS, GSI, and Goethe University Frankfurt am Main) for valuable discussions.
D.V.\ and J.K. thank the C.~W.~Fueck Foundation for its support.
\appendix
\section{Identity for a scalar-valued function $S$ of an $(n,m)$-tensor $T$ and the metric}
\begin{proposition}
Let $S=S(g,T)\in\RB$ be a scalar-valued function constructed from the metric tensor $g_{\mu\nu}$
and an $(n,m)$-tensor $T\indices{^{\xi_{1}\ldots\xi_{n}}_{\eta_{1}\ldots\eta_{m}}}$. 
Then the following identity holds:
\begin{align}
\pfrac{S}{g_{\mu\beta}}g_{\nu\beta}+\pfrac{S}{g_{\beta\mu}}g_{\beta\nu}&-
\pfrac{S}{T\indices{^{\nu\,\xi_{2}\ldots\xi_{n}}_{\eta_{1}\ldots\eta_{m}}}}
T\indices{^{\mu\,\xi_{2}\ldots\xi_{n}}_{\eta_{1}\ldots\eta_{m}}}-\ldots-
\pfrac{S}{T\indices{^{\xi_{1}\ldots\xi_{n-1}\,\nu}_{\eta_{1}\ldots\eta_{m}}}}
T\indices{^{\xi_{1}\ldots\xi_{n-1}\,\mu}_{\eta_{1}\ldots\eta_{m}}}\nonumber\\
&+\pfrac{S}{T\indices{^{\xi_{1}\ldots\xi_{n}}_{\mu\,\eta_{2}\ldots\eta_{m}}}}
T\indices{^{\xi_{1}\ldots\xi_{n}}_{\nu\,\eta_{2}\ldots\eta_{m}}}+\ldots+
\pfrac{S}{T\indices{^{\xi_{1}\ldots\xi_{n}}_{\eta_{1}\ldots\eta_{m-1}\,\mu}}}
T\indices{^{\xi_{1}\ldots\xi_{n}}_{\eta_{1}\ldots\eta_{m-1}\,\nu}}\equiv0\,\delta_{\nu}^{\mu}.
\label{eq:assertion1}
\end{align}
\end{proposition}
\begin{proof}
The induction hypothesis is immediately verified for scalars constructed from second rank tensors and, if necessary, the metric,
hence, $S=T\indices{^{\alpha}_{\alpha}}$, $S=T^{\alpha\beta}\,g_{\alpha\beta}$, and $S=T_{\alpha\beta}\,g^{\alpha\beta}$.
Let Eq.~(\ref{eq:assertion1}) hold for an $(n,m)$-tensor $T\indices{^{\xi_{1}\ldots\xi_{n}}_{\eta_{1}\ldots\eta_{m}}}$.
We first consider an $(n+1,m+1)$-tensor $\bar{T}\indices{^{\xi_{1}\ldots\xi_{n}\alpha}_{\eta_{1}\ldots\eta_{m}\alpha}}$
with the last indices contracted in order to recover a scalar.
Setting up $S$ according to~(\ref{eq:assertion1}) with the tensor $\bar{T}$, one encounters the two additional terms
\begin{align*}
&\quad-\pfrac{S}{\bar{T}\indices{^{\xi_{1}\ldots\xi_{n}\nu}_{\eta_{1}\ldots\eta_{m}\alpha}}}
\bar{T}\indices{^{\xi_{1}\ldots\xi_{n}\mu}_{\eta_{1}\ldots\eta_{m}\alpha}}+
\pfrac{S}{\bar{T}\indices{^{\xi_{1}\ldots\xi_{n}\,\alpha}_{\eta_{1}\ldots\eta_{m}\mu}}}
\bar{T}\indices{^{\xi_{1}\ldots\xi_{n}\,\alpha}_{\eta_{1}\ldots\eta_{m}\nu}}\\
&=-\delta_{\nu}^{\alpha}\,\bar{T}\indices{^{\xi_{1}\ldots\xi_{n}\mu}_{\eta_{1}\ldots\eta_{m}\alpha}}+
\delta_{\alpha}^{\mu}\,\bar{T}\indices{^{\xi_{1}\ldots\xi_{n}\,\alpha}_{\eta_{1}\ldots\eta_{m}\nu}}=0.
\end{align*}
Equation~(\ref{eq:assertion1}) thus also holds for the scalar $S$ formed from the $(n+1,m+1)$-tensor $\bar{T}$.

For the case that $\bar{T}$ represents an $(n+2,m)$-tensor $\bar{T}\indices{^{\xi_{1}\ldots\xi_{n}\alpha\beta}_{\eta_{1}\ldots\eta_{m}}}$,
the scalar $S$ must have one additional factor $g_{\alpha\beta}$.
One thus encounters four additional terms:
\begin{align*}
&\quad-\pfrac{S}{\bar{T}\indices{^{\xi_{1}\ldots\xi_{n}\nu\beta}_{\eta_{1}\ldots\eta_{m}}}}
\bar{T}\indices{^{\xi_{1}\ldots\xi_{n}\mu\beta}_{\eta_{1}\ldots\eta_{m}}}-
\pfrac{S}{\bar{T}\indices{^{\xi_{1}\ldots\xi_{n}\,\alpha\nu}_{\eta_{1}\ldots\eta_{m}}}}
\bar{T}\indices{^{\xi_{1}\ldots\xi_{n}\,\alpha\mu}_{\eta_{1}\ldots\eta_{m}}}+
\pfrac{S}{g_{\mu\beta}}g_{\nu\beta}+\pfrac{S}{g_{\alpha\mu}}g_{\alpha\nu}\\
&=-\delta_{\nu}^{\alpha}\,\bar{T}\indices{^{\xi_{1}\ldots\xi_{n}\mu\beta}_{\eta_{1}\ldots\eta_{m}}}g_{\alpha\beta}-
\delta_{\nu}^{\beta}\,\bar{T}\indices{^{\xi_{1}\ldots\xi_{n}\,\alpha\mu}_{\eta_{1}\ldots\eta_{m}}}g_{\alpha\beta}+
\delta_{\alpha}^{\mu}g_{\nu\beta}\bar{T}\indices{^{\xi_{1}\ldots\xi_{n}\alpha\beta}_{\eta_{1}\ldots\eta_{m}}}+
\delta_{\beta}^{\mu}g_{\alpha\nu}\bar{T}\indices{^{\xi_{1}\ldots\xi_{n}\alpha\beta}_{\eta_{1}\ldots\eta_{m}}}\\
&=-g_{\nu\beta}\bar{T}\indices{^{\xi_{1}\ldots\xi_{n}\mu\beta}_{\eta_{1}\ldots\eta_{m}}}-
g_{\alpha\nu}\bar{T}\indices{^{\xi_{1}\ldots\xi_{n}\,\alpha\mu}_{\eta_{1}\ldots\eta_{m}}}+
g_{\nu\beta}\bar{T}\indices{^{\xi_{1}\ldots\xi_{n}\mu\beta}_{\eta_{1}\ldots\eta_{m}}}+
g_{\alpha\nu}\bar{T}\indices{^{\xi_{1}\ldots\xi_{n}\alpha\mu}_{\eta_{1}\ldots\eta_{m}}}\\
&=0.
\end{align*}
For the case that $\bar{T}$ represents an $(n,m+2)$-tensor $\bar{T}\indices{^{\xi_{1}\ldots\xi_{n}}_{\eta_{1}\ldots\eta_{m}\alpha\beta}}$,
the scalar $S$ must have one additional factor $g^{\alpha\beta}$.
Owing to
\begin{equation*}
\pfrac{S}{g_{\mu\beta}}g_{\nu\beta}=-\pfrac{S}{g^{\alpha\nu}}g^{\alpha\mu},\qquad
\pfrac{S}{g_{\alpha\mu}}g_{\alpha\nu}=-\pfrac{S}{g^{\nu\beta}}g^{\mu\beta},
\end{equation*}
one thus encounters the four additional terms:
\begin{align*}
&\quad\pfrac{S}{\bar{T}\indices{^{\xi_{1}\ldots\xi_{n}}_{\eta_{1}\ldots\eta_{m}\mu\beta}}}
\bar{T}\indices{^{\xi_{1}\ldots\xi_{n}}_{\eta_{1}\ldots\eta_{m}\nu\beta}}+
\pfrac{S}{\bar{T}\indices{^{\xi_{1}\ldots\xi_{n}}_{\eta_{1}\ldots\eta_{m}\,\alpha\mu}}}
\bar{T}\indices{^{\xi_{1}\ldots\xi_{n}}_{\eta_{1}\ldots\eta_{m}\,\alpha\nu}}-
\pfrac{S}{g^{\alpha\nu}}g^{\alpha\mu}-\pfrac{S}{g^{\nu\beta}}g^{\mu\beta}\\
&=\delta_{\alpha}^{\mu}\,\bar{T}\indices{^{\xi_{1}\ldots\xi_{n}}_{\eta_{1}\ldots\eta_{m}\nu\beta}}g^{\alpha\beta}+
\delta_{\beta}^{\mu}\,\bar{T}\indices{^{\xi_{1}\ldots\xi_{n}}_{\eta_{1}\ldots\eta_{m}\,\alpha\nu}}g^{\alpha\beta}-
\delta_{\nu}^{\beta}g^{\alpha\mu}\bar{T}\indices{^{\xi_{1}\ldots\xi_{n}}_{\eta_{1}\ldots\eta_{m}\alpha\beta}}-
\delta_{\nu}^{\alpha}g^{\mu\beta}\bar{T}\indices{^{\xi_{1}\ldots\xi_{n}}_{\eta_{1}\ldots\eta_{m}\alpha\beta}}\\
&=g^{\mu\beta}\bar{T}\indices{^{\xi_{1}\ldots\xi_{n}}_{\eta_{1}\ldots\eta_{m}\nu\beta}}+
g^{\alpha\mu}\bar{T}\indices{^{\xi_{1}\ldots\xi_{n}}_{\eta_{1}\ldots\eta_{m}\,\alpha\nu}}-
g^{\alpha\mu}\bar{T}\indices{^{\xi_{1}\ldots\xi_{n}}_{\eta_{1}\ldots\eta_{m}\alpha\nu}}-
g^{\mu\beta}\bar{T}\indices{^{\xi_{1}\ldots\xi_{n}}_{\eta_{1}\ldots\eta_{m}\nu\beta}}\\
&=0.
\end{align*}
$\hfill\square$
\end{proof}
The derivative of a Lagrangian $\LCd$ with respect to the metric $g_{\mu\nu}$
can thus always be replaced by the derivatives with respect to the appertaining tensors $T$ that are
made into a scalar by means of the metric.
The identity thus provides the correlation of the \emph{metric} and the \emph{canonical}
energy-momentum tensors of a given system.
\begin{corollary}
The contraction of Eq.~(\ref{eq:assertion1}) then yields a condition for the scalar $S$:
\begin{equation}\label{eq:assertion-contr}
\pfrac{S}{g_{\alpha\beta}}g_{\alpha\beta}\equiv\frac{n-m}{2}\pfrac{S}{T\indices{^{\xi_{1}\ldots\xi_{n}}_{\eta_{1}\ldots\eta_{m}}}}
T\indices{^{\xi_{1}\ldots\xi_{n}}_{\eta_{1}\ldots\eta_{m}}}.
\end{equation}
\end{corollary}
\begin{proof}
Contracting Eq.~(\ref{eq:assertion1}) directly yields Eq.~(\ref{eq:assertion-contr}).
\end{proof}
\begin{corollary}\label{sec:equiv}
Let a \emph{relative scalar of weight} $w$---denoted by $\tilde{S}=\tilde{S}(g,T_k)\in\RB$---be
given as a function of the (symmetric) metric $g_{\mu\nu}$ and a sum of $k$ tensors
$T\indices{_k^{\xi_{1}\ldots\xi_{n_k}}_{\eta_{1}\ldots\eta_{m_k}}}$ of respective rank $(n_k,m_k)$.
Then the following identity holds:
\begin{align}
2\pfrac{\tilde{S}}{g_{\mu\beta}}g_{\nu\beta}&-
\pfrac{\tilde{S}}{T\indices{_k^{\nu\,\xi_{2}\ldots\xi_{n_k}}_{\eta_{1}\ldots\eta_{m_k}}}}
T\indices{_k^{\mu\,\xi_{2}\ldots\xi_{n_k}}_{\eta_{1}\ldots\eta_{m_k}}}-\ldots-
\pfrac{\tilde{S}}{T\indices{_k^{\xi_{1}\ldots\xi_{n_k-1}\,\nu}_{\eta_{1}\ldots\eta_{m_k}}}}
T\indices{_k^{\xi_{1}\ldots\xi_{n_k-1}\,\mu}_{\eta_{1}\ldots\eta_{m_k}}}\nonumber\\
&+\pfrac{\tilde{S}}{T\indices{_k^{\xi_{1}\ldots\xi_{n_k}}_{\mu\,\eta_{2}\ldots\eta_{m_k}}}}
T\indices{_k^{\xi_{1}\ldots\xi_{n_k}}_{\nu\,\eta_{2}\ldots\eta_{m_k}}}+\ldots+
\pfrac{\tilde{S}}{T\indices{_k^{\xi_{1}\ldots\xi_{n_k}}_{\eta_{1}\ldots\eta_{m_k-1}\,\mu}}}
T\indices{_k^{\xi_{1}\ldots\xi_{n_k}}_{\eta_{1}\ldots\eta_{m_k-1}\,\nu}}\equiv\delta_{\nu}^{\mu}\,w\,\tilde{S}.
\label{eq:assertion2}
\end{align}
\end{corollary}
\begin{proof}
Multiply Eq.~(\ref{eq:assertion1}) with ${\left(\sqrt{-g}\right)}^w$, add $\delta_{\nu}^{\mu}\,w\,\tilde{S}$
on both sides of the identity, and combine the appropriate terms on the left-hand side according to Eq.~(\ref{eq:Lag-Lagden}):
\begin{equation*}
2\pfrac{\tilde{S}}{g_{\mu\beta}}g_{\nu\beta}=2\pfrac{S}{g_{\mu\beta}}g_{\nu\beta}\,{\left(\sqrt{-g}\right)}^w+\delta_{\nu}^{\mu}\,w\,\tilde{S}.
\end{equation*}
\end{proof}
Equation~(\ref{eq:assertion2}) is obviously a representation of Euler's theorem on homogeneous functions
in the realm of tensor calculus.
\begin{corollary}\label{sec:Lagr-case}
For a \emph{scalar density} Lagrangian $\tilde{\LCd}$, i.e.\ for a relative scalar of weight $w=1$, the identity~(\ref{eq:assertion2}) writes:
\begin{align}
-2\pfrac{\tilde{\LCd}}{g^{\nu\beta}}g^{\mu\beta}&-
\pfrac{\tilde{\LCd}}{T\indices{_k^{\nu\,\xi_{2}\ldots\xi_{n_k}}_{\eta_{1}\ldots\eta_{m_k}}}}
T\indices{_k^{\mu\,\xi_{2}\ldots\xi_{n_k}}_{\eta_{1}\ldots\eta_{m_k}}}-\ldots-
\pfrac{\tilde{\LCd}}{T\indices{_k^{\xi_{1}\ldots\xi_{n_k-1}\,\nu}_{\eta_{1}\ldots\eta_{m_k}}}}
T\indices{_k^{\xi_{1}\ldots\xi_{n_k-1}\,\mu}_{\eta_{1}\ldots\eta_{m_k}}}\nonumber\\
&+\pfrac{\tilde{\LCd}}{T\indices{_k^{\xi_{1}\ldots\xi_{n_k}}_{\mu\,\eta_{2}\ldots\eta_{m_k}}}}
T\indices{_k^{\xi_{1}\ldots\xi_{n_k}}_{\nu\,\eta_{2}\ldots\eta_{m_k}}}+\ldots+
\pfrac{\tilde{\LCd}}{T\indices{_k^{\xi_{1}\ldots\xi_{n_k}}_{\eta_{1}\ldots\eta_{m_k-1}\,\mu}}}
T\indices{_k^{\xi_{1}\ldots\xi_{n_k}}_{\eta_{1}\ldots\eta_{m_k-1}\,\nu}}\equiv\delta_{\nu}^{\mu}\,\tilde{\LCd}.
\label{eq:assertion3}
\end{align}
\end{corollary}
\begin{corollary}\label{sec:Ham-case}
For a \emph{scalar density} Hamiltonian $\tilde{\HCd}$, the momentum fields $\tilde{p}$ representing tensor densities
must be expressed as absolute tensors $p=\tilde{p}/\sqrt{-g}$ prior to setting up the invariant according to Eq.~(\ref{eq:assertion3}).
An example is worked out in~\ref{ex:proca-ham}.
\end{corollary}
\section{Examples of the identity~(\ref{eq:assertion3}) for scalar-valued functions of tensors}
\subsection{Determinant of the covariant metric tensor $(g_{\mu\nu})$}
The components of the covariant metric tensor $(g_{\mu\nu}(x))$ transform under the transition $x\mapsto X$ of the space-time location as:
\vspace*{-2mm}
\begin{equation*}
g_{\mu\nu}(X)=g_{\alpha\beta}(x)\pfrac{x^\alpha}{X^\mu}\pfrac{x^\beta}{X^\nu},
\end{equation*}
and hence its determinant $g\equiv\left|g_{\alpha\beta}(x)\right|$:
\begin{equation*}
\left|g_{\mu\nu}(X)\right|=\left|g_{\alpha\beta}(x)\right|{\left|\pfrac{x}{X}\right|}^2.
\end{equation*}
The determinant $g$ of the covariant metric tensor thus transforms as a \emph{relative scalar} of weight $w=2$.
According to the general form of the identity for relative scalars of weight $w$ from Eq.~(\ref{eq:assertion2}), we get for $g$
due to the symmetry $g_{\beta\alpha}=g_{\alpha\beta}$:
\begin{equation}\label{eq:metric-identity}
\pfrac{g}{g_{\beta\mu}}g_{\beta\alpha}\equiv\delta_\alpha^\mu\,g.
\end{equation}
Contracting~(\ref{eq:metric-identity}) with the inverse metric $g^{\alpha\nu}$ yields the derivative of the determinant
$g$ of the covariant metric with respect to the component $g_{\nu\mu}$ of the metric:
\begin{equation*}
\pfrac{g}{g_{\nu\mu}}\equiv g^{\mu\nu}\,g,
\end{equation*}
and thus for the negative square root of $g$
\begin{equation}\label{eq:metric-identity2}
\pfrac{\sqrt{-g}}{g_{\nu\mu}}\equiv \onehalf g^{\mu\nu}\sqrt{-g}.
\end{equation}
\subsection{Contraction of a rank-$4$ tensor density and two absolute rank-$2$ tensors}
Let $\tilde{S}\equiv S\sqrt{-g}$ be a scalar density, i.e., a relative scalar of weight $w=1$,
emerging from the contraction of an arbitrary tensor density $\tilde{A}\indices{^\alpha^\beta_\xi}\equiv A\indices{^\alpha^\beta_\xi}\sqrt{-g}$
with absolute ($w=0$) tensors $B\indices{_\alpha^\xi^\eta}$ and $C\indices{_\beta_\eta}$:
\begin{equation}
\tilde{S}=\tilde{A}\indices{^\alpha^\beta_\xi}\,B\indices{_\alpha^\xi^\eta}\,C\indices{_\beta_\eta}.
\end{equation}
Then
\begin{align}
&-\pfrac{\tilde{S}}{\tilde{A}\indices{^{\nu\beta}_{\xi}}}\tilde{A}\indices{^{\mu\beta}_{\xi}}
-\pfrac{\tilde{S}}{\tilde{A}\indices{^{\alpha\nu}_{\xi}}}\tilde{A}\indices{^{\alpha\mu}_{\xi}}
+\pfrac{\tilde{S}}{\tilde{A}\indices{^{\alpha\beta}_{\mu}}}\tilde{A}\indices{^{\alpha\beta}_{\nu}}
-\pfrac{\tilde{S}}{B\indices{_\alpha^\nu^\eta}}B\indices{_\alpha^\mu^\eta}
-\pfrac{\tilde{S}}{B\indices{_\alpha^\xi^\nu}}B\indices{_\alpha^\xi^\mu}\nonumber\\
&+\pfrac{\tilde{S}}{B\indices{_\mu^\xi^\eta}}B\indices{_\nu^\xi^\eta}
+\pfrac{\tilde{S}}{C\indices{_\mu_\eta}}C\indices{_\nu_\eta}
+\pfrac{\tilde{S}}{C\indices{_\beta_\mu}}C\indices{_\beta_\nu}
+2\pfrac{\tilde{S}}{g_{\mu\beta}}g_{\nu\beta}\equiv\delta_\nu^\mu\,\tilde{S}.
\label{eq:arb-scal}
\end{align}
To prove the identity~(\ref{eq:arb-scal}), the respective terms of the sum are worked out explicitly:
\begin{align*}
-\pfrac{\tilde{S}}{\tilde{A}\indices{^{\nu\beta}_{\xi}}}\tilde{A}\indices{^{\mu\beta}_{\xi}}&=
-\delta_\nu^\alpha\,\tilde{A}\indices{^{\mu\beta}_{\xi}}\,B\indices{_\alpha^\xi^\eta}\,C\indices{_\beta_\eta}
=-\tilde{A}\indices{^{\mu\beta}_{\xi}}\,B\indices{_\nu^\xi^\eta}\,C\indices{_\beta_\eta}\\
-\pfrac{\tilde{S}}{\tilde{A}\indices{^{\alpha\nu}_{\xi}}}\tilde{A}\indices{^{\alpha\mu}_{\xi}}&=
-\delta_\nu^\beta\,\tilde{A}\indices{^{\alpha\mu}_{\xi}}\,B\indices{_\alpha^\xi^\eta}\,C\indices{_\beta_\eta}
=-\tilde{A}\indices{^{\alpha\mu}_{\xi}}\,B\indices{_\alpha^\xi^\eta}\,C\indices{_\nu_\eta}\\
\hphantom{-}\pfrac{\tilde{S}}{\tilde{A}\indices{^{\alpha\beta}_{\mu}}}\tilde{A}\indices{^{\alpha\beta}_{\nu}}&=
\hphantom{-}\delta_\xi^\mu\,\tilde{A}\indices{^{\alpha\beta}_{\nu}}\,B\indices{_\alpha^\xi^\eta}\,C\indices{_\beta_\eta}
=\hphantom{-}\tilde{A}\indices{^{\alpha\beta}_{\nu}}\,B\indices{_\alpha^\mu^\eta}\,C\indices{_\beta_\eta}\\
-\pfrac{\tilde{S}}{B\indices{_\alpha^\nu^\eta}}B\indices{_\alpha^\mu^\eta}&=
-\tilde{A}\indices{^{\alpha\beta}_{\xi}}\,\delta_\nu^\xi\,B\indices{_\alpha^\mu^\eta}\,C\indices{_\beta_\eta}
=-\tilde{A}\indices{^{\alpha\beta}_{\nu}}\,B\indices{_\alpha^\mu^\eta}\,C\indices{_\beta_\eta}\\
-\pfrac{\tilde{S}}{B\indices{_{\alpha}^{\xi\nu}}}B\indices{_{\alpha}^{\xi\mu}}
&=-\tilde{A}\indices{^{\alpha\beta}_{\xi}}\,\delta_\nu^\eta\,B\indices{_\alpha^\xi^\mu}\,C\indices{_\beta_\eta}
=-\tilde{A}\indices{^{\alpha\beta}_{\xi}}\,B\indices{_\alpha^\xi^\mu}\,C\indices{_\beta_\nu}\\
\hphantom{-}\pfrac{\tilde{S}}{B\indices{_{\mu}^{\xi\eta}}}B\indices{_{\nu}^{\xi\eta}}
&=\hphantom{-}\tilde{A}\indices{^{\alpha\beta}_{\xi}}\,\delta_\alpha^\mu\,B\indices{_\nu^\xi^\eta}C\indices{_\beta_\eta}
=\hphantom{-}\tilde{A}\indices{^{\mu\beta}_{\xi}}\,B\indices{_\nu^\xi^\eta}\,C\indices{_\beta_\eta}\\
\hphantom{-}\pfrac{\tilde{S}}{C\indices{_{\beta}_{\mu}}}C\indices{_{\beta}_{\nu}}
&=\hphantom{-}\tilde{A}\indices{^{\alpha\beta}_{\xi}}\,B\indices{_\alpha^\xi^\eta}\,\delta^\mu_\eta\,C\indices{_\beta_\nu}
=\hphantom{-}\tilde{A}\indices{^{\alpha\beta}_{\xi}}\,B\indices{_\alpha^\xi^\mu}\,C\indices{_\beta_\nu}\\
\hphantom{-}\pfrac{\tilde{S}}{C\indices{_{\mu}_{\eta}}}\,C\indices{_{\nu}_{\eta}}
&=\hphantom{-}\tilde{A}\indices{^{\alpha\beta}_{\xi}}\,B\indices{_\alpha^\xi^\eta}\,\delta_\beta^\mu\,C\indices{_\nu_\eta}
=\hphantom{-}\tilde{A}\indices{^{\alpha\mu}_{\xi}}\,B\indices{_\alpha^\xi^\eta}\,C\indices{_\nu_\eta}.
\end{align*}
The terms on the right-hand sides in each case occur twice with opposite signs---and thus sum up to zero.
The remaining term on the left-hand side of Eq.~(\ref{eq:arb-scal}) is the derivative of $\tilde{S}$ with respect to the metric.
Making use of Eq.~(\ref{eq:metric-identity2}), this gives
\begin{equation*}
2\pfrac{\tilde{S}}{g_{\mu\beta}}g_{\nu\beta}=2Sg_{\nu\beta}\pfrac{\sqrt{-g}}{g_{\mu\beta}}
=S\,g_{\nu\beta}\,g^{\beta\mu}\,\sqrt{-g}=\delta_\nu^\mu\,\tilde{S},
\end{equation*}
which verifies the assertion~(\ref{eq:arb-scal}).

As the corresponding \emph{scalar} $S$ does not depend on the metric $g_{\mu\nu}$, its number
of upper and lower indices must be equal, hence $n-m=0$.
Both, the left- and right-hand side of Eq.~(\ref{eq:assertion-contr}) are thus zero in this case.

\subsection{Hilbert Lagrangian $\tilde{\LCd}_\mathrm{H}$}
The Hilbert Lagrangian $\tilde{\LCd}_\mathrm{H}$ is defined as the Ricci scalar density $\tilde{R}=R\sqrt{-g}$,
which in turn is defined as the following contraction of the Riemann-Cartan curvature tensor
$R\indices{^{\eta}_{\alpha\eta\lambda}}$ from Eq.~(\ref{eq:riemann-tensor}):
\begin{equation}\label{eq:ricci-scalar}
\tilde{\LCd}_\mathrm{H}\equiv\tilde{R}=R\indices{^{\eta}_{\alpha\eta\lambda}}\,g^{\alpha\lambda}\,\sqrt{-g}.
\end{equation}
With the scalar density $\tilde{R}$, the general equation~(\ref{eq:assertion3}) takes on the particular form
\begin{equation}\label{eq:Ricci-identity}
-2\pfrac{\tilde{R}}{g^{\nu\beta}}g^{\mu\beta}-
\pfrac{\tilde{R}}{R\indices{^{\nu}_{\alpha\eta\lambda}}}R\indices{^{\mu}_{\alpha\eta\lambda}}+
\pfrac{\tilde{R}}{R\indices{^{\eta}_{\mu\eta\lambda}}}R\indices{^{\xi}_{\nu\xi\lambda}}+
2\pfrac{\tilde{R}}{R\indices{^{\eta}_{\alpha\eta\mu}}}R\indices{^{\xi}_{\alpha\xi\nu}}\equiv\delta_\nu^\mu\,\tilde{R}.
\end{equation}
The factors ``$2$'' occur due to the symmetry of the metric and the skew-symmetry of the Riemann tensor in its last index pair.
With
\begin{equation*}
\pfrac{\sqrt{-g}}{g^{\nu\beta}}=-\onehalf g_{\beta\nu}\,\sqrt{-g}
\end{equation*}
and making use of the skew-symmetries of the Riemann tensor in its first and second index pairs, one finds
\begin{subequations}\label{eq:Hilbert-deris}
\begin{align}
2\pfrac{\tilde{R}}{g^{\nu\beta}}=\tilde{R}_{\nu\beta}+\tilde{R}_{\beta\nu}-g_{\beta\nu}\,\tilde{R},\qquad
\tilde{R}_{\nu\beta}\equiv\tilde{R}\indices{^{\eta}_{\nu\eta\beta}}.
\label{eq:Hilbert-deri1}
\end{align}
The derivatives of $\tilde{R}$ with respect to the Riemann tensor are
\begin{align}
\pfrac{\tilde{R}}{R\indices{^{\nu}_{\alpha\eta\lambda}}}R\indices{^{\mu}_{\alpha\eta\lambda}}
&=\delta^{\eta}_{\nu}\,g^{\alpha\lambda}\tilde{R}\indices{^{\mu}_{\alpha\eta\lambda}}
=\tilde{R}\indices{^{\mu\lambda}_{\nu\lambda}}=\tilde{R}\indices{^\mu_\nu}\label{eq:Hilbert-deri2}\\
\pfrac{\tilde{R}}{R\indices{^{\eta}_{\mu\eta\lambda}}}R\indices{^{\xi}_{\nu\xi\lambda}}
&=\delta_{\alpha}^{\mu}\,g^{\alpha\lambda}\tilde{R}\indices{^{\eta}_{\nu\eta\lambda}}
=\tilde{R}\indices{^{\eta}_{\nu\eta}^{\mu}}=\tilde{R}\indices{_\nu^\mu}\label{eq:Hilbert-deri3}\\
\pfrac{\tilde{R}}{R\indices{^{\eta}_{\alpha\eta\mu}}}R\indices{^{\xi}_{\alpha\xi\nu}}
&=\delta_{\lambda}^{\mu}\,g^{\alpha\lambda}\tilde{R}\indices{^{\xi}_{\alpha\xi\nu}}
=\tilde{R}\indices{^{\eta\mu}_{\eta\nu}}=\tilde{R}\indices{^\mu_\nu}\label{eq:Hilbert-deri4}.
\end{align}
\end{subequations}
The identity~(\ref{eq:Ricci-identity}) is obviously satisfied by Eqs.~(\ref{eq:Hilbert-deris}).

From Eq.~(\ref{eq:Hilbert-deri1}) one observes that the leftmost term in Eq.~(\ref{eq:Ricci-identity})
provides the \emph{symmetric} part of the Ricci tensor
\begin{equation*}
\pfrac{R}{g^{\nu\mu}}=\onehalf\left(R_{\nu\mu}+R_{\mu\nu}\right),
\end{equation*}
whereas the sum of the second and the third terms of Eq.~(\ref{eq:Ricci-identity}) yields its \emph{skew-symmetric} part
\begin{equation*}
\onehalf\left(\pfrac{R}{R\indices{^{\eta}_{\beta\eta\lambda}}}R\indices{^{\xi}_{\nu\xi\lambda}}-
\pfrac{R}{R\indices{^{\nu}_{\alpha\eta\lambda}}}R\indices{^{\beta}_{\alpha\eta\lambda}}\right)g_{\beta\mu}
=\onehalf\left(R_{\nu\mu}-R_{\mu\nu}\right).
\end{equation*}
The skew-symmetric part of $R_{\nu\mu}$ vanishes for the case of zero torsion.
For this case, the identity~(\ref{eq:Ricci-identity}) simplifies to
\begin{equation*}
\frac{1}{\sqrt{-g}}\pfrac{\tilde{R}}{g^{\nu\beta}}g^{\mu\beta}\equiv
\pfrac{R}{R\indices{^{\eta}_{\alpha\eta\mu}}}R\indices{^{\xi}_{\alpha\xi\nu}}-\onehalf\delta_\nu^\mu\,R
=R\indices{^\mu_\nu}-\onehalf\delta_\nu^\mu\,R.
\end{equation*}
\subsection{Riemann tensor squared}
Any \emph{absolute} scalar Lagrangian $\LCd_{R}(R,g)$ built from the Riemann-Cartan tensor~(\ref{eq:riemann-tensor})
and the metric satisfies the identity~(\ref{eq:assertion1})
\begin{equation}\label{eq:constraint2-lag}
2\pfrac{\LCd_{R}}{g^{\nu\beta}}g^{\mu\beta}\equiv
-\pfrac{\LCd_{R}}{R\indices{^{\nu}_{\alpha\beta\lambda}}}R\indices{^{\mu}_{\alpha\beta\lambda}}+
\pfrac{\LCd_{R}}{R\indices{^{\eta}_{\mu\beta\lambda}}}R\indices{^{\eta}_{\nu\beta\lambda}}+
2\pfrac{\LCd_{R}}{R\indices{^{\eta}_{\alpha\beta\mu}}}R\indices{^{\eta}_{\alpha\beta\nu}}.
\end{equation}
The factors ``$2$'' again emerge from the symmetry of the metric and the skew-symmetry of the Riemann-Cartan tensor in its last index pair.
The identity is easily verified for a Lagrangian linear and quadratic in the Riemann tensor:
\begin{equation*}
\LCd_{R}=R\indices{^{\eta}_{\alpha\beta\lambda}}\left[\quarter R\indices{^{\alpha}_{\eta\xi\tau}}g^{\beta\xi}g^{\tau\lambda}-
g_{2}\left(\delta_{\eta}^{\beta}\,g^{\alpha\lambda}-\delta_{\eta}^{\lambda}\,g^{\alpha\beta}\right)\right].
\end{equation*}
The left-hand side of Eq.~(\ref{eq:constraint2-lag}) evaluates to
\begin{equation*}
2\pfrac{\LCd_{R}}{g^{\nu\beta}}g^{\mu\beta}=-R^{\eta\alpha\beta\mu}\,
R_{\eta\alpha\beta\nu}-2g_{2}\left(R\indices{_{\nu}^{\mu}}+R\indices{^{\mu}_{\nu}}\right),
\end{equation*}
which indeed agrees with the terms obtained from the right-hand side:
\begin{align*}
\pfrac{\LCd_{R}}{R\indices{^{\eta}_{\mu\beta\lambda}}}R\indices{^{\eta}_{\nu\beta\lambda}}
-\pfrac{\LCd_{R}}{R\indices{^{\nu}_{\alpha\beta\lambda}}}R\indices{^{\mu}_{\alpha\beta\lambda}}
&=2g_{2}\left(R\indices{^{\mu}_{\nu}}-R\indices{_{\nu}^{\mu}}\right)\\
2\pfrac{\LCd_{R}}{R\indices{^{\eta}_{\alpha\beta\mu}}}R\indices{^{\eta}_{\alpha\beta\nu}}
&=-R\indices{^{\eta\alpha\beta\mu}}R\indices{_{\eta\alpha\beta\nu}}-4g_2\,R\indices{^{\mu}_{\nu}}.
\end{align*}
\subsection{Ricci tensor squared}
The scalar made of the (not necessarily symmetric) Ricci tensor $R_{\eta\alpha}$ is defined by the following contraction with the metric
\begin{equation}\label{eq:ricci-tensor-scalar}
\LCd_{R}=R\indices{_{\eta\alpha}}R\indices{_{\xi\lambda}}\,g^{\eta\xi}\,g^{\alpha\lambda}.
\end{equation}
With Eq.~(\ref{eq:ricci-tensor-scalar}), the general Eq.~(\ref{eq:assertion1}) now takes on the particular form
\begin{equation}\label{eq:Ricci-tensor-identity}
2\pfrac{\LCd_{R}}{g^{\nu\beta}}g^{\mu\beta}-
\pfrac{\LCd_{R}}{R_{\mu\beta}}R_{\nu\beta}-\pfrac{\LCd_{R}}{R_{\beta\mu}}R_{\beta\nu}\equiv0.
\end{equation}
Without making use of the symmetries of the Ricci tensor and the metric, this identity is actually fulfilled as
\begin{align*}
\pfrac{\LCd_{R}}{g^{\nu\beta}}g^{\mu\beta}&=R\indices{_{\eta\alpha}}R\indices{_{\xi\lambda}}\left(
\delta_{\nu}^{\eta}\,\delta_{\beta}^{\xi}\,g^{\alpha\lambda}+
g^{\eta\xi}\,\delta_{\nu}^{\alpha}\,\delta_{\beta}^{\lambda}\right)g^{\beta\mu}\\
&=\left(R\indices{_{\nu\alpha}}R\indices{_{\beta\lambda}}\,g^{\alpha\lambda}+
g^{\eta\xi}\,R\indices{_{\eta\nu}}R\indices{_{\xi\beta}}\right)g^{\beta\mu}\\
&=R\indices{_{\nu\beta}}R\indices{^{\mu\beta}}+R\indices{_{\beta\nu}}R\indices{^{\beta\mu}}.
\end{align*}
The derivative terms of the Ricci tensor are
\begin{equation*}
\pfrac{\LCd_{R}}{R_{\mu\beta}}R_{\nu\beta}=\left(\delta_{\eta}^{\mu}\delta_{\alpha}^{\beta}R_{\xi\lambda}+
R_{\eta\alpha}\delta_{\xi}^{\mu}\delta_{\lambda}^{\beta}\right)g^{\eta\xi}\,g^{\alpha\lambda}\,R_{\nu\beta}
=2R\indices{_{\nu\beta}}\,R\indices{^{\mu\beta}}
\end{equation*}
and
\begin{align*}
\pfrac{\LCd_{R}}{R_{\beta\mu}}R_{\beta\nu}=\left(\delta_{\eta}^{\beta}\delta_{\alpha}^{\mu}R_{\xi\lambda}+
R_{\eta\alpha}\delta_{\xi}^{\beta}\delta_{\lambda}^{\mu}\right)g^{\eta\xi}\,g^{\alpha\lambda}\,R_{\beta\nu}
=2R\indices{_{\beta\nu}}\,R\indices{^{\beta\mu}},
\end{align*}
which obviously cancel the terms emerging from the derivatives with respect to the metric.

For zero torsion, the Ricci tensor is symmetric.
Then
\begin{equation}
\pfrac{\LCd_{R}}{g^{\nu\beta}}g^{\mu\beta}\equiv\pfrac{\LCd_{R}}{R_{\mu\beta}}R_{\nu\beta}
\qquad\Leftrightarrow\qquad\pfrac{\LCd_{R}}{g^{\nu\mu}}\equiv\pfrac{\LCd_{R}}{R\indices{_{\alpha\beta}}}R_{\nu\beta}\,g_{\alpha\mu}.
\end{equation}
\subsection{Klein-Gordon Lagrangian and Hamiltonian\label{ex:KG-ham}}
The Klein-Gordon Lagrangian density $\tilde{\LCd}_{\mathrm{KG}}=\LCd_{\mathrm{KG}}\sqrt{-g}$
for a massive \emph{complex} scalar field $\phi(x)$ is given by:
\begin{align}
\tilde{\LCd}_{\mathrm{KG}}\left(\phi,\bar{\phi},\partial_\mu\phi,\partial_\nu\bar{\phi},g^{\mu\nu}\right)
&=\left(\pfrac{\bar{\phi}}{x^{\alpha}}\pfrac{\phi}{x^{\beta}}g^{\alpha\beta}-m^2\bar{\phi}\,\phi\right)\sqrt{-g}\nonumber\\
&=\left[\frac{1}{2}\left(\pfrac{\bar{\phi}}{x^{\alpha}}\pfrac{\phi}{x^{\beta}}
+\pfrac{\bar{\phi}}{x^{\beta}}\pfrac{\phi}{x^{\alpha}}\right)g^{\alpha\beta}-m^2\bar{\phi}\,\phi\right]\sqrt{-g}.
\label{eq:KG-Lagrangian-complex}
\end{align}
In order to set up the pertaining identity~(\ref{eq:assertion2}), we set the required derivatives of $\tilde{\LCd}_{\mathrm{KG}}$.
With the derivative of the determinant $g$ of the covariant metric $g_{\mu\nu}$ with respect to the contravariant metric,
\begin{equation*}
\pfrac{\sqrt{-g}}{g^{\nu\beta}}g^{\mu\beta}=-\onehalf g_{\beta\nu}g^{\mu\beta}\sqrt{-g}=-\onehalf\delta_\nu^\mu\sqrt{-g},
\end{equation*}
we find
\begin{align*}
2\pfrac{\tilde{\LCd}_{\mathrm{KG}}}{g^{\nu\beta}}g^{\mu\beta}&=\left(\pfrac{\bar{\phi}}{x^{\nu}}\pfrac{\phi}{x^{\beta}}
+\pfrac{\bar{\phi}}{x^{\beta}}\pfrac{\phi}{x^{\nu}}\right)g^{\mu\beta}\sqrt{-g}-\delta_\nu^\mu\,\tilde{\LCd}_{\mathrm{KG}}\\
\pfrac{\bar{\phi}}{x^\nu}\pfrac{\tilde{\LCd}_{\mathrm{KG}}}{\left(\pfrac{\bar{\phi}}{x^\mu}\right)}
&=\frac{1}{2}\pfrac{\bar{\phi}}{x^{\nu}}\left(\pfrac{\phi}{x^{\beta}}g^{\mu\beta}+\pfrac{\phi}{x^{\alpha}}g^{\alpha\mu}\right)\sqrt{-g}
=\pfrac{\bar{\phi}}{x^{\nu}}\pfrac{\phi}{x^{\beta}}g^{\mu\beta}\sqrt{-g}\\
\pfrac{\tilde{\LCd}_{\mathrm{KG}}}{\left(\pfrac{\phi}{x^\mu}\right)}\pfrac{\phi}{x^\nu}
&=\frac{1}{2}\left(\pfrac{\bar{\phi}}{x^{\alpha}}g^{\alpha\mu}+\pfrac{\bar{\phi}}{x^{\beta}}g^{\mu\beta}\right)\pfrac{\phi}{x^{\nu}}\sqrt{-g}
=\pfrac{\bar{\phi}}{x^{\beta}}\pfrac{\phi}{x^{\nu}}g^{\mu\beta}\sqrt{-g}.
\end{align*}
The particular identity for the Lagrangian density~(\ref{eq:KG-Lagrangian-complex}) thus writes:
\begin{equation}
-2\pfrac{\tilde{\LCd}_{\mathrm{KG}}}{g^{\nu\beta}}g^{\mu\beta}
+\pfrac{\bar{\phi}}{x^\nu}\pfrac{\tilde{\LCd}_{\mathrm{KG}}}{\left(\pfrac{\bar{\phi}}{x^\mu}\right)}
+\pfrac{\tilde{\LCd}_{\mathrm{KG}}}{\left(\pfrac{\phi}{x^\mu}\right)}\pfrac{\phi}{x^\nu}\equiv\delta_\nu^\mu\,\tilde{\LCd}_{\mathrm{KG}}.
\end{equation}
The equivalent covariant Klein-Gordon Hamiltonian density $\tilde{\HCd}_{\mathrm{KG}}\big(\phi,\bar{\phi},\tilde{\pi}^{\mu},\tilde{\bar{\pi}}^{\nu},g_{\mu\nu}\big)$
for a system of complex fields in a dynamic spacetime is given by:
\begin{align}
\tilde{\HCd}_{\mathrm{KG}}&=\tilde{\bar{\pi}}^{\alpha}\pfrac{\phi}{x^\alpha}+\pfrac{\bar{\phi}}{x^\alpha}\tilde{\pi}^{\alpha}-\tilde{\LCd}_{\mathrm{KG}},\qquad
\tilde{\bar{\pi}}^{\alpha}=\pfrac{\tilde{\LCd}_{\mathrm{KG}}}{\left(\pfrac{\phi}{x^\mu}\right)},\qquad
\tilde{\pi}^{\alpha}=\pfrac{\tilde{\LCd}_{\mathrm{KG}}}{\left(\pfrac{\bar{\phi}}{x^\mu}\right)}\nonumber\\
&=\tilde{\bar{\pi}}^{\alpha}\tilde{\pi}^{\,\beta}\,g_{\alpha\beta}\frac{1}{\sqrt{-g}}+m^{2}\,\bar{\phi}\,\phi\sqrt{-g}.
\label{eq:KG-Hamiltonian-complex}
\end{align}
We observe that the scalar density $\tilde{\HCd}_{\mathrm{KG}}$ is defined as a function of the \emph{tensor densities}
$\tilde{\bar{\pi}}^{\alpha}$ and $\tilde{\pi}^{\,\beta}$ rather than of absolute tensors.
To set up the pertaining identity, $\tilde{\HCd}_{\mathrm{KG}}$ can be rewritten as:
\begin{align}
\tilde{\HCd}_{\mathrm{KG}}^{\prime}&=\left(\bar{\pi}^{\alpha}\pi^{\,\beta}\,g_{\alpha\beta}+m^{2}\,\bar{\phi}\,\phi\right)\sqrt{-g}\nonumber\\
&=\left[\onehalf\left(\bar{\pi}^{\alpha}\pi^{\,\beta}+\bar{\pi}^{\beta}\pi^{\,\alpha}\right)g_{\alpha\beta}+m^{2}\,\bar{\phi}\,\phi\right]\sqrt{-g}.
\label{eq:KG-Hamiltonian-complex-prime}
\end{align}
As can be directly verified, the invariant now takes on the form:
\begin{equation}
2\pfrac{\tilde{\HCd}_{\mathrm{KG}}^{\prime}}{g_{\mu\beta}}g_{\nu\beta}
-\bar{\pi}^\mu\pfrac{\tilde{\HCd}_{\mathrm{KG}}^{\prime}}{\bar{\pi}^\nu}
-\pfrac{\tilde{\HCd}_{\mathrm{KG}}^{\prime}}{\pi^\nu}\pi^\mu\equiv\delta_\nu^\mu\,\tilde{\HCd}_{\mathrm{KG}}^{\prime}.
\end{equation}
The correlation of the derivatives with respect to the metric of $\tilde{\HCd}_{\mathrm{KG}}^{\prime}$ and $\tilde{\HCd}_{\mathrm{KG}}$ follows as:
\begin{align*}
2\pfrac{\tilde{\HCd}_{\mathrm{KG}}^\prime}{g_{\beta\mu}}g_{\beta\nu}&=2\pfrac{\tilde{\HCd}_{\mathrm{KG}}}{g_{\beta\mu}}g_{\beta\nu}
+\delta_\nu^\mu\left(\tilde{\HCd}_{\mathrm{KG}}^\prime+\tilde{\LCd}_{\mathrm{KG}}\right)
=2\pfrac{\tilde{\LCd}_{\mathrm{KG}}}{g^{\beta\nu}}g^{\beta\mu}+\delta_\nu^\mu\left(\tilde{\HCd}_{\mathrm{KG}}^\prime+\tilde{\LCd}_{\mathrm{KG}}\right),
\end{align*}
hence
\begin{equation*}
\pfrac{\tilde{\HCd}_{\mathrm{KG}}}{g_{\beta\mu}}g_{\beta\nu}=\pfrac{\tilde{\LCd}_{\mathrm{KG}}}{g^{\beta\nu}}g^{\beta\mu}.
\end{equation*}
In terms of the proper Hamiltonian~(\ref{eq:KG-Hamiltonian-complex}), the invariant writes:
\begin{equation}
2\pfrac{\tilde{\HCd}_{\mathrm{KG}}}{g_{\mu\beta}}g_{\nu\beta}
-\tilde{\bar{\pi}}^\mu\pfrac{\tilde{\HCd}_{\mathrm{KG}}}{\tilde{\bar{\pi}}^\nu}
-\pfrac{\tilde{\HCd}_{\mathrm{KG}}}{\tilde{\pi}^\nu}\tilde{\pi}^\mu\equiv-\delta_\nu^\mu\,\tilde{\LCd}_{\mathrm{KG}}.
\end{equation}
\subsection{Proca Lagrangian and Hamiltonian\label{ex:proca-ham}}
In a dynamic spacetime, the Proca Lagrangian from Eq.~(\ref{eq:Proca-Lag}) is defined by:
\begin{equation*}
\tilde{\LCd}_{\mathrm{P}}=\left(-\quarter p_{\alpha\beta}\,p_{\xi\eta}g^{\alpha\xi}\,g^{\beta\eta}
+\onehalf m^{2}a_{\alpha}\,a_{\beta}\,g^{\alpha\beta}\right)\sqrt{-g},\qquad
p_{\alpha\beta}\equiv f_{\alpha\beta}\equiv\pfrac{a_\beta}{x^\alpha}-\pfrac{a_\alpha}{x^\beta}.
\end{equation*}
The identity~(\ref{eq:assertion3}) for the scalar density $\tilde{\LCd}_{\mathrm{P}}$ is then:
\begin{equation*}
-2\pfrac{\tilde{\LCd}_{\mathrm{P}}}{g^{\beta\nu}}g^{\beta\mu}+\pfrac{\tilde{\LCd}_{\mathrm{P}}}{p_{\mu\beta}}p_{\nu\beta}
+\pfrac{\tilde{\LCd}_{\mathrm{P}}}{p_{\beta\mu}}p_{\beta\nu}+\pfrac{\tilde{\LCd}_{\mathrm{P}}}{a_\mu}a_\nu\equiv\delta_\nu^\mu\tilde{\LCd}_{\mathrm{P}}.
\end{equation*}
The equivalent covariant Proca Hamiltonian density $\tilde{\HCd}_{\mathrm{P}}\big(\tilde{p}^{\nu\mu},a_{\nu},g_{\mu\nu}\big)$
is obtained by a complete Legendre transformation as:
\begin{align}
\tilde{\HCd}_{\mathrm{P}}&=\tilde{p}^{\alpha\beta}\,\pfrac{a_\alpha}{x^\beta}-\tilde{\LCd}_{\mathrm{P}}\nonumber\\
&=-\quarter\tilde{p}^{\alpha\beta}\,\tilde{p}^{\xi\eta}g_{\alpha\xi}\,g_{\beta\eta}\frac{1}{\sqrt{-g}}
-\onehalf m^{2}a_{\alpha}\,a_{\xi}\,g^{\alpha\xi}\sqrt{-g}.\label{eq:Proca-Ham}
\end{align}
In order to set up the related identity, the Hamiltonian $\tilde{\HCd}_{\mathrm{P}}$ must first be expressed as the equivalent function of absolute tensors:
\begin{equation}\label{eq:Proca-Ham-temp}
\tilde{\HCd}_{\mathrm{P}}^\prime=-\quarter p^{\alpha\beta}\,p^{\xi\eta}g_{\alpha\xi}\,g_{\beta\eta}\sqrt{-g}-
\onehalf m^{2}a_{\alpha}\,a_{\xi}\,g^{\alpha\xi}\sqrt{-g}.
\end{equation}
The identity~(\ref{eq:assertion3}) then follows as
\begin{equation}\label{eq:Proca-Ham-tensor-identity}
2\pfrac{\tilde{\HCd}_{\mathrm{P}}^\prime}{g_{\beta\mu}}g_{\beta\nu}
-\pfrac{\tilde{\HCd}_{\mathrm{P}}^\prime}{p^{\nu\beta}}p^{\mu\beta}-\pfrac{\tilde{\HCd}_{\mathrm{P}}^\prime}{p^{\beta\nu}}p^{\beta\mu}
+\pfrac{\tilde{\HCd}_{\mathrm{P}}^\prime}{a_\mu}\,a_\nu\equiv\delta_\nu^\mu\,\tilde{\HCd}_{\mathrm{P}}^\prime,
\end{equation}
which is again verified by direct calculation:
\allowdisplaybreaks
\begin{align*}
-\pfrac{\tilde{\HCd}_{\mathrm{P}}^\prime}{p^{\nu\beta}}p^{\mu\beta}&=
\onehalf\,p^{\mu\beta}\,p^{\xi\eta}g_{\nu\xi}\,g_{\beta\eta}\sqrt{-g}=\onehalf\,p^{\mu\beta}\,p_{\nu\beta}\sqrt{-g}\\
-\pfrac{\tilde{\HCd}_{\mathrm{P}}^\prime}{p^{\alpha\nu}}p^{\alpha\mu}&=
\onehalf\,p^{\alpha\mu}\,p^{\xi\eta}g_{\alpha\xi}\,g_{\nu\eta}\sqrt{-g}=\onehalf\,p^{\beta\mu}\,p_{\beta\nu}\sqrt{-g}\\
\pfrac{\tilde{\HCd}_{\mathrm{P}}^\prime}{a_\mu}\,a_\nu&=-m^{2}a_{\nu}\,a_{\xi}\,g^{\mu\xi}\sqrt{-g}=-m^{2}a^{\mu}\,a_{\nu}\sqrt{-g}\\
2\pfrac{\tilde{\HCd}_{\mathrm{P}}^\prime}{g_{\beta\mu}}g_{\beta\nu}&=-\onehalf p^{\beta\lambda}p^{\mu\eta}g_{\beta\nu}g_{\lambda\eta}\sqrt{-g}
-\onehalf p^{\alpha\beta}p^{\xi\mu}g_{\alpha\xi}g_{\beta\nu}\sqrt{-g}
+m^{2}a_{\alpha}a_{\xi}g^{\alpha\beta}g^{\xi\mu}g_{\beta\nu}\sqrt{-g}+g^{\mu\beta}g_{\beta\nu}\tilde{\HCd}_{\mathrm{P}}^\prime\\
&=-\onehalf p_{\nu\beta}p^{\mu\beta}\sqrt{-g}-\onehalf p_{\beta\nu}p^{\beta\mu}\sqrt{-g}+m^{2}a_{\nu}\,a^{\mu}\sqrt{-g}+\delta_\nu^\mu\,\tilde{\HCd}_{\mathrm{P}}^\prime.
\end{align*}
Summing up, all terms on the right-hand side cancel---except $\delta_\nu^\mu\,\tilde{\HCd}_{\mathrm{P}}^\prime$, which thus gives Eq.~(\ref{eq:Proca-Ham-tensor-identity}).

For the proper Hamiltonian~(\ref{eq:Proca-Ham}), we have
\begin{align*}
\pfrac{\tilde{\HCd}_{\mathrm{P}}^\prime}{p^{\nu\beta}}p^{\mu\beta}&=-\onehalf\,p^{\mu\beta}\,p_{\nu\beta}\sqrt{-g}
=-\onehalf\,\tilde{p}^{\mu\beta}\,\tilde{p}_{\nu\beta}\frac{1}{\sqrt{-g}}=\pfrac{\tilde{\HCd}_{\mathrm{P}}}{\tilde{p}^{\nu\beta}}\tilde{p}^{\mu\beta}\\
\pfrac{\tilde{\HCd}_{\mathrm{P}}^\prime}{a_\mu}\,a_\nu&=-m^{2}a_{\nu}\,a^{\mu}\sqrt{-g}=\pfrac{\tilde{\HCd}_{\mathrm{P}}}{a_\mu}\,a_\nu.
\end{align*}
Yet the derivatives with respect to the metric of $\tilde{\HCd}_{\mathrm{P}}^\prime$ and $\tilde{\HCd}_{\mathrm{P}}$ are different:
\begin{align*}
2\pfrac{\tilde{\HCd}_{\mathrm{P}}^\prime}{g_{\beta\mu}}g_{\beta\nu}&=2\pfrac{\tilde{\HCd}_{\mathrm{P}}}{g_{\beta\mu}}g_{\beta\nu}
+\delta_\nu^\mu\left(\tilde{\HCd}_{\mathrm{P}}^\prime+\tilde{\LCd}_{\mathrm{P}}\right)
=2\pfrac{\tilde{\LCd}_{\mathrm{P}}}{g^{\beta\nu}}g^{\beta\mu}+\delta_\nu^\mu\left(\tilde{\HCd}_{\mathrm{P}}^\prime+\tilde{\LCd}_{\mathrm{P}}\right),
\end{align*}
hence
\begin{equation*}
\pfrac{\tilde{\HCd}_{\mathrm{P}}}{g_{\beta\mu}}g_{\beta\nu}=\pfrac{\tilde{\LCd}_{\mathrm{P}}}{g^{\beta\nu}}g^{\beta\mu}.
\end{equation*}
The identity~(\ref{eq:Proca-Ham-tensor-identity}) can now be expressed in terms of the proper Proca Hamiltonian $\tilde{\HCd}_{\mathrm{P}}$ as:
\begin{equation*}
2\pfrac{\tilde{\HCd}_{\mathrm{P}}}{g_{\beta\mu}}g_{\beta\nu}
-\pfrac{\tilde{\HCd}_{\mathrm{P}}}{\tilde{p}^{\nu\beta}}\tilde{p}^{\mu\beta}-\pfrac{\tilde{\HCd}_{\mathrm{P}}}{\tilde{p}^{\beta\nu}}\tilde{p}^{\beta\mu}
+\pfrac{\tilde{\HCd}_{\mathrm{P}}}{a_\mu}\,a_\nu\equiv
-\delta_\nu^\mu\,\tilde{\LCd}_{\mathrm{P}}.
\end{equation*}
With the \emph{canonical} energy-momentum tensor density $\tilde{\theta}\indices{_\nu^\mu}=\theta\indices{_\nu^\mu}\sqrt{-g}$ of the Proca system,
\begin{equation*}
\tilde{\theta}\indices{_\nu^\mu}
=f_{\beta\nu}\pfrac{\tilde{\LCd}_{\mathrm{P}}}{f_{\beta\mu}}-\delta_\nu^\mu\,\tilde{\LCd}_{\mathrm{P}}
=\pfrac{\tilde{\HCd}_{\mathrm{P}}}{\tilde{p}^{\beta\nu}}\tilde{p}^{\beta\mu}-\delta_\nu^\mu\,\tilde{\LCd}_{\mathrm{P}},
\end{equation*}
this yields the Hamiltonian representation of $\tilde{\theta}\indices{_\nu^\mu}$:
\begin{equation*}
\tilde{\theta}\indices{_\nu^\mu}\equiv 2\pfrac{\tilde{\HCd}_{\mathrm{P}}}{g_{\beta\mu}}g_{\beta\nu}
-\pfrac{\tilde{\HCd}_{\mathrm{P}}}{\tilde{p}^{\nu\beta}}\tilde{p}^{\mu\beta}
+\pfrac{\tilde{\HCd}_{\mathrm{P}}}{a_\mu}\,a_\nu.
\end{equation*}
\subsection{Dirac Lagrangian}
The regularized Dirac Lagrangian $\LCd_{\mathrm{D}}\big(\psi,\partial\psi,\bar{\psi},\partial\bar{\psi},\bgamma^\mu\big)$,
constructed upon the Dirac equation writes
\begin{equation}\label{ld-regular}
\LCd_{\mathrm{D}}=\frac{\rmi}{2}\left(\bar{\psi}\,\bgamma^{\alpha}\pfrac{\psi}{x^{\alpha}}
-\pfrac{\bar{\psi}}{x^{\alpha}}\bgamma^{\alpha}\psi\right)-m\,\bar{\psi}\psi
+\frac{\rmi}{3M}\pfrac{\bar{\psi}}{x^{\alpha}}\,\bsigma^{\alpha\beta}\pfrac{\psi}{x^{\beta}},
\end{equation}
wherein the $(1,1)$-spinor-$(2,0)$-tensor field $\bsigma^{\alpha\beta}$ is defined
as the \emph{commutator} of the matrix product $\bgamma^{\alpha}\bgamma^{\nu}$:
\begin{equation}\label{eq:sigma-def}
\bsigma^{\alpha\nu}=\frac{\rmi}{2}\left(\bgamma^{\alpha}\bgamma^{\nu}-\bgamma^{\nu}\bgamma^{\alpha}\right)
\qquad\Leftrightarrow\qquad\sigma\indices{^{a}_{b}^{\alpha\nu}}=\frac{\rmi}{2}
\left(\gamma\indices{^{a}_{c}^{\alpha}}\,\gamma\indices{^{c}_{b}^{\nu}}-
\gamma\indices{^{a}_{c}^{\nu}}\,\gamma\indices{^{c}_{b}^{\alpha}}\right).
\end{equation}
With an explicit notation of the spinor indices as lower case Latin letters,
the Dirac Lagrangian takes on the form:
\begin{equation}\label{ld-regular-expl}
\LCd_{\mathrm{D}}=\frac{\rmi}{2}\left(\bar{\psi}_a\,\gamma\indices{^a_c^\alpha}\pfrac{\psi^c}{x^{\alpha}}
-\pfrac{\bar{\psi}_a}{x^{\alpha}}\gamma\indices{^a_c^\alpha}\psi^c\right)-m\,\bar{\psi}_a\psi^a
+\frac{\rmi}{3M}\pfrac{\bar{\psi}_a}{x^{\alpha}}\,\sigma\indices{^a_c^{\alpha\beta}}\pfrac{\psi^c}{x^{\beta}}.
\end{equation}
The invariant~(\ref{eq:assertion1}) for the scalar quantity $\LCd_{\mathrm{D}}$ is then given by:
\begin{equation}\label{ld-invariant}
\pfrac{\LCd_{\mathrm{D}}}{\left(\pfrac{\psi^c}{x^\mu}\right)}\pfrac{\psi^c}{x^\nu}
+\pfrac{\bar{\psi}_a}{x^\nu}\pfrac{\LCd_{\mathrm{D}}}{\left(\pfrac{\bar{\psi}_a}{x^\mu}\right)}
-\pfrac{\LCd_{\mathrm{D}}}{\gamma\indices{^a_c^\nu}}\gamma\indices{^a_c^\mu}\equiv0.
\end{equation}
Equation~(\ref{ld-invariant}) is verified by direct calculation:
\begin{align*}
\pfrac{\LCd_{\mathrm{D}}}{\left(\pfrac{\psi^c}{x^\mu}\right)}\pfrac{\psi^c}{x^\nu}
&=\frac{\rmi}{2}\bar{\psi}_a\,\gamma\indices{^a_c^\mu}\pfrac{\psi^c}{x^{\nu}}
+\frac{\rmi}{3M}\pfrac{\bar{\psi}_a}{x^{\alpha}}\,\sigma\indices{^a_c^{\alpha\mu}}\pfrac{\psi^c}{x^{\nu}}\\
\pfrac{\bar{\psi}_a}{x^\nu}\pfrac{\LCd_{\mathrm{D}}}{\left(\pfrac{\bar{\psi}_a}{x^\mu}\right)}
&=-\frac{\rmi}{2}\pfrac{\bar{\psi}_a}{x^\nu}\gamma\indices{^a_c^\mu}\psi^c
+\frac{\rmi}{3M}\pfrac{\bar{\psi}_a}{x^\nu}\sigma\indices{^a_c^{\mu\beta}}\pfrac{\psi^c}{x^{\beta}}\\
-\pfrac{\LCd_{\mathrm{D}}}{\gamma\indices{^a_c^\nu}}\gamma\indices{^a_c^\mu}
&=-\frac{\rmi}{2}\left(\bar{\psi}_a\,\gamma\indices{^a_c^\mu}\pfrac{\psi^c}{x^{\nu}}
-\pfrac{\bar{\psi}_a}{x^{\nu}}\gamma\indices{^a_c^\mu}\psi^c\right)
-\frac{\rmi}{3M}\left(\pfrac{\bar{\psi}_a}{x^{\nu}}\,\sigma\indices{^a_c^{\mu\beta}}\pfrac{\psi^c}{x^{\beta}}
+\pfrac{\bar{\psi}_a}{x^{\alpha}}\,\sigma\indices{^a_c^{\alpha\mu}}\pfrac{\psi^c}{x^{\nu}}\right),
\end{align*}
which obviously sums up to zero.
The identity~(\ref{eq:assertion3}) for the scalar \emph{density} $\tilde{\LCd}_{\mathrm{D}}$ is then
\begin{equation*}
\pfrac{\tilde{\LCd}_{\mathrm{D}}}{\left(\pfrac{\psi^c}{x^\mu}\right)}\pfrac{\psi^c}{x^\nu}
+\pfrac{\bar{\psi}_a}{x^\nu}\pfrac{\tilde{\LCd}_{\mathrm{D}}}{\left(\pfrac{\bar{\psi}_a}{x^\mu}\right)}
-\pfrac{\tilde{\LCd}_{\mathrm{D}}}{\gamma\indices{^a_c^\nu}}\gamma\indices{^a_c^\mu}\equiv\delta_\nu^\mu\,\tilde{\LCd}_{\mathrm{D}},
\end{equation*}
hence, rearranging the terms and skipping the spinor indices:
\begin{equation}\label{ld-invariant2}
\tilde{\theta}\indices{_\nu^\mu}\equiv
\pfrac{\tilde{\LCd}_{\mathrm{D}}}{\left(\pfrac{\psi}{x^\mu}\right)}\pfrac{\psi}{x^\nu}
+\pfrac{\bar{\psi}}{x^\nu}\pfrac{\tilde{\LCd}_{\mathrm{D}}}{\left(\pfrac{\bar{\psi}}{x^\mu}\right)}
-\delta_\nu^\mu\,\tilde{\LCd}_{\mathrm{D}}\equiv
\Tr\left\{\pfrac{\tilde{\LCd}_{\mathrm{D}}}{\bgamma\indices{^\nu}}\bgamma\indices{^\mu}\right\}\equiv\tilde{T}\indices{_\nu^\mu}.
\end{equation}
Equation~(\ref{ld-invariant2}) states that the \emph{canonical} energy-momentum tensor density $\tilde{\theta}\indices{_\nu^\mu}$
coincides for the Dirac Lagrangian with the \emph{metric} energy-momentum tensor $\tilde{T}\indices{_\nu^\mu}$.

The corresponding invariant~(\ref{eq:assertion1}) for the scalar $\LCd_{\mathrm{D}}$
is obtained if we contract the spacetime indices and leave open the spinor indices:
\begin{equation}\label{ld-invariant-spinor}
\pfrac{\LCd_{\mathrm{D}}}{\left(\pfrac{\psi^b}{x^\alpha}\right)}\pfrac{\psi^a}{x^\alpha}
-\pfrac{\bar{\psi}_b}{x^\alpha}\pfrac{\LCd_{\mathrm{D}}}{\left(\pfrac{\bar{\psi}_a}{x^\alpha}\right)}
+\pfrac{\LCd_{\mathrm{D}}}{\psi^b}\psi^a-\bar{\psi}_b\pfrac{\LCd_{\mathrm{D}}}{\bar{\psi}_a}
+\pfrac{\LCd_{\mathrm{D}}}{\gamma\indices{^b_c^\alpha}}\gamma\indices{^a_c^\alpha}
-\gamma\indices{^c_b^\alpha}\pfrac{\LCd_{\mathrm{D}}}{\gamma\indices{^c_a^\alpha}}\equiv0.
\end{equation}
The invariant~(\ref{ld-invariant-spinor}) is also verified by direct calculation:
\begin{align*}
\pfrac{\LCd_{\mathrm{D}}}{\left(\pfrac{\psi^b}{x^\alpha}\right)}\pfrac{\psi^a}{x^\alpha}
&=\frac{\rmi}{2}\bar{\psi}_c\,\gamma\indices{^c_b^\alpha}\pfrac{\psi^a}{x^{\alpha}}
+\frac{\rmi}{3M}\pfrac{\bar{\psi}_c}{x^{\alpha}}\,\sigma\indices{^c_b^{\alpha\beta}}\pfrac{\psi^a}{x^{\beta}}\\
-\pfrac{\bar{\psi}_b}{x^\alpha}\pfrac{\LCd_{\mathrm{D}}}{\left(\pfrac{\bar{\psi}_a}{x^\alpha}\right)}
&=\frac{\rmi}{2}\pfrac{\bar{\psi}_b}{x^\alpha}\gamma\indices{^a_c^\alpha}\psi^c
-\frac{\rmi}{3M}\pfrac{\bar{\psi}_b}{x^\alpha}\sigma\indices{^a_c^{\alpha\beta}}\pfrac{\psi^c}{x^{\beta}}\\
\pfrac{\LCd_{\mathrm{D}}}{\psi^b}\psi^a&=-\frac{\rmi}{2}\pfrac{\bar{\psi}_c}{x^\alpha}\gamma\indices{^c_b^\alpha}\psi^a-m\,\bar{\psi}_b\psi^a\\
-\bar{\psi}_b\pfrac{\LCd_{\mathrm{D}}}{\bar{\psi}_a}&=-\frac{\rmi}{2}\bar{\psi}_b\gamma\indices{^a_c^\alpha}\pfrac{\psi^c}{x^\alpha}+m\,\bar{\psi}_b\psi^a\\
\pfrac{\LCd_{\mathrm{D}}}{\gamma\indices{^b_c^\alpha}}\gamma\indices{^a_c^\alpha}
&=\frac{\rmi}{2}\left(\bar{\psi}_b\,\gamma\indices{^a_c^\alpha}\pfrac{\psi^c}{x^{\alpha}}
-\pfrac{\bar{\psi}_b}{x^{\alpha}}\gamma\indices{^a_c^\alpha}\psi^c\right)\\
&\quad+\frac{\rmi}{3M}\pfrac{\bar{\psi}_b}{x^{\alpha}}\,\sigma\indices{^a_c^{\alpha\beta}}\pfrac{\psi^c}{x^{\beta}}
-\frac{1}{6M}\pfrac{\bar{\psi}_c}{x^{\alpha}}\left(\gamma\indices{^c_b^\alpha}\gamma\indices{^a_d^\beta}-
\gamma\indices{^c_b^\beta}\gamma\indices{^a_d^\alpha}\right)\pfrac{\psi^d}{x^{\beta}}\\
-\gamma\indices{^c_b^\alpha}\pfrac{\LCd_{\mathrm{D}}}{\gamma\indices{^c_a^\alpha}}
&=-\frac{\rmi}{2}\left(\bar{\psi}_c\,\gamma\indices{^c_b^\alpha}\pfrac{\psi^a}{x^{\alpha}}
-\pfrac{\bar{\psi}_c}{x^{\alpha}}\gamma\indices{^c_b^\alpha}\psi^a\right)\\
&\quad-\frac{\rmi}{3M}\pfrac{\bar{\psi}_c}{x^{\alpha}}\,\sigma\indices{^c_b^{\alpha\beta}}\pfrac{\psi^a}{x^{\beta}}
+\frac{1}{6M}\pfrac{\bar{\psi}_c}{x^{\alpha}}\left(\gamma\indices{^c_b^\alpha}\gamma\indices{^a_d^\beta}-
\gamma\indices{^c_b^\beta}\gamma\indices{^a_d^\alpha}\right)\pfrac{\psi^d}{x^{\beta}}.
\end{align*}
Again, all terms on the right-hand side sum up to zero.
\subsection{Dirac Hamiltonian}
The Dirac momentum density fields $\tilde{\pibar}^\mu$ and $\tilde{\pi}^\mu$
represent the \emph{conjugates} of the spinor fields $\psi$ and $\psibar$, respectively, and hence
the \emph{duals} of their partial derivatives, $\partial\psi/\partial x^\mu$ and $\partial\psibar/\partial x^\mu$.
They are derived from the Lagrangian density $\tilde{\LCd}_{\mathrm{D}}=\LCd_{\mathrm{D}}\sqrt{-g}$ via
\begin{equation}
\tilde{\pibar}^\mu=\pfrac{\tilde{\LCd}_{\mathrm{D}}}{\left(\pfrac{\psi}{x^\mu}\right)},\qquad
\tilde{\pi}^\mu=\pfrac{\tilde{\LCd}_{\mathrm{D}}}{\left(\pfrac{\psibar}{x^\mu}\right)}.
\end{equation}
Due to the quadratic ``velocity'' dependence of~(\ref{ld-regular}), the corresponding
covariant Hamil\-tonian \cite{struckmeier08,StrRei12} is obtained via a \emph{regular} Legendre transformation
\begin{equation*}
\tilde{\HCd}_{\mathrm{D}}\left(\psi,\tilde{\pibar}^\alpha,\psibar,\tilde{\pi}^\alpha,\bgamma^\alpha\right)
=\tilde{\pibar}^\alpha\pfrac{\psi}{x^\alpha}+\pfrac{\psibar}{x^\alpha}\tilde{\pi}^\alpha-\tilde{\LCd}_{\mathrm{D}}
\left(\psi,\partial_\alpha\psi,\psibar,\partial_\alpha\psibar,\bgamma^\alpha\right)
\end{equation*}
as:
\begin{equation}\label{eq:hd-dirac}
\tilde{\HCd}_{\mathrm{D}}=\frac{\rmi M}{2}\left(\psibar\,\bgamma_\alpha\tilde{\pi}^\alpha
-\tilde{\pibar}^\alpha\frac{6\btau_{\alpha\beta}}{\sqrt{-g}}\,\tilde{\pi}^\beta
-\tilde{\pibar}^\alpha\bgamma_\alpha\psi\right)+\left(m-M\right)\psibar\psi\sqrt{-g},
\end{equation}
with $\btau_{\alpha\beta}$ the inverse of the matrix $\bsigma^{\beta\alpha}$, the latter defined in Eq.~(\ref{eq:sigma-def}),
\begin{equation*}
\btau_{\alpha\beta}=\frac{\rmi}{6}\left(\bgamma_{\alpha}\bgamma_{\beta}+3\bgamma_{\beta}\bgamma_{\alpha}\right),\qquad
\btau_{\nu\alpha}\,\bsigma^{\alpha\mu}=\delta_{\nu}^{\mu}\,\bEins.
\end{equation*}
The Hamiltonian representation of the identity~(\ref{ld-invariant2}) is then:
\begin{equation}\label{hd-invariant2}
\Tr\left\{\pfrac{\tilde{\HCd}_{\mathrm{D}}}{\bgamma\indices{_\mu}}\bgamma\indices{_\nu}\right\}
-\pfrac{\tilde{\HCd}_{\mathrm{D}}}{\tilde{\pi}^\nu}\tilde{\pi}^\mu
-\tilde{\pibar}^\mu\pfrac{\tilde{\HCd}_{\mathrm{D}}}{\tilde{\pibar}^\nu}
\equiv-\delta_\nu^\mu\,\tilde{\LCd}_{\mathrm{D}}.
\end{equation}
\section*{References}
\providecommand{\newblock}{}

\end{document}